\documentclass[12pt,fleqn]{article}
\usepackage{amssymb}
\usepackage{amsmath}
\usepackage{geometry}
\usepackage{setspace}
\usepackage{caption}
\usepackage{soul}

\usepackage[colorinlistoftodos]{todonotes}


\usepackage{amsthm}
\newtheorem{theorem}{Theorem}

\newtheorem{lemma}[theorem]{Lemma}

\theoremstyle{definition}
\newtheorem{definition}[theorem]{Definition}

\newtheorem{notation}[theorem]{Notation}

\newcommand{\suchthat}{\ifnum\currentgrouptype=16 \;\middle|\;\else\mid\fi}

\usepackage{graphicx}
\graphicspath{ {./Figures/} }
\usepackage{subcaption}

\usepackage{thm-restate}

\usepackage[
    backend=biber,
    style=bwl-FU,
  ]{biblatex}

\makeatletter
\def\and{%
  \end{tabular}%
  \hskip 0.7em \@plus.17fil\relax
  \begin{tabular}[t]{c}}
\makeatother

\begin{document}

\title{Mapping Firms' Locations in Technological Space: \\
A Topological Analysis of Patent Statistics\thanks{%
First version: August 31, 2019 (https://arxiv.org/abs/1909.00257v1). For helpful comments, we thank Susan Athey, Iain Cockburn, Marek Giebel, David Hsu, Adam Jaffe, and Yihan Yan, as well as participants at seminars and conferences, including Yale IO Seminar, the 2019 \textit{NBER Innovation Information Initiative} meeting, \textit{Joint Conference on Applied Mathematics 2019} by the Mathematical Society of Japan (MSJ), \textit{MSJ Spring Meeting 2020}, Kyoto University Applied Mathematics Seminar, \textit{TDA for Applications - Tutorial \& Workshop} at Tohoku University, the 2020 \textit{Econometric Society World Congress} at Bocconi University, the Hong Kong University of Science and Technology, \textit{TopoNets 2020}, the \textit{MaCCI/EPoS Conference on Innovation}, the 2021 \textit{International Industrial Organization Conference}, KU Leuven \textit{Data \& Algorithms for ST\&I Studies} conference, Indian Institute of Technology Bombay, and Instituto Tecnol\'ogico Aut\'onomo de M\'exico. We thank Alan Chiang and Chise Igami for research assistance.}}
\author{Emerson G. Escolar\thanks{%
Kobe University Graduate School of Human Development and Environment, and RIKEN Center for Advanced Intelligence Project. E-mail: e.g.escolar@people.kobe-u.ac.jp.}
\and Yasuaki Hiraoka\thanks{%
WPI-ASHBi, Kyoto University Institute for Advanced Study, Kyoto University and Center for Advanced Intelligence Project, RIKEN. E-mail: hiraoka.yasuaki.6z@kyoto-u.ac.jp.}
\and Mitsuru Igami\thanks{%
Yale Department of Economics. E-mail: mitsuru.igami@yale.edu.} \and Yasin
Ozcan\thanks{%
FTI Consulting. Email: ozcan@alum.mit.edu.}}
\date{March 31, 2022}
\maketitle

\begin{abstract}
Where do firms innovate? Mapping their locations and directions in technological space is challenging due to its high dimensionality. We propose a new method to characterize firms’ inventive activities via topological data analysis (TDA) that represents high-dimensional data in a shape graph. Applying this method to 333 major firms’ patents in 1976--2005 reveals substantial heterogeneity: some firms remain undifferentiated; others develop unique portfolios. Firms with unique trajectories, which we define and measure graph-theoretically as “flares” in the Mapper graph, perform better. This association is statistically and economically significant, and continues to hold after we control for portfolio size, firm survivorship, industry classification, and firm fixed effects. By contrast, existing techniques---such as principal component analysis (PCA) and Jaffe's (1989) clustering method---struggle to track these firm-level dynamics.

%\medskip


%\medskip

\noindent \textit{Journal of Economic Literature (JEL) classifications}: C65, C88, L10, O30.
\end{abstract}

%\linenumbers

%\clearpage
\section{Introduction}

The ``rate and direction of inventive activity'' have been recognized as one of the main themes in economics since at least the conference of the same title in 1960 (\cite{Nelson1962}, \cite{LernerStern2012}). Whereas the rate of innovation has been studied extensively, research on its direction has seen much less progress. Nevertheless, recent studies suggest the direction of scientific change is both an important choice for individual researchers and a critical outcome for scientific communities (\cite{AzoulayEtAl2019}, \cite{Myers2020}). These observations, along with the central role of product differentiation in the theory of industrial organization (IO), suggest the direction of inventive activity is important for firms and industries as well. 

Mapping the locations and directions of firms’ research and development (R\&D) activities is a challenging problem because technological space has many dimensions, unlike physical/geographical space.\footnote{Whereas a large literature exists on the geography of innovation (pioneered by \cite{JaffeTrajtenbergHenderson1993}), relatively few papers explore technological space, because of methodological challenges.} Even a relatively ``coarse'' classification system by the US Patent and Trademark Office (USPTO) uses more than 400 categories (patent classes), and large firms frequently conduct R\&D in more than 100 classes, obtaining thousands of patents each year. As a result, the dimensionality of the action/state space is extremely high, and infinitely many directions of inventive activity are possible in principle. Studying something we cannot even visualize and describe is difficult. Hence, developing a method for faithfully mapping their technological positions and documenting empirical regularities (i.e., measurement and exploratory data analysis) would be a crucial step.

Given the high dimensionality of the problem, some dimensionality reduction seems warranted. Commonly used methods include principal component analysis (PCA), multi-dimensional scaling (MDS), and various algorithms for clustering (e.g., k-means clustering). However, even though these existing methods provide some simplified visualization and description, fundamental issues remain unresolved: collapsing data would eliminate useful information about the direction of inventive activity. For example, Figure \ref{Figure - mapper(n20_m0_cos)} (a) shows a PCA that projects onto a two-dimensional plane 333 major firms’ patent portfolios (vectors of logged patent counts across 430 USPTO classes) in 1976--2005. Huge clusters of points on the left side would seem to suggest many firms conduct R\&D in close proximity, but this “densely populated area” could partly be an artifact of collapsing the other 428 dimensions. Similar issues arise in other existing methods, due to information loss (see section 4.4 for an example of clustering). Thus, a faithful representation of the positions and directions of R\&D requires new descriptive tools that avoid arbitrarily collapsing data, provide intuitive visualizations of how firms’ patent portfolios evolve over time, and permit quantification of these dynamics.

\begin{figure}[htb!!!!]
\caption{Firms' Locations in Technological Space, 1976--2005}%

\begin{subfigure}{0.5\textwidth}
\caption{Two-Dimensional PCA}%
\centering
\includegraphics[width=0.9\linewidth]{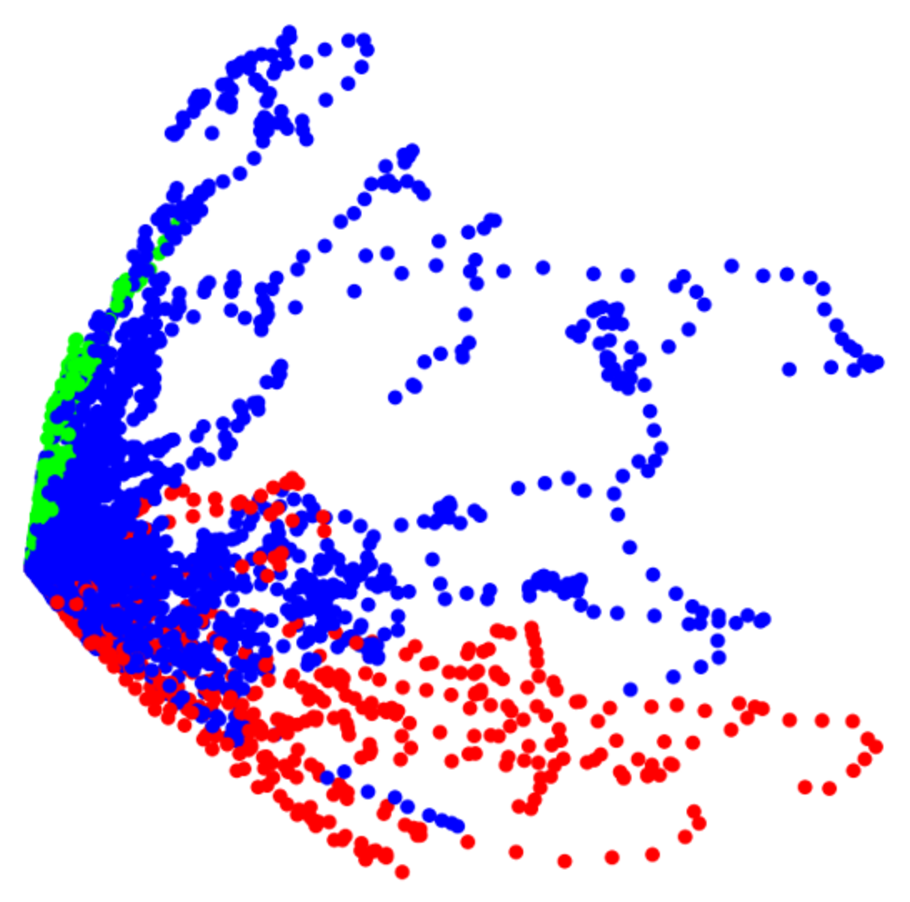}
\end{subfigure}
\begin{subfigure}{0.5\textwidth}
\caption{Shape Graph by Mapper}%
\centering
\includegraphics[width=0.9\linewidth]{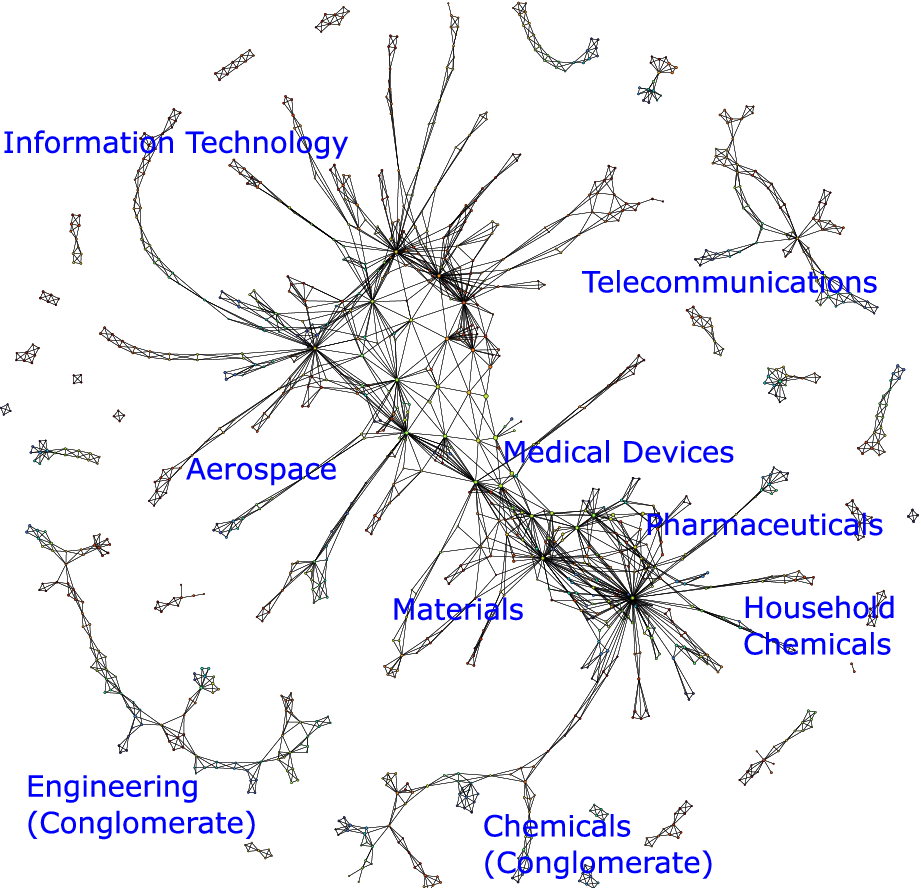}
\end{subfigure}
\caption*{\footnotesize {%
\textit{Note}: Both pictures represent the evolution of 333 major firms' portfolios of US patents that are acquired by in-house R\&D between 1976 and 2005. Each firm-year is a vector of log patent counts across 430 technological classes. The left panel is a two-dimensional PCA (red markers are IT firms, green markers are drug makers, and blue markers are all others). The right panel is a Mapper graph based on the same data (see section 4 for details).  See Appendix H for a three-dimensional PCA and the same Mapper graph in a PCA layout.}}%
\label{Figure - mapper(n20_m0_cos)}
\end{figure}%

This paper presents such a new method to represent firms’ locations as a combinatorial/topological object (shape graph), which can be easily visualized and quantified in a variety of ways using graph theory. We adapt and extend a tool from computational topology called the Mapper procedure \parencite{singh2007topological}. This algorithm is well founded on mathematical concepts from computational topology and geometry, such as the Reeb graph, and aims to preserve the topological and geometric information of the original data, in two steps. First, it clusters data points in each local neighborhood based on a distance metric of one’s choice (e.g., cosine distance). Second, it connects clusters with edges if a pair of clusters shares at least one data point. Hence, even though the resulting graph might appear to visualize data on a two-dimensional plane--—see Figure \ref{Figure - mapper(n20_m0_cos)} (b)---as in the PCA plot, the shape graph retains the notions of proximity and continuity (in the original space) with edges between neighboring nodes.

We apply this method to the dynamic evolution of the 333 major firms’ patent portfolios across 430 USPTO classes in 1976--2005, and report three sets of results. First, we visualize these firms' technological positions and trajectories over the three decades. (Whereas ``data visualization'' plays only a minor role in most empirical studies, it embodies one of the main results in our context, because the systematic mapping of technological space is the central empirical problem that this paper addresses.) We find many engineering firms remain undifferentiated and cluster together in the densely populated “trunk” or the “continental” part of the map. However, a few dozen firms, primarily in the information technology (IT) sector, start differentiating from the rest in the 1980s and the 1990s, developing unique portfolios and exhibiting distinctive trajectories, as represented by long “branches” or “flares” that spike out of the main trunk. In the topological space, which is coordinate free, these shapes provide explicit signatures of the unique ``directions'' of inventive activity.

Second, we propose a formal definition of such flares based on graph theory, as well as a computational method to measure their length, and find 40.3 \% of the firms exhibit some flares. We assess the empirical relevance of this new measure by evaluating its statistical relationships with the firms’ financial performances (revenue, profit, and market value). Regression results suggest positive correlations between the flare length and the performance metrics. This association is statistically significant at conventional levels, and economically significant in magnitude (e.g., an extra length of flare in 1976--2005 is associated with 31\%--40\% higher performances as of 2005). Moreover, these patterns continue to hold after controlling for (i) portfolio size, (ii) firm survivorship, (iii) industry classification, and (iv) firm fixed effects.

Third, we show how our method and results compare with \cite{JAFFE198987}, which is based on k-means clustering and is one of the most prominent methods to study firms’ technological locations. The scope of Jaffe’s clustering is global, which makes it suitable for splitting firms into industries. But \cite{JAFFE198987} struggles to track firm-level trajectories and fails to find any statistically significant relationship between their moves and performances. By contrast, our scope of clustering is only local, which allows us to preserve details at the firm-year level. Moreover, the whole procedure is designed to retain and recover the \textit{continuum} of firms and industries in the original data, and allows us to characterize firm-level trajectories. Our discovery of statistically significant relationships between the firms' financial performances and their length of unique technological trajectories (flares) demonstrates the benefit of this high-precision approach.

Thus, our approach is complementary to the existing methods and can generate new insights that are difficult to obtain otherwise. It helps us answer some of the most basic questions, including where firms innovate, how their technological trajectories are related to their product-market performances, and how industries and technologies evolve over time.

\medskip

We organize the rest of the paper as follows. Section 2 presents a model of competition and innovation in a high-dimensional space. Section 3 explains the data. Section 4 introduces our topological method and presents a historical map of firms' inventive activities. Section 5 explains our method to measure flare length and assesses its correlation with firms' performances. Section 6 concludes. The Online Appendix contains (A) the details of our economic model, (B) raw-data patterns, (C) an introduction to TDA, and formal definitions and proofs, (D) the details of Jaffe-style clustering, (E) sensitivity analysis, (F) panel-data regressions and out-of-sample predictions, (G) comparison with network-centrality measures and Jaffe's distance measure, and (H) additional exhibits.

%%% Local Variables:
%%% mode: latex
%%% TeX-master: "TDAPatent"
%%% End:

\section{Conceptual Framework}

We propose an economic model of firms' competition and innovation to (i) highlight key economic forces that
affect firms' behaviors and market outcomes, (ii) guide our exploratory data analysis, and (iii) facilitate the interpretation of our
empirical findings.

\subsection{Competition and Innovation in High-Dimensional
Space}

We combine elements of the workhorse IO models of \cite{BLP1995} (BLP) and \cite{EP1995} (EP) in the presence of many product markets that are embedded
in the space of technologies.

\paragraph{Markets and Technologies.}

Consider many product markets indexed by $m=1,2,...,\left \vert \mathcal{M%
}\right \vert $, each of which is populated by $M_{m,t}$ consumers and $%
N_{m,t} $ firms in period $t$. They are independent of each other.
Their main difference from geographical markets---whose physical locations
can be characterized by only two numbers, longitude and latitude---is that
we characterize their \textquotedblleft locations\textquotedblright \ from
the viewpoint of technologies that are required to serve them. Let $l\left(
m\right) \equiv \left( l_{1}\left( m\right) ,l_{2}\left( m\right)
,...,l_{K}\left( m\right) \right) $ denote the location of market $m$ in the 
$K$-dimensional space, where $l_{k}\left( m\right) \geq 0$ is its $k$th
coordinate.\footnote{%
We abstract from the distinction between product space and technology space
because we use only patent statistics and financial data in
our empirical analysis. See \cite{BloomVanReenenSchankerman2013} for an
example that makes this distinction.}

\paragraph{Period Profit.}

Each of the $N_{m,t}$ firms earns period profit,%
\begin{equation}
\pi _{i,t}=f^{\pi}\left( M_{m,t},N_{m,t},\xi _{i,t},c_{i,t}\right) ,
\label{eq - period profit}
\end{equation}%
where $\xi _{i,t}$ is product quality (we assume single-product firms) and $%
c_{i,t}$ is constant marginal cost of production. This reduced-form profit function
encapsulates a BLP-style model of a differentiated-product demand system and
Bertrand competition (see Appendix A.1). Hence, $\pi _{i,t}$ is
increasing in $M_{m,t}$ and $\xi _{i,t}$ but decreasing in $N_{m,t}$ and $%
c_{i,t}$. These four objects are determined by the history of (all) firms'
actions, $h_{t}\equiv \left( h_{i,t}\right) _{i=1}^{N_{t}}$, where $%
h_{i,t}\equiv \left( a_{i,\tau }\right) _{\tau =0}^{t-1}$ is firm $i$'s
actions up to period $t-1$, and $N_{t}$ denotes the total number of firms
that have operated in any of the $\left \vert \mathcal{M}\right \vert $
markets in any period up to $t$.

\paragraph{Market Size.}

Each market $m$'s size is realized at $t=0$ following some distribution $%
F_{M}$ with spatial correlations, $M_{m,0}\sim F_{M}$. In any subsequent
period $t>0$, its effective size is the portion of consumers that have not
purchased anything yet,%
\begin{equation}
M_{m,t}=M_{m,t-1}-\sum_{h\in \mathcal{H}_{m,t-1}}\mathbb{I}\left \{
d_{h,t-1}\neq 0\right \} ,  \label{eq - market size, transition}
\end{equation}%
where $\mathcal{H}_{m.t-1}$ is the set of remaining consumers in market $m$
at $t-1$, $\mathbb{I}\left \{ \cdot \right \} $\ is an indicator function, $%
d_{h,t-1}$ is the discrete choice of consumer $h$ at $t-1$, and $%
d_{h,t-1}\neq 0$ means the consumer bought something.

\paragraph{Number of Firms.}

The number of active firms in market $m$ at time $t$ is the sum of firms whose
technological locations $l_{i,t}\equiv \left(
l_{i,t,1},l_{i,t,2},...,l_{i,t,K}\right) $\ are in the neighborhood of $%
l\left( m\right) $:%
\begin{equation}
N_{m,t}=\sum_{i}\mathbb{I}\left \{ l_{i,t}\in \mathcal{N}\left( l\left(
m\right) \right) \right \} ,  \label{eq - number of firms in market m}
\end{equation}%
where $\mathcal{N}\left( \cdot \right) $ is the set of neighborhood
locations (specified in section 4). Thus, firms can serve market $m$ only when they possess \textquotedblleft
relevant\textquotedblright \ technologies $l_{i,t}\in \mathcal{N}\left( l\left(
m\right) \right)$.

\paragraph{R\&D Investments.}

Each firm's location is determined by $l_{i,t}=f^{l}\left( x_{i,t}\right) $,
where $f^{l}$ is an increasing function (specified in section 4) and $x_{i,t}\equiv \left(
x_{i,t,1},x_{i,t,2},...,x_{i,t,K}\right) $ is the amount of successful R\&D
investment in each of the $K$ technological areas at time $t$. Not all R\&D investments are
successful, and firms could be heterogeneous in their R\&D productivity. We
encapsulate these notions in a stochastic R\&D-production function, 
\begin{equation}
x_{i,t,k}=f^{x}\left( b_{i,t-1,k}^{x};\omega _{i,t-1,k}^{x}\right)
+\varepsilon _{i,t,k}^{x},  \label{eq - law of motion for x}
\end{equation}%
where $f^{x}$ is an increasing function of $b_{i,t-1,k}^{x}$ ($i$'s R\&D budget in the previous period in area $k$), $\omega _{i,t-1,k}^{x}$\ is its area-specific R\&D
productivity that follows some exogenous Markov process, and $\varepsilon
_{i,t,k}^{x}$ is an i.i.d. shock. Let $b_{i,t}^{x}\equiv
\sum_{k=1}^{K}b_{i,t,k}^{x}$ denote the total R\&D\ expenditure across all
areas, and $\mathbf{\omega}_{i,t}^{x}\equiv \left( \omega_{i,t,k}^{x}\right) _{k=1}^{K}$ the vector of area-specific R\&D productivity.

\paragraph{Other Investments.}

Firms can engage in two other categories of investments---marketing
and operations---which determine the firm's product quality $\xi _{i,t}$
and production cost $c_{i,t}$, respectively. These state variables evolve according
to some controlled\ Markov processes, $\xi _{i,t}=f^{\xi }\left( \xi
_{i,t-1},b_{i,t-1}^{\xi };\omega _{i,t-1}^{\xi }\right) $ and $%
c_{i,t}=f^{c}\left( c_{i,t-1},b_{i,t-1}^{c};\omega _{i,t-1}^{c}\right) $,
where $b_{i,t}^{\xi }$ and $b_{i,t}^{c}$ are $i$'s budgets for
marketing and operations, respectively, and $\omega _{i,t}^{\xi }$ and $%
\omega _{i,t}^{c}$ are $i$'s productivity in these activities, which follow some exogenous Markov processes as well.

\paragraph{Budget.}

The firm's total budget is constrained by the amount of available cash,%
\begin{equation}
b_{i,t}\equiv b_{i,t}^{x}+b_{i,t}^{\xi }+b_{i,t}^{c}\leq cash_{i,t},
\label{eq - budget constraint}
\end{equation}%
which is determined by the following accounting rule,%
\begin{equation}
cash_{i,t}=cash_{i,t-1}-b_{i,t-1}+\pi _{i,t-1}+fin_{i,t-1},
\label{eq - cash law of motion}
\end{equation}%
where the first three terms on the right-hand side (RHS) reflect cash holding, expenditure, and
profits in the previous period, respectively, and $fin_{i,t-1}\gtrless 0$ is
the cashflow from financing activities.\footnote{%
We assume $fin_{i,t}$ follows some exogenous Markov process and do
not model the underlying financial markets. We include it to incorporate the possibility that retained earnings are not the only
source of cash and that a firm can go bankrupt (see Appendix A.2 for entry and exit).}

\paragraph{Dynamic Optimization.}

Each firm allocates its budget to R\&D $\mathbf{b}%
_{i,t}^{x}\equiv \left( b_{i,t,k}^{x}\right) _{k=1}^{K}$, marketing $b_{i,t}^{\xi }$,
and operations $b_{i,t}^{c}$, to maximize the discounted present value of its current
and future profits,%
\begin{equation}
a_{i,t}\equiv \left( \mathbf{b}_{i,t}^{x},b_{i,t}^{\xi },b_{i,t}^{c}\right)
=\arg \max \sum_{\tau =t}^{\infty }\beta _{i}^{\tau -t}E_{i,t}\left[ \pi
_{i,\tau }\right] ,  \label{eq - value maximization}
\end{equation}%
subject to the budget constraint (\ref{eq - budget constraint}). $\beta
_{i}\in \left( 0,1\right) $ is $i$'s discount factor. $E_{i,t}$ is
the expectation operator given its information set and beliefs at $t$. We do not fully specify these objects because computing equilibria of this dynamic game is outside the scope of this paper, but we intend our framework as a model of the EP class (i.e., strategic industry dynamics with Markov-perfect equilibrium).

\subsection{Implications for the Analysis of Technological Space}

Five features of the model are particularly relevant for the analysis of firms' technologies:

\begin{enumerate}
\item {Profit }$\pi _{i,t}$ is increasing in $M_{m,t}$ and $\xi _{i,t}$ but
decreasing in $N_{m,t}$ and $c_{i,t}$;

\item These four objects are determined by the history $h_{t}$ of (all)
firms' actions $a_{i,t}$;

\item Firms are heterogeneous in their productivity, $\omega _{i,t}\equiv
\left( \mathbf{\omega}_{i,t}^{x},\omega _{i,t}^{\xi },\omega _{i,t}^{c}\right) $;

\item The size $M_{m,t}$ of each market is finite and could only decrease over time; and

\item Current profit $\pi_{i,t}$ could increase future R\&D budget $\mathbf{b}%
_{i,t+1}^{x}$ via (\ref{eq - budget constraint}) and (\ref{eq - cash law of motion}).
\end{enumerate}

A direct implication of Features 1 and 2 is that firms would try to operate
in markets with high $M_{m,t}$ and low $N_{m,t}$. Thus, the realized profile
of locations, $l_{t}\equiv \left( l_{i,t}\right) _{i}$ will reflect firms' tradeoff between \textquotedblleft chasing consumers\textquotedblright \ and \textquotedblleft avoiding
competitors.\textquotedblright \ Feature 3 suggests firms with comparative advantage in R\&D (i.e., relatively high $\mathbf{\omega}_{i,t}$)
would move away from crowded markets and try to carve out their own niches. The
high dimensionality $K$ of the technological space, combined with firms'
heterogeneous R\&D capabilities across $K$ areas, offers ample room for
such differentiation. Feature 4 limits the extent to which firms can \textquotedblleft
rest on their laurels\textquotedblright \ (i.e., remain profitable in the
same locations). Because potential demand in any given market is like an oil reserve
that becomes increasingly difficult to extract, firms have to either
constantly explore and conquer new markets or
keep investing in $\xi _{i,t}$ and $c_{i,t}$ to dig deeper. Finally, Feature 5 highlights the possibility of a virtuous cycle in which ``the rich gets richer.'' That is, those who succeed in developing unique technologies earn extra profits, which can be reinvested in future innovations to pursue further growth opportunities.

These considerations suggest the \emph{locations of firms relative to each
other} $\left \{ l_{i,t}\right \} $ could exhibit rich
variation and contain relevant information about their performances and underlying capabilities. In particular, a string of unique positions occupied by a firm may be indicative of its long track record of successful innovations and sustained profitability. We present our method for describing $\left \{ l_{i,t}\right \} $ in section 4, and formalize the measurement of firms' unique technological trajectories in section 5.
\section{Data}

\paragraph{\protect Patents.} We use Ozcan's (2015) data on patents that are granted by the USPTO between 1976 and 2010.\footnote{\cite{Ozcan2015} uses the USPTO's Patent Data Files, which contain raw assignee names at the individual patent level. By contrast, the NBER Patent Data File (another commonly used source of patent data) records standardized assignee names at the ``pdpass'' (unique firm identifier) level, which is less granular than the original assignee name.} We use their application years (instead of years in which they are granted) in our analysis, because the former is closer than the latter to the time of actual invention. We focus on patents that are applied through 2005, because a substantial fraction of later applications would still be under review as of 2010, which raises concerns about sample selection. We sometimes call these patents ``R\&D patents'' to distinguish them from ``M\&A patents'' (see below).

\paragraph{\protect Mergers and Acquisitions (M\&As).} Aside from conducting in-house R\&D and applying for patent protection, firms often obtain patents by acquiring firms that have their own portfolios of patents. Ozcan's (2015) dataset links the USPTO data to the Securities Data Company's M\&A data module. This part of the dataset contains M\&A deals between 1979 and 2010 in which both the acquiring firm and the target firm have at least one patent between 1976 and 2010.\footnote{The data include merger, acquisition, acquisition of majority interest, acquisition of assets, and acquisition of certain assets, but exclude incomplete deals, rumors, and repurchases. We use data on these transactions through 2005.}

\paragraph{\protect Financial Performances.} We use Compustat data on the firms' revenues, EBIT (earnings before interest and taxes), and stock-market capitalization in 2005 (or the last available fiscal year if the firm disappears before 2005). Our purpose is to assess the relevance of our topological measures in terms of their correlations with the firms' eventual financial performances (in section 5).

\paragraph{\protect Descriptive Statistics.} To keep the sample size suitable for visual inspection and detailed exploratory analysis, we focus on firms that acquired at least four firms with patents between 1976 and 2005. This criterion keeps 333 major firms that conduct nontrivial amount of both R\&D and M\&A. Table \ref{Table - Sumstats} reports their descriptive statistics. The average patent count (2,081 for R\&D and 268 for M\&A) is much higher than the median, which suggests relatively few firms have disproportionately large portfolios even within our selective sample. The three financial-performance metrics exhibit similar skewness. Consequently, we use the natural logarithm of these variables to mitigate heteroskedasticity in our subsequent analysis.

\begin{table}[tbh!!!!]
\caption{Summary Statistics of 333 Major Firms}
\begin{center}
\fontsize{9pt}{11pt}\selectfont%
\begin{tabular}{lcccccc}
\hline \hline
& Number of & Mean & Median & Standard & Minimum & Maximum \\ 
Variables & observations &  &  & deviation &  &  \\ \hline
\multicolumn{1}{l}{(a) Patent count} & & & & & & \\ 
\multicolumn{1}{l}{ \ \ In-house R\&D} & $333$ & $2,081$ & $270$ & $5,578
$ & $1$ & $62,382$ \\ 
\multicolumn{1}{l}{ \ \ Acquired by M\&A} & $333$ & $268$ & $59$ & $883$
& $4$ & $9,453$ \\ 
\multicolumn{1}{l}{ \ \ Both R\&D and M\&A} & $333$ & $2,349$ & $405$ & $5,833$ & $5$ & $62,561$ \\ 
\multicolumn{1}{l}{(b) Financial performance} & & & & & & \\ 
\multicolumn{1}{l}{ \ \ Revenue (million US\$)} & $331$ & $10,641$ & $2,306$ & $25,137$ & $15$ & $309,979$ \\ 
\multicolumn{1}{l}{ \ \ EBIT (million US\$)} & $331$ & $1,429$ & $250$ & $3,763$ & $-450$ & $37,159$ \\ 
\multicolumn{1}{l}{ \ \ Market value (million US\$)} & $328$ & $17,957$ & $3,471$ & $39,153$ & $12$ & $367,474$ \\ 
\multicolumn{1}{l}{(c) Number of classes with $>0$ patents} & & & & & & \\ 
\multicolumn{1}{l}{ \ \ In-house R\&D} & $330$ & $65.0$ & $34.5$ & $71.6$ & $1$ & $358$ \\ 
\multicolumn{1}{l}{ \ \ Acquired by M\&A} & $326$ & $22.5$ & $12.5$ & $30.4$ & $1$ & $225$ \\ 
\multicolumn{1}{l}{ \ \ Both R\&D and M\&A} & $333$ & $72.4$ & $43.0$ & $72.3$ & $2$ & $358$ \\ \hline \hline
\end{tabular}
\begin{minipage}{450pt}
{\fontsize{9pt}{9pt}\selectfont \smallskip  \textit{Note}: Financial-performance metrics are as of 2005 or the firm's last available fiscal year. Panels (b) and (c) display fewer observations than the sample size, because some firms are not in Compustat and some patents' classes are unknown.}
\end{minipage}
\end{center}
\label{Table - Sumstats}
\end{table}

\paragraph{\protect Where Do Firms Patent?} Panel (c) of Table \ref{Table - Sumstats} counts the number of USPTO classes in which the firms have patents. The median firm conducts R\&D in 34.5 classes, whereas the mean is 65. The most diversified portfolio (Mitsubishi Electric) covers 358 of the 430 classes, followed by General Electric's 347. Hence, the portfolio aspect of innovation is highly heterogeneous. Appendix B illustrates what these portfolios look like in raw data.

%%% Local Variables:
%%% mode: latex
%%% TeX-master: "TDAPatent"
%%% End:

\section{Mapping Firms' Locations Over Time}

We explain our method to study firms' locations in
technological space in sections 4.1 and 4.2, and investigate its output---a shape
graph---in section 4.3. Section 4.4 compares Mapper with Jaffe's (1989) clustering method.

\subsection{The Mapper Algorithm}

We propose patents as a measure of successful R\&D\ investment. For each
firm $i=1,2,...,333$, each year $t=1976,1977,...,2005$, and each patent
class $c=1,2,...,430$, we count the number of patent applications, $%
p_{i,t,c} $. Hence, each firm-year observation is a 430-dimensional vector $%
p_{i,t}\in \mathbb{R}^{430}$ (i.e., we use patent class $c$ as an empirical analog of
technological area $k$ in our theoretical model and assume $K=430$).

\paragraph{Preprocessing.}

Because firms' patent applications in any single year tend to be volatile
and may not be representative of their underlying R\&D activities, we follow
\cite{BennerWaldfogel2008} to smooth out yearly
fluctuations by aggregating them in a five-year moving window: $\tilde{p}%
_{i,t}=\sum_{\tau =t}^{t+4}p_{i,\tau }$. We take its natural logarithm to
accommodate the highly skewed distribution of patent count (see section 3),%
\footnote{%
This equation is our main specification of $f^{l}\left( \cdot \right) $ in
section 2. We also use an alternative transformation (calculating shares of
classes within each firm-year) due to Jaffe (1989) in Appendix D.}%
\begin{equation}
l_{i,t}=\ln \left( \tilde{p}_{i,t}+1\right) .  \label{eq - log of patents}
\end{equation}%
Let $L=\left \{ l_{i,t}\right \} $ denote the entire panel dataset of firms'
locations.

We propose mapping the entire $L$ in a single graph, instead of creating a map
for each $i$ or $t$ (see Appendix H for such plots), for two reasons. First, our model in section 2 suggests
firms' locations \textit{relative to each other} determine the number of
competitors $N_{m,t}$ in each market $m$, which in turn affects profits.
Second, the model also suggests their historical trajectories contain
relevant information about firms' R\&D capabilities and profitability:
dynamics matter. Fortunately, our topological method works well with such a
dataset (i.e., many data points, or a \textquotedblleft point
cloud,\textquotedblright \ with many dimensions).

\paragraph{Mapper.}

We first present the Mapper procedure in purely mathematical terms, and then
provide more intuitive explanations. The procedure
creates a simplified representation of complicated data in a graph
(\textquotedblleft shape graph\textquotedblright \ or \textquotedblleft
Mapper graph\textquotedblright ) that captures topological features such as
branching, flares, and islands. Mathematically, this shape graph $G\left(
L\right) $ is constructed in four steps.

\begin{enumerate}
\item Project $L$ into $\mathbb{R}^{d}$ by some filter function $%
f:L\rightarrow \mathbb{R}^{d}$, where $d<K$ is the dimensionality of a
lower-dimensional space.

\item Cover the image $f(L)$ using an overlapping cover $\mathcal{C}%
=\{C_{j}\}_{j=1}^{J}$.

\item For each cover element $C_{j}$, apply some clustering algorithm to its
pre-image $f^{-1}(C_{j})$ based on the dissimilarity function $\delta $ to
obtain a partition of $f^{-1}(C_{j})$ into $Q_{j}$ clusters, $V_{j,q}$ ($q=1,%
\hdots,Q_{j}$):
\begin{equation*}
f^{-1}(C_{j})=\bigsqcup_{q=1}^{Q_{j}}V_{j,q},
\end{equation*}
where the notation $\sqcup $ represents a disjoint union.

\item Construct the graph $G$ with nodes (vertices) consisting of all $%
V_{j,q}$s. Connect two nodes, $V_{j,q}$ and $V_{j^{\prime },q^{\prime }}$,
by an edge if $V_{j,q}\cap V_{j^{\prime },q^{\prime }}\neq \emptyset $.
\end{enumerate}

Conceptually, the idea is to simplify the raw data $L$ by clustering data
points within each local region (in steps 1, 2, and 3, which define a set $V$
of vertices or nodes) but make sure to preserve the sense of continuity
across regions (in step 4, which defines a set $E$ of edges), so that the
resulting graph $G=\left( V,E\right) $ retains the topology of the data on a
global scale. Appendix C.1 offers a brief introduction to TDA. Appendix C.2 features an illustrated example (with $K=2$, $d=1$%
, and $J=4$) to help the reader develop a more concrete understanding.

\paragraph{Connections to the Economic Model.}

The graph $G(L)$ provides a topological map of firms' technological locations $L$. The set
of nodes $V$ is an empirical analog of the set of product markets $%
\mathcal{M}$ that have ever been visited by any of the firms in our data.
Hence, the local clustering in step 3 empirically determines the neighborhood $\mathcal{N}$ in equation (\ref{eq - number of firms in market m}). The set of edges $E$ preserves their relative positions by indicating for each market which other markets are adjacent to it.

\subsection{Practical Considerations}

The Mapper procedure offers a
\textquotedblleft telescope\textquotedblright \ to directly look at data
points---even when they reside in a high-dimensional space---by
focusing on a coordinate-free representation of the underlying data in terms
of a graph. This graph preserves the relative positions of the original data
points as long as they form a continuum. Hence, it is suitable for
visualizing any high-dimensional data points that exhibit some sort of
continuity.

As is the case with a real telescope, its practical usefulness depends on
properly tuning its \textquotedblleft parameters\textquotedblright : (i) the
filter function $f$, (ii) the number of cover elements $J$, (iii) the
clustering method, (iv) the dissimilarity function $\delta $, and (v) the
degree of overlap $o$ between cover elements. We explain the role of each parameter
and our baseline specification.

\paragraph{Filter.}

The choice of $f$ in step 1 determines the \textquotedblleft
angle\textquotedblright \ at which we look at the data. Some
angles allow us to see richer patterns than others because they expose
greater variation. A typical choice is PCA or MDS, but any other
\textquotedblleft off-the-shelf\textquotedblright \ technique for
dimensionality reduction can be used in principle. We use two-dimensional PCA as
our baseline $f$ (i.e., we project $L$ to its first two principal axes, $%
f:L\rightarrow \mathbb{R}^{2}$) because PCA is fast, deterministic, and
well-understood, and preserves the largest variation in data by definition.
As a sensitivity analysis, we also use MDS and three-dimensional PCA in section 5.4.

\paragraph{Resolution.}

In step 2, $J$ determines the resolution of the graph. The higher the
resolution, the more details are revealed. But a fundamental limit exists.
An arbitrarily high $J$ would result in a degenerate graph with as many
nodes as data points but no edges. Because data points are discrete objects,
we cannot preserve the sense of continuity between them if our
scope is narrower than the distance between them. We set $J=400$ because it
reveals sufficiently detailed patterns at the individual-firm level without
losing their historical trajectories. We assess sensitivity with $225$ and $%
625$ as well.\footnote{%
We use the Python implementation, KeplerMapper, by \cite{KeplerMapper2019}, in which this parameter is operationalized as the \textquotedblleft
number of cubes,\textquotedblright \ $n$, in each of the $d$ dimensions
(e.g., $J=n^{2}$ when $d=2$). Thus, we implement $J=225$, $400$, and $625$
by setting $n=15$, $20$, and $25$, respectively.}

\paragraph{Clustering.}

Step 3 performs the main simplification task: clustering nearby data points.
Conceptually, the most important point of Mapper is not the choice of
clustering algorithm but the idea that \textit{this operation is performed only on a
specific subset of data points} (i.e., those within each $f^{-1}\left( C_{j}\right) 
$) \textit{at a time}. Hence, any
reasonable clustering method may be used. We use hierarchical clustering
with single-linkage method, and follow Sing, M\'{e}moli, and Carlsson's (2007) heuristic for choosing
the number of clusters. We assess sensitivity with five other specifications.

\paragraph{Dissimilarity.}

Clustering requires a measure of (dis)similarity between a given pair of
firm-year observations, say $\left( i,t\right) $ and $\left( i^{\prime
},t^{\prime }\right) $. We use the cosine distance,%
\begin{equation*}
\delta (l_{i,t},l_{i^{\prime },t^{\prime }})=1-\frac{\sum_{c}l_{i,t,c}l_{i^{%
\prime },t^{\prime },c}}{\sqrt{\sum_{c}l_{i,t,c}^{2}}\sqrt{%
\sum_{c}l_{i^{\prime },t^{\prime },c}^{2}}},
\end{equation*}%
because it has been commonly used since
\cite{Jaffe1986}. We also use Euclidean,
correlation, min-complement (\cite{BarLeiponen2012}), and Mahalanobis distances.

\paragraph{Overlap.}

Step 4 completes the graph representation by adding an edge to any pair of
clusters (nodes) that share at least one observation. This
\textquotedblleft sharing\textquotedblright \ of observations requires an overlapping region between adjacent cover elements.
The degree of overlap $o\in \left( 0,1\right) $ governs the tolerance for detecting
continuity, with values close to $0$ generating almost no edges and values close to $1$
detecting continuity almost everywhere. Such extreme values defeat
the purpose of capturing the shape of the data; we set $o=0.5$ (i.e., 50\%
of a cover element's \textquotedblleft area\textquotedblright \ overlaps
with each of its neighbors), and assess sensitivity with $0.3$ and $0.7$.

\subsection{A Topological Map of the Technological Space, 1976--2005}

The shape graph of Figure \ref{Figure - mapper(n20_m0_cos)} (b) embodies our first main result: a faithful representation of the 333 firms’ inventive activities across 430 technological areas. Pooling all 30 years of panel data allows us to track their movements within a single map, including many unique trajectories. Appendix H reports alternative results based on year-by-year Mapper graphs.

\paragraph{IT.} 

Figure \ref{Figure - mapper(cos_log_details_1)} reproduces the northern half of Figure 1 (b) with greater detail. The main trunk consists of large nodes containing hundreds of firm-years (see the lower-middle part labeled “many engineering firms”). Their patents are relatively few and undifferentiated. Even famous IT firms started from this densely populated “heartland” of electronics in the 1970s, but their inventive activities diverged from the rest in the 1980s and evolved into unique trajectories in the 1990s and the 2000s. These dynamics coincide with the macroeconomic trend in which IT emerged as a dominant sector with new technological opportunities in many directions. To demonstrate the authenticity of our map more concretely, we investigate five historically important cases.

\begin{figure}[htb!!!!]
\caption{IT and Electronics}%
\centering
\includegraphics[width=1\linewidth]{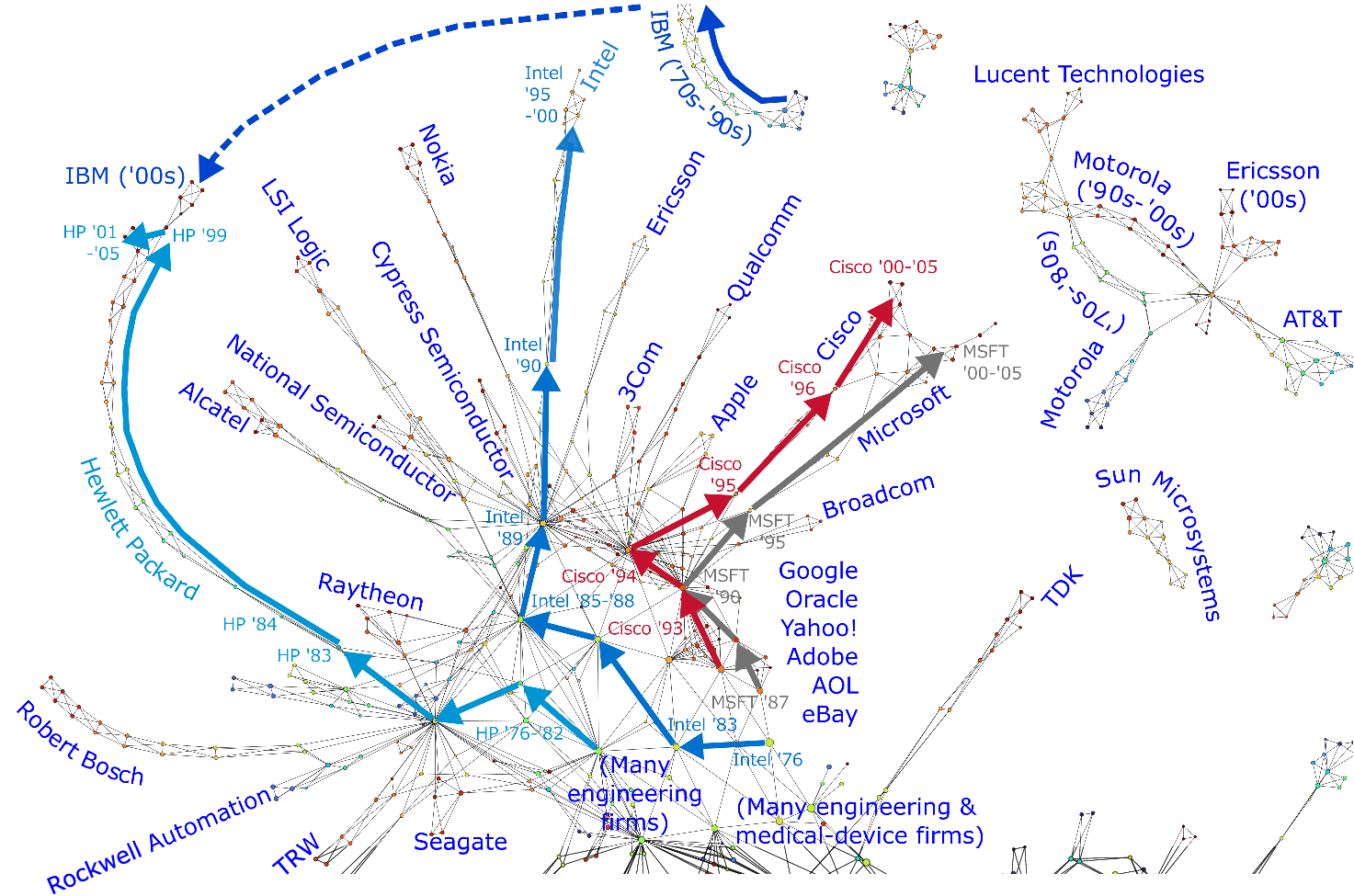}
\caption*{\footnotesize {%
\textit{Note}: Arrows indicate the directions of moves of the five IT firms (see main text). Node colors represent the average year of the firm-years in that cluster, with earlier years in blue and later years in red.}}%
\label{Figure - mapper(cos_log_details_1)}
\end{figure}%

First, the patenting activities of Intel---a leading chip maker---used to be indistinguishable from the rest. Between 1976 and 1988, it moved around but was always surrounded by many other firms. In 1989--1990, however, it started marching in a new direction, and established a clearly unique track record by 1995. This timing coincides with Intel’s “near-death experience” in the mid 1980s, in which Japanese rivals squeezed it out of the memory market, and its subsequent shift to microprocessors (see \cite{grove1996}). During the 1990s, it invested heavily in new microprocessor designs and became a household name (“intel inside”) as personal computers (PCs) became popular. Our map successfully captures these developments as an outward flare, because the underlying patent data distinguishes between “memory” (class 711) and “processors” (712), and Mapper handles all of the 430 dimensions equally well, including the ones for classes 711 and 712.

Second, HP is recognized as the symbolic founder of Silicon Valley because it produced the world’s first PC in 1968.\footnote{``The First PC'' (https://www.wired.com/2000/12/the-first-pc/). \textit{Wired}. December 1, 2000.} In 1984, HP introduced inkjet and laser printers for desktop computers, and retained focus on computers and printers through the 1990s, while its older business in test and measurement instruments was spun off into Agilent Technologies in 1999. Figure \ref{Figure - mapper(cos_log_details_1)} summarizes this history well. HP operated in the middle of the electronics heartland in 1976–1980 alongside many other device makers and defense firms. But its unique direction became clearly visible by 1984, as it started breaking new grounds with patents in class 347 (incremental printing of symbolic information). This path continued to grow into one of the longest flares in our graph. HP briefly “touched” IBM in 1999 (see below), before the Agilent deal made HP unique again.

Third, IBM generated more US patents than any other businesses. Its patenting activities are “off the chart” in both scale and scope, which our map visualizes as an ``island'' detached from all other firms. Nevertheless, IBM in 2001--2005 was sufficiently similar to HP in 1999--2003, and the two firms were briefly collocated near the end of HP’s flare. This rendezvous is not a coincidence: IBM went through major restructuring in 1993–2002 (see \cite{gerstner2002}). Thus, this collocation reflects IBM's downsizing as well as HP's growth.

Fourth, Cisco became a poster child of the Internet age, as the world adopted the Internet Protocol (IP) in the mid-to-late 1990s. Founded in 1984, Cisco makes networking hardware and software. Its first patent was filed in as late as 1993. But its focus on classes 370 (multiplex communications) and 709 (multicomputer data transferring), which together account for 60\% of its patents in our data, was so unique that its trajectory quickly evolved into a flare in the mid 1990s. Thus, a firm does not have to be patenting a lot to develop a flare as long as its direction is unique. Note Cisco's flare touches Microsoft's at two points in the 1990s, when the latter began to expand into networking (see below). This episode highlights another key aspect of competition and innovation: uniqueness is a relative concept. A firm’s flare length is based on the entire graph. Hence, it is determined not only by its own innovations but also by all other firms’.

Fifth, Microsoft dominated the PC operating system (OS) market, first with MS-DOS and then with Windows, which was released in 1985. Since the 1990s, Microsoft has increasingly diversified from the OS market. It introduced the Office suite in 1990, Internet Explorer in 1995, and Xbox in 2001. Hence, Microsoft’s patent portfolio is more diversified than Cisco’s, but their overall trajectories are similar: both of them were close to other IT firms until the late 1980s (Microsoft) or the early 1990s (Cisco) and then grew into individual flares. Their paths crossed again in the mid-to-late 1990s as Microsoft expanded into computer networking in 1995.

\paragraph{Engineering Conglomerates.}

Engineering giants cluster together and constitute a large island in Figure \ref{Figure - mapper(cos_log_details_2)} (a). General Electric (GE), an archetypical conglomerate, holds one of the most diversified portfolios in our data. Its only peers are similarly diversified manufacturers of electronic and capital goods, such as Siemens, Philips, and Mitsubishi Electric.

\begin{figure}[htb!!!!]
\caption{Engineering, Pharmaceuticals, and Chemicals}%
\begin{center}
\begin{subfigure}{0.34\textwidth}
\centering
\includegraphics[width=1\linewidth]{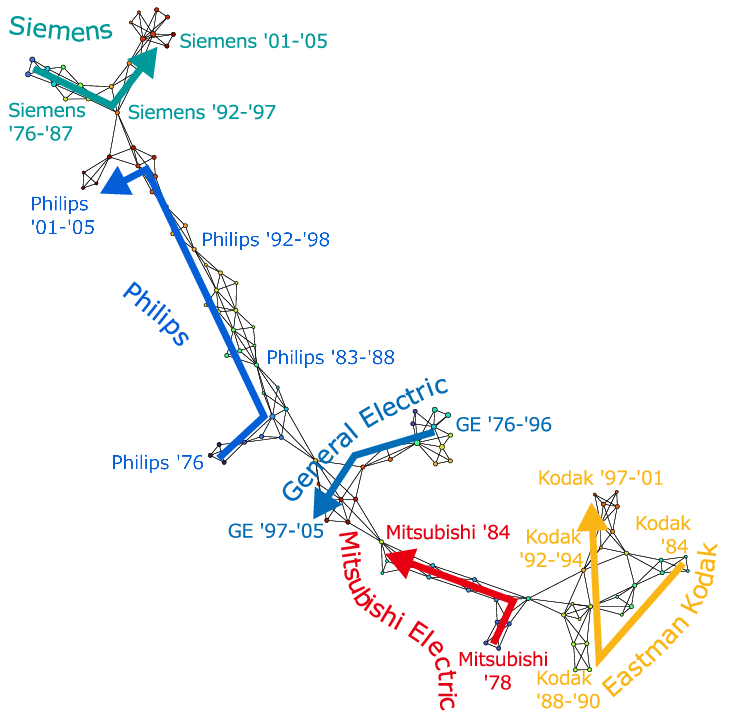}
\caption{Conglomerates}%
\end{subfigure}
\begin{subfigure}{0.64\textwidth}
\centering
\includegraphics[width=1\linewidth]{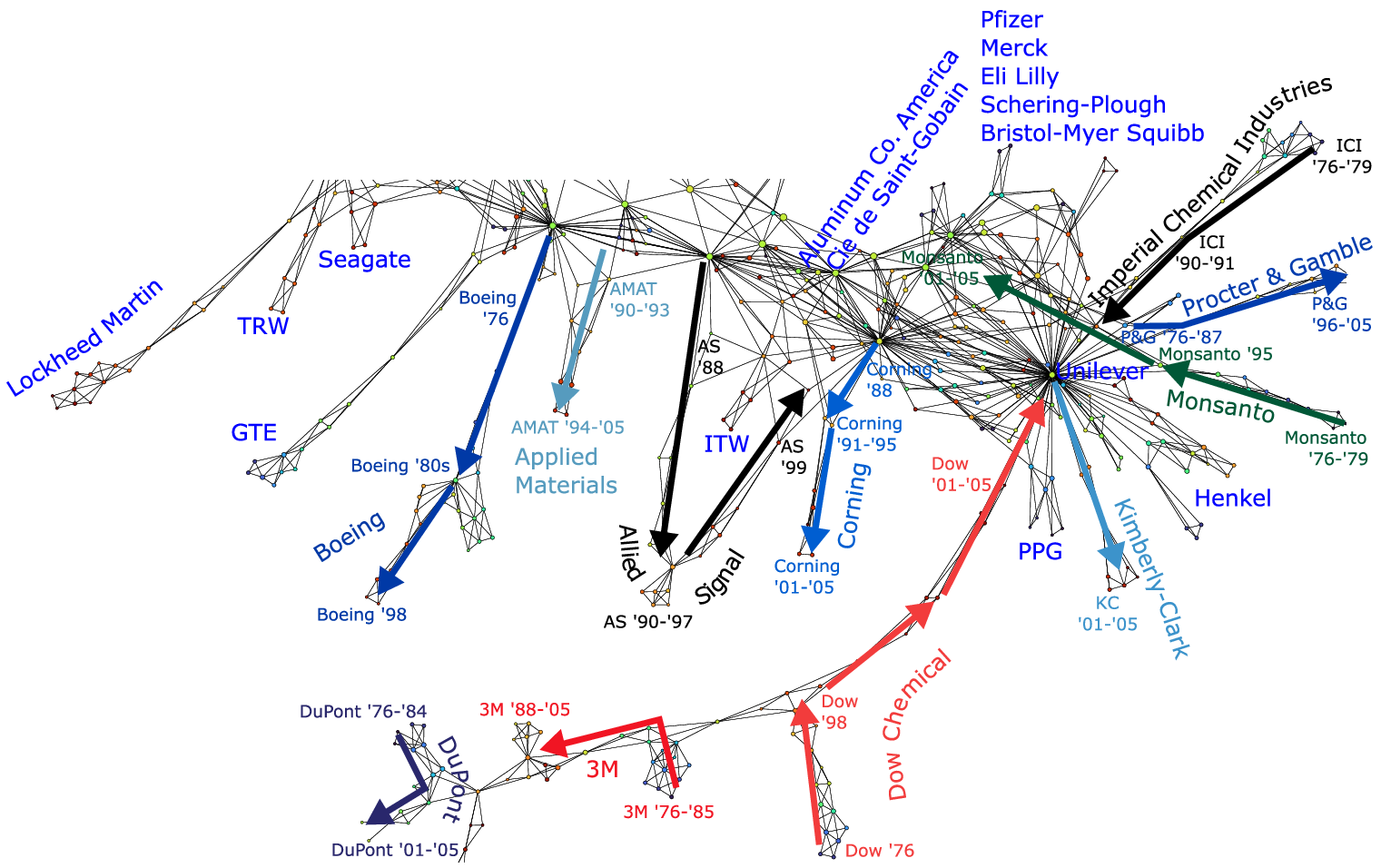}
\caption{Pharmaceuticals and Chemicals}%
\end{subfigure}
\end{center}
\caption*{\footnotesize {%
\textit{Note}: These maps are detailed versions of the southern half of Figure \ref{Figure - mapper(n20_m0_cos)} (b). Arrows indicate the directions of moves of selected firms. Node colors represent the average year of the firm-years in that cluster, with earlier years in blue and later years in red.}}%
\label{Figure - mapper(cos_log_details_2)}
\end{figure}%

\paragraph{Pharmaceuticals and Chemicals.}

Health care is another R\&D-intensive sector, and patent protection is crucial for its business model. Unlike IT firms, however, pharmaceutical firms do not appear in flares or islands. Large drug makers, such as Pfizer, Merck, and Eli Lilly, are clustered in the southern ``peninsula,'' as Figure \ref{Figure - mapper(cos_log_details_2)} (b) shows, because most of the drug patents are in either class 424 or 514 (both are labeled “drug, bio-affecting, and body-treating compositions”), which limits the extent to which their patent portfolios could differ from each other. Further investigations into drugs would require subclass-level data.

Household chemicals firms appear near drug makers because some of their products are based on similar materials. Johnson and Johnson (J\&J), Unilever, Procter and Gamble (P\&G), and Kimberly-Clark hold patents in not only classes such as 510 (cleaning compositions), but also 424 (drugs) and 604 (surgery).

Whereas most of the flares that we have scrutinized so far represented firms’ outward movements, the chemicals industry features a few counterexamples, that is, firms whose technological trajectories are \textit{centripetal} (i.e., moving inward) rather than \textit{centrifugal} (i.e., moving outward). Monsanto was famous for Roundup, a herbicide developed in the 1970s, but became an agri-biotech business in the 1980s and a major producer of genetically engineered crops. In 1997–2002, it divested most of agrochemical businesses and focused on biotechnology, adopting the R\&D/patent-intensive business model of biotech drug companies. This novel strategy shows up as a long march inward, from the periphery to one of the core drugs clusters.

Imperial Chemical Industries (ICI) forms another centripetal flare. ICI used to be one of the largest British firms, but divested most of its bulk chemicals businesses in 1991–2007 to focus on specialty chemicals. One of its spin-offs, Zeneca, merged with Astra to form AstraZeneca, a drugs company, in 1999.

Finally, conglomerates in general chemistry (DuPont, 3M, and Dow) form their own long flares together, not unlike the engineering conglomerates’ island. Dow connects with the rest of the chemicals firms via its long centripetal flare, because it has been increasingly focusing on specialty chemicals, including materials for pharmaceuticals, paper coatings, and advanced electronics. Seeds from genetically modified plants also play an important role in its agri-business. Hence, its strategy is broadly similar to ICI's and Monsanto's.

Whereas most of the IT success stories are associated with long, centrifugal flares, some of the most interesting chemicals firms appear in centripetal flares. The reason is that many of them had already become big conglomerates by 1976 and were ripe for restructuring and divestiture, which tend to generate centripetal movements due to downsizing (recall the path of IBM). Thus, the contrast between IT and chemicals reflects their historical differences.

\paragraph{Summary.}

These examples demonstrate close connections between firms' locations on the map and their actual histories of R\&D (we also investigate M\&A patents in Appendix E.4). The ability to accurately track the trajectories of individual firms, as well as their collective patterns at the industry and sector levels, is Mapper's advantage over existing methods, such as PCA and clustering.

\subsection{Comparison with Jaffe's (1989) Clustering Method}

How do our results differ from Jaffe’s (1989)?\footnote{Appendix D explains their methodological differences in detail and presents an alternative Mapper graph based on Jaffe’s data-transformation convention.} Table \ref{Table - K-Medoids Clustering} shows a list of clusters that global clustering \`a la Jaffe generates. The grouping seems intuitive, with clusters of firms in engineering (cluster 1), telecommunications (2), materials (3), medical devices (4), pharmaceuticals (5), and so on. Jaffe studies firms that ``move'' over time, which he defines as firms that belong to multiple clusters over the years. For example, clusters 7 (computers), 10 (semiconductors), and 11 (electronics) commonly feature Intel and HP. Monsanto appears in both clusters 6 (chemicals) and 15 (genomics). Classifying them as ``movers'' is consistent with their long flares in our Mapper graph (see section 4.3).

\begin{table}[tbh]
\caption{K-Medoids Clustering \`{a} la Jaffe (1989)}
\begin{center}
\fontsize{9pt}{11pt}\selectfont%
\begin{tabular}{lccc}
\hline \hline
Cluster & Number of & Number of & Representative firms\\ 
& firm-years & unique firms & (non-exclusive list of longest-appearing firms in each cluster) \\ \hline
1 & 814 & 88 & Bosch, Halliburton, Schlumberger, Westinghouse \\ 
2 & 576 & 57 & Ericsson, Alcatel, AT\&T, Siemens, Qualcomm \\ 
3 & 548 & 61 & Saint-Gobain, 3M, International Paper, TDK, Alcoa \\ 
4 & 461 & 35 & Stryker, Kimberly-Clark, C.R.Bard, Baxter Travenol, J\&J, P\&G
\\ 
5 & 433 & 37 & Abbott, Eli Lilly, Merck, Pfizer, Sandoz, Schering-Plough \\ 
6 & 371 & 40 & Dow, DuPont, Henkel, Bayer, Monsanto, Bausch \& Lomb \\ 
7 & 365 & 56 & Seagate, Unisys, IBM, Dell, Sun, Intel, HP \\ 
8 & 303 & 38 & Millipore, Pall, Parker Hannifin, Osmonics, Dover, U.S.Filter
\\ 
9 & 287 & 38 & Lockheed Martin, Raytheon, BAE Systems, Northrop Grumman \\ 
10 & 285 & 23 & TI, National Semiconductor, LSI Logic, Cypress, Intel,
Motorola \\ 
11 & 283 & 44 & Tektronix, Teradyne, Philips, Varian, HP, Baker Hughes \\ 
12 & 262 & 41 & Pitney Bowes, BMC Software, Oracle, AOL, Yahoo!, Google, eBay
\\ 
13 & 237 & 31 & Asyst, Rubbermaid, K2, Tenneco Automotive, TRW \\ 
14 & 234 & 27 & Teleflex, Eaton, Dana, Deere, EG\&G, Roper Industries \\ 
15 & 214 & 31 & Chiron, Amgen, Genzyme, Invitrogen, Beckman Coulter, Monsanto
\\ 
16 & 195 & 27 & Apple, Silicon Graphics, Adobe, Sun, Dassault, Disney, NVIDIA
\\ 
17 & 156 & 22 & Coherent, Electro Scientific, Finisar, Newport, Corning,
Alcoa, TRW \\ 
18 & 148 & 27 & AMAT, Nordson, Advanced Energy, EMCORE, Cookson \\ 
19 & 141 & 14 & Mattel, Hasbro, Leggett \& Platt, Tyco, International Game
Tech \\ 
20 & 73 & 9 & Medtronic, Greatbatch, Cordis, Respironics, Roper Industries
\\ 
21 & 48 & 10 & Nuance Communications, Lernout \& Hauspie Speech, ScanSoft \\ 
Total & 6,434 & 756 &  \\ \hline \hline
\end{tabular}
\begin{minipage}{450pt}
{\fontsize{9pt}{9pt}\selectfont \smallskip  \textit{Note}: The number of clusters (21) follows Jaffe's original specification. The total number of unique firms exceeds 333, because many firms appear in multiple clusters. Whereas Jaffe (1989) uses k-means clustering, we use its variant, k-medoids clustering. K-means clustering of our data leads to an extreme result in which a single cluster contains more than 70\% of all firm-years, because so many firm-years are located in the densely populated neighborhood of electronics and engineering. See Appendix D for the result of k-means clustering.}
\end{minipage}
\end{center}
\label{Table - K-Medoids Clustering}
\end{table}

However, Jaffe-style clustering misclassify many other firms. The following firms exhibit flares---and therefore clearly move---in our Mapper graph but do not ``move'' between the Jaffe clusters in Table \ref{Table - K-Medoids Clustering}: Bosch (cluster 1), Ericsson (2), Kimberly-Clark (4), P\&G (4), Dow (6), IBM (7), Lockheed Martin (9), National Semiconductor (10), Corning (17), and Applied Materials (AMAT, 18). They happen to be near the centers of their respective clusters. By contrast, Alcoa (clusters 3 and 17) and Roper Industries (14 and 20) appear in multiple clusters and would be classified as ``movers'' by Jaffe even though they hardly show any flares in our graph. They appear to ``move'' only because the clustering algorithm happens to draw boundaries in the middle of their data points (and not because they actually traveled long distances).

These ``false negatives'' and ``false positives'' highlight the arbitrariness of cluster boundaries. Jaffe's clusters do contain similar firms on average, but their boundaries are ultimately an artifact of discretization and add too much noise at the firm level. This lack of precision is consequential: Jaffe tried but failed to find statistically significant relationships between firms’ performances and whether they ``moved'' in the technological space. We tackle the same problem and find statistically significant relationships in the next section.

%%% Local Variables:
%%% mode: latex
%%% TeX-master: "TDAPatent"
%%% End:

\section{Measuring Unique Technological Trajectories}

Given the prominence of flares and islands in the shape graph of our data, as well as their apparent connections to the firms' R\&D strategies, their systematic measurement seems desirable. We formalize the notion of \textquotedblleft firms' unique technological
trajectories\textquotedblright \ and propose a method to measure their
lengths in section 5.1. We then establish their empirical relevance in terms of
correlations with the firms' financial performances in section 5.2. Sections 5.3--5.5 present their economic interpretations, sensitivity analysis, and comparisons with other measures, respectively.

\subsection{Definition and Measurement of Flares}

We use graph theory to formalize the notion of firms' unique technological
trajectories. Our exposition here is brief and
intuitive; see Appendix C.3 for proofs and computational details.

We aim to define each firm's unique trajectory as a flare and measure its
length in the graph $G=\left( V,E\right) $ of our data, which requires several auxiliary concepts. Let us focus on a
subgraph $G_{i}$ of $G$ that consists of nodes that contain firm $i$ and the
edges among them. We decompose $G_{i}$ into \textquotedblleft
interior\textquotedblright \ and \textquotedblleft
boundary.\textquotedblright \ The \textit{interior} $F_{i}$ is the nodes in $G_{i}$ whose immediate neighbors also contain firm $i$, whereas the \textit{boundary} $G_{i}\setminus F_{i}$ (i.e., the
rest of $F_{i}$) consists of the nodes in $G_{i}$ that connect with nodes not containing firm $i$. Appendix C.3 features a pictured example.

We further decompose $F_{i}$ into \textquotedblleft isolated
pieces\textquotedblright \ (\textit{connected components}, formally) as $%
F_{i}=R_{1}\sqcup R_{2}\sqcup ...\sqcup R_{S}$,\footnote{%
In graph theory, a (connected) component of an undirected graph is a
connected subgraph that is not part of any larger connected subgraph.} and
classify each $R_{s}$ as either a \textquotedblleft flare\textquotedblright
\ or an \textquotedblleft island.\textquotedblright \ If $R_{s}$ is also a
connected component of $G$ (i.e., if it is \textquotedblleft
isolated\textquotedblright \ in the context of the full graph), we call $%
R_{s} $ an \textit{island} of firm $i$. Otherwise, we call it a \textit{flare%
} of firm $i$.

To introduce the notion of length, we define an \textit{exit distance} for each node $u$ in $%
F_{i}$ as%
\begin{equation}
e_{i}\left( u\right) =\min \left \{ d\left( u,v\right) |v\in G\setminus
F_{i}\right \} ,  \label{eq - exit distance}
\end{equation}%
where $d\left( u,v\right) $ is the distance between nodes $u$ and $v$ in $G$.%
\footnote{%
In graph theory, distance $d_{G}\left( u,v\right) $ is defined as the
minimum length of paths in $G$ from $u$ to $v$, which we write $d\left(
u,v\right) $ for short. We assume a unit weight on every edge when we
calculate path lengths, but our method can be extended to handle any
positive weights.} In words, the exit distance is the shortest length of
path to get out of firm $i$'s interior. Thus, $e_{i}\left( u\right) $
represents the extent to which technological location $u$ (or all firm-year
observations $l_{i,t}$ that constitute cluster $u$) is differentiated from
the nearest rival's subgraph. In the case of islands, we set $e_{i}\left(
u\right) =\infty $ because no such path exists.

Computing $e_{i}\left( u\right) $ based on its definition (\ref{eq - exit
distance}) is costly because it requires information on the length of 
\textit{all} paths in $G$. Fortunately, we can show that%
\begin{equation}
e_{i}\left( u\right) =\min \left \{ d_{G_{i}}\left( u,w\right) |w\in
G_{i}\setminus F_{i}\right \} ,  \label{eq - exit distance 2}
\end{equation}%
where $d_{G_{i}}\left( u,w\right) $ is the distance between $u$ and $w$ in $%
G_{i}$ (see Appendix C.3 for the proof). Thus, we can compute $e_{i}\left( u\right) $ based only on firm $i$'s
subgraph $G_{i}$, not the entirety of $G$.

Next, we characterize each connected component (i.e., flare or island) $R_{s}
$ of $F_{i}$ based on the longest exit distance of its constituent nodes,%
\begin{equation*}
\lambda _{i}\left( R_{s}\right) =\max_{u\in R_{s}}e_{i}\left( u\right) ,
\end{equation*}%
and call it the \textit{flare index} of $R_{s}$. In other words, we
aggregate the node-level information about exit distances at the level of
connected components. We further aggregate $\lambda _{i}\left( R_{s}\right) $
at the firm level by defining the \textit{flare signature} of firm $i$ as
the multiset\footnote{%
A multiset is a modification of the concept of a set that, unlike a set,
allows for multiple instances for each of its elements. We denote it by
double braces $\{ \{,\} \}$ to distinguish it from a set.}%
\begin{equation*}
\vec{\lambda}_{i}=\left \{ \left \{ \lambda _{i}\left( R_{s}\right)
|s=1,...,S\right \} \right \} .
\end{equation*}%
Four cases are possible. First, if $F_{i}$ is empty (i.e., no interior exists
in $G_{i}$), no flares or islands exist, and we define $\vec{\lambda}_{i}$\
as an empty multiset. Second, if only flares exist in $F_{i}$, $\vec{\lambda}%
_{i}$ contains only finite elements. Third, if only islands exist in $F_{i}$%
, $\vec{\lambda}_{i}$\ contains only copies of $\infty $. Fourth, if both
flares and islands exist in $F_{i}$, $\vec{\lambda}_{i}$\ contains both
finite elements and copies of $\infty $.

Finally, we define the \textit{flare length} of firm $i$ as%
\begin{equation*}
\lambda _{i}=\left \{ 
\begin{array}{ll}
0 & \text{if }F_{i}=\emptyset \text{ (i.e., no flares or islands in $G_{i}$)} \\ 
\mathop{\mathrm{finmax}}(\vec{\lambda}_{i}) & \text{if }\vec{\lambda}_{i}%
\text{\ has at least one finite element (i.e., some flares exist)} \\ 
\infty  & \text{otherwise (i.e., only islands exist in $G_{i}$),}%
\end{array}%
\right. 
\end{equation*}%
where $\mathop{\mathrm{finmax}}(\vec{\lambda}_{i})$ is the maximum among all
finite elements of $\vec{\lambda}_{i}$. Thus, we propose to measure the
length of firm $i$'s unique technological trajectory by its longest flare.

\subsection{Flares and Firms' Performances}

These formal definitions help us detect \textit{all} firms' flares, including those that are located within the densely populated areas. Table \ref{Table - Flare histogram} shows that, whereas our visual inspection in section 4 identified only a few dozen flares and islands, this systematic examination reveals the existence of many more: 40.3 \% of our sample (133 firms) shows some flares.

\begin{table}[tbh!!!!]
\caption{Firm Count by Flare Length}
\begin{center}
\fontsize{9pt}{11pt}\selectfont%
\begin{tabular}{lcccccccccc}
\hline \hline
Flare length & $0$ & $1$ & $2$ & $3$ & $4$ & $5$ & $6$ & $7$ & $8$ & $\infty $ (islands only) \\ \hline
Frequency & $197$ & $78$ & $19$ & $13$ & $10$ & $5$ & $3$ & $1$ & $1$ & $3$
\\ 
Percentage & $59.70$ & $23.64$ & $5.76$ & $3.94$ & $3.03$ & $1.52$ & $0.91$
& $0.30$ & $0.30$ & $0.91$ \\ 
Cumulative \% & $59.70$ & $83.33$ & $89.09$ & $93.03$ & $96.06$ & $97.58$ & $%
98.48$ & $98.79$ & $99.09$ & $100.00$ \\ \hline \hline
\end{tabular}
\begin{minipage}{450pt}
{\fontsize{9pt}{9pt}\selectfont \smallskip  \textit{Note}: The underlying Mapper graph uses log-transform, cosine distance, $n=20$, and $o=0.5$. See section 4.2.}
\end{minipage}
\end{center}
\label{Table - Flare histogram}
\end{table}

What makes portfolios ``unique''? Raw data at the firm level suggest both the quantity and variety of patents help make their portfolios unique. For example, HP has a massive portfolio and a flare of length 6, whereas Dell's portfolio is much smaller and its flare length is 1 (see Appendix B for further details on HP, Dell, and Qualcomm). However, these conditions are not sufficient for long flares, because uniqueness is a relative concept. Our definition of flare is based on $G$, the graph of all firms in all years. Hence, the firm's flare length depends on not only its own activities but also all other firms'.  

\medskip 
 
In the remainder of this section, we investigate whether flares contain any ``relevant'' information. Following a common practice in the patent statistics literature (e.g., \cite{PakesGriliches1984}, \cite{JAFFE198987}, and \cite{hall2005market}), we look for correlations between these topological characteristics and the firms' performance metrics, including revenue, profit, and stock market value.

Let us study their correlations by running regressions of the following form:%
\begin{equation}
\ln (y_{i})=\alpha _{1}+\alpha _{2}\lambda_{i}+\alpha _{3}\mathbb{I}\left \{ \lambda_{i}=\infty \right \}
+\alpha _{4} \ln (p_{i}) +\varepsilon _{i},
\end{equation}%
where $y_{i}$ is firm $i$'s revenue (or other
performance metrics) in 2005, $\lambda_{i}$ is the flare length of its patent portfolio's evolution in 1976--2005, $\mathbb{I}\left \{ \lambda_{i}=\infty \right \} $ is a dummy variable indicating the islands-only type, $p_{i}$ is the total count of firm $i$'s patents in 1976--2005 (i.e., $p_{i}=\sum_{t} \sum_{c}p_{i,t,c}$), $\alpha $s are their coefficients, and $\varepsilon _{i}$ is an error term.\footnote{%
Note we do not intend to prove causal relationships or their specific channels. Our purpose is to
assess the extent to which our topological measures predict these
performance metrics.} We include $\ln({p_{i}}) $\ to control for the size of the firm's inventive activities.%

\begin{table}[tbh!!!]
\caption{Flares, Counts, and Performances}
\begin{center}
\fontsize{9pt}{11pt}\selectfont%
\begin{tabular}{cccccccccccc}
\hline \hline
LHS variable: & \multicolumn{3}{c}{Log(Revenue)} &  & \multicolumn{3}{c}{
Log(EBIT)} &  & \multicolumn{3}{c}{Log(Market value)} \\ 
\cline{2-4}\cline{3-4}\cline{6-8}\cline{10-12}
& (1) & (2) & (3) &  & (4) & (5) & (6) &  & (7) & (8) & (9) \\ \hline
\multicolumn{1}{l}{Flare length} & $0.65$ & $-$ & $0.34$ &  & $0.65$ & $-$ & 
$0.33$ &  & $0.66$ & $-$ & $0.27$ \\ 
\multicolumn{1}{l}{} & $\left( 0.06\right) $ & $\left( -\right) $ & $\left(
0.08\right) $ &  & $\left( 0.07\right) $ & $\left( -\right) $ & $\left(
0.08\right) $ &  & $\left( 0.07\right) $ & $\left( -\right) $ & $\left(
0.08\right) $ \\ 
\multicolumn{1}{l}{Islands only} & $2.27$ & $-$ & $0.95$ &  & $2.32$ & $-$ & 
$0.94$ &  & $2.31$ & $-$ & $0.70$ \\ 
\multicolumn{1}{l}{} & $\left( 0.86\right) $ & $\left( -\right) $ & $\left(
0.84\right) $ &  & $\left( 0.91\right) $ & $\left( -\right) $ & $\left(
0.88\right) $ &  & $\left( 0.94\right) $ & $\left( -\right) $ & $\left(
0.90\right) $ \\ 
\multicolumn{1}{l}{Log(Patents)} & $-$ & $0.40$ & $0.28$ &  & $-$ & $0.41$ & 
$0.29$ &  & $-$ & $0.44$ & $0.34$ \\ 
\multicolumn{1}{l}{} & $\left( -\right) $ & $\left( 0.03\right) $ & $\left(
0.04\right) $ &  & $\left( -\right) $ & $\left( 0.04\right) $ & $\left(
0.05\right) $ &  & $\left( -\right) $ & $\left( 0.04\right) $ & $\left(
0.05\right) $ \\ 
\multicolumn{1}{l}{Constant} & $7.38$ & $5.65$ & $6.08$ &  & $5.35
$ & $3.56$ & $3.97$ &  & $7.80$ & $5.86$ & $6.20$ \\ 
\multicolumn{1}{l}{} & $\left( 0.10\right) $ & $\left( 0.20\right) $ & $%
\left( 0.22\right) $ &  & $\left( 0.10\right) $ & $\left( 0.22\right) $ & $%
\left( 0.24\right) $ &  & $\left( 0.10\right) $ & $\left( 0.22\right) $ & $%
\left( 0.24\right) $ \\ 
\multicolumn{1}{l}{$R^{2}$} & $0.261$ & $0.305$ & $0.345$ &  & $0.256$ & $0.308$
& $0.345$ &  & $0.233$ & $0.320$ & $0.343$ \\ 
\multicolumn{1}{l}{Adjusted $R^{2}$} & $0.257$ & $0.303$ & $0.339$ &  & $0.251$
& $0.306$ & $0.338$ &  & $0.228$ & $0.317$ & $0.336$ \\ 
\multicolumn{1}{l}{Number of observations} & $328$ & $328$ & $328$ &  & $301$
& $301$ & $301$ &  & $325$ & $325$ & $325$ \\ \hline \hline
\end{tabular}
\begin{minipage}{450pt}
{\fontsize{9pt}{9pt}\selectfont \smallskip  \textit{Note}: The left-hand side (LHS) variables are as of 2005 or the latest years available in Compustat. The RHS variables are based on our topological characterization of the patent statistics in 1976--2005. The number of observations varies across columns, because some firms in our patent database lack information on certain metrics in Compustat. In columns 4--6, firms with negative EBIT drop out due to log-transformation. See Appendices E and F for results under alternative specifications and in panel data, respectively. Standard errors are in parentheses.}
\end{minipage}
\end{center}
\label{Table - Regressions}
\end{table}

Table \ref{Table - Regressions} shows flare length is positively correlated with the firm's revenue, EBIT, and market value in 2005. Columns 1, 4, and 7 use the flare variables alone; columns 2, 5, and 8 use $\ln (p_{i})$ alone; and columns 3, 6, and 9 use both. The purpose of comparison is to assess whether our topological characteristics convey additional information above and beyond what patent count alone could predict. The differences between the adjusted $R^{2}$s suggest they do. More formally, the F-tests of a linear restriction, $\alpha _{2}=\alpha _{3}=0$, reject the null hypothesis at the 0.01\%, 0.1\%, and 1\% levels for the revenue, EBIT, and market-value regressions, respectively.\footnote{We calculate $F=[(R^{2}_{ur} - R^{2}_{r})/2] / [(1 - R^{2}_{ur})/(\#obs - 4)]$, where $R^{2}_{ur}$ is the $R^{2}$ of the unrestricted model in column 3 (6 or 9), $R^{2}_{r}$ is the $R^{2}$ of the restricted model in column 2 (5 or 8), and $\#obs$ is the number of observations (328, 301, or 325). We reject the null hypothesis, $\alpha_{2}=\alpha_{3}=0$, if $F$ is greater than the corresponding critical value of the F distribution.} Hence, the incremental contribution of the flare-and-island variables is statistically highly significant.

What about their economic significance? The estimates of $\alpha _{2}$ are 0.34, 0.33, and 0.27 in columns 3, 6, and 9 (i.e., after controlling for $p_{i}$), respectively, which imply an extra length of flare is associated with 40\%, 39\%, and 31\% higher performances in terms of revenue, EBIT, and market value, respectively.\footnote{Likewise, the estimates of $\alpha _{3}$ (0.95, 0.94, and 0.70 in the same three columns) suggest islands-only firms tend to outperform no-flare firms by 159\%, 156\%, and 101\% in these measures, respectively. However, their standard errors are large. Only three firms belong to this category, and all of them have relatively large patent portfolios, which makes $\alpha _{3}$ difficult to isolate from $\alpha _{4}$. Nevertheless, we keep $\mathbb{I}\left \{ \lambda_{i}=\infty \right \} $\ in these columns, because dropping it (and thereby grouping them with no-flare firms) would be unwise given the results on columns 1, 4, and 7.}

\subsection{Economic Interpretations}

Why do flares predict firms' success? Let us interpret these findings based on our model in section 2. First, flares reflect unique technological trajectories. Unique technologies permit product differentiation, which softens price competition (or avoid competition altogether) and increases profits. Specifically, unique technological location $l_{i,t}$ allows the firm to enter a new product market with low $N_{m,t}$. This mechanism directly connects $l_{i,t}$ with $\pi_{i,t}$.

Second, these extra profits could help finance subsequent R\&D expenditure $b^{x}_{i,t}$, thereby reinforcing the firm's technological differentiation and conquest of new markets: a virtuous cycle. The length of flare reflects a string of unique $l_{i,t}$s and a track record of successful technological development in a unique direction. Hence, it is a good proxy for the duration of such virtuous cycles. These dynamics imply positive correlations between $\lambda_{i}$ and $\pi_{i}$.\footnote{One might wonder how our definition of flare length---which does not explicitly incorporate the time dimension---can capture the firm's actual duration of travel without bumping into its rivals in real time. We discuss this issue in Appendix C.4.}

Third, the fact that $\lambda_{i}$ conveys information above and beyond what $p_{i,t}$ predicts---which is known to be strongly correlated with firm size and R\&D expenditure (e.g., \cite{Cohen2010})---suggests $\lambda_{i}$ captures more than just budget size $b_{i,t}$. Our model predicts connections between $\lambda_{i}$, $l_{i,t}$, and technological capabilities $\mathbf{\omega}^{x}_{i,t}$; our findings from panel-data regressions (in section 5.4) confirm the presence of persistent firm heterogeneity and its correlation with $\lambda_{i}$.

Thus, our empirical results---interpreted in the context of our model of competition and innovation---highlight the importance of the direction of innovation. Unique technological positions directly contribute to profits, which reinforces subsequent innovations and long track records. These dynamics reflect the firms' desire to avoid competition, conquer new markets, and exploit their idiosyncratic technological capabilities.

Finally, why are some firms profitable despite showing short or no flares? Our model permits two firm-level characteristics other than technologies: quality $\xi_{i,t}$ and cost $c_{i,t}$. Those who have comparative advantage in marketing or operations (i.e., high $\omega^{\xi}_{i,t}$ or $\omega^{c}_{i,t}$) would keep exploiting the existing markets by investing in $\xi_{i,t}$ or $c_{i,t}$ instead of technologies.

\subsection{Sensitivity Analysis}

This section assesses the sensitivity of our results to (i) the specification of the Mapper procedure, (ii) subsampling based on firms’ survival, (iii) subsampling based on industry classification, and (iv) panel-data regressions.

\begin{table}[tbh]
\caption{Sensitivity of Mapper Graph}
\begin{center}
\fontsize{9pt}{11pt}\selectfont%
\begin{tabular}{lccccccc}
\hline \hline
Specification & \# nodes & \# edges & Avg. & \# conn. & Avg. \# & Avg. flare
& Readability of \\ 
&  &  & degree & comp. & nodes/firm & length & output graph \\ \hline
1. Baseline & $1,214$ & $2,926$ & $4.82$ & $27$ & $8.70$ & $0.77$ & 
Informative and tractable \\ 
2. $f=$ 3D-PCA & $4,041$ & $21,300$ & $10.54$ & $16$ & $22.20$ & $1.01$ & 
Redundant nodes \& edges \\ 
3. $f=$ 2D-MDS & $822$ & $2,279$ & $5.55$ & $13$ & $12.79$ & $0.44$ & Fewer
flares, lattice-like \\ 
4. $J=225$ & $993$ & $2,371$ & $4.78$ & $21$ & $7.16$ & $0.74$ & Simpler but
similar to S1 \\ 
5. $J=625$ & \thinspace $1,482$ & $3,628$ & $4.90$ & $32$ & $10.07$ & $0.91$
& Messier but similar to S1 \\ 
6. Cluster = HC-w & $787$ & $2,106$ & $5.35$ & $16$ & $8.48$ & $0.61$ & 
Shorter flares but similar \\ 
7. Cluster = HC-a & $743$ & $1,930$ & $5.20$ & $15$ & $8.44$ & $0.60$ & 
Shorter flares but similar \\ 
8. Cluster = HC-c & $692$ & $1,912$ & $5.53$ & $13$ & $8.42$ & $0.55$ & 
Shorter flares but similar \\ 
9. Cluster = HC-m & $1,145$ & $2,831$ & $4.95$ & $19$ & $8.57$ & $0.72$ & 
Similar to S1 \\ 
10. Cluster = HC-l & $1,050$ & $2,702$ & $5.15$ & $13$ & $8.48$ & $0.65$ & 
Similar to S1 \\ 
11. $\delta =$ Euclid & $1,192$ & $2,907$ & $4.88$ & $23$ & $8.74$ & $0.83$
& Similar to S1 \\ 
12. $\delta =$ Correlation & $1,226$ & $2,994$ & $4.88$ & $28$ & $8.69$ & $0.77$
& Similar to S1 \\ 
13. $\delta =$ Min-comp. & $1,146$ & $2,808$ & $4.90$ & $21$ & $8.77$ & $0.76
$ & Similar to S1 \\ 
14. $\delta =$ Mahalanobis & $937$ & $2,418$ & $5.16$ & $23$ & $9.70$ & $0.46$ & Shorter flares but similar \\ 
15. $o=0.3$ & $962$ & $1,503$ & $3.12$ & $143$ & $5.68$ & $0.67$ & Many
fragmented nodes \\ 
16. $o=0.7$ & $1,803$ & $18,153$ & $20.14$ & $8$ & $16.76$ & $0.45$ & 
Redundant nodes \& edges \\ \hline \hline
\end{tabular}
\begin{minipage}{475pt}
{\fontsize{9pt}{9pt}\selectfont \smallskip  \textit{Note}: Baseline specification uses 2D-PCA filter, resolution $J=400$ ($n=20$), single-linkage hierarchical clustering with first-gap heuristics, cosine dissimilarity, and overlap $o=0.5$. HC-w, HC-a, and HC-c stand for hierarchical clustering with weighted, average, and complete linkage methods, respectively. HC-m and HC-l stand for hierarchical clustering with mid-gap and last-gap heuristics to determine the number of clusters, respectively. See the main text of sections 4.2 and 5.4, and Appendix E.1 for details.}
\end{minipage}
\end{center}
\label{Table - List of Specifications}
\end{table}

\paragraph{Mapper Specification.} Table \ref{Table - List of Specifications} reports descriptive statistics of the Mapper graphs under 16 different specifications. Our baseline Specification 1 (S1) generates a graph with 1,214 nodes, 2,926 edges, the average degree of 4.82 (edges per node), 27 connected components, 8.70 nodes per firm, and the average flare length of 0.77. Most of the alternative specifications lead to changes that are either small (S6--S14) or in directions that are consistent with Mapper’s mechanism (S2, S4--S5, and S15--S16). S3’s direction of change is less obvious because it is the only one that uses a non-PCA filter (i.e., takes a different ``angle'' at the data). Nevertheless, its descriptive statistics are comparable to others. Given the diverse set of specifications, perhaps the most surprising finding is that their regression results are remarkably similar to the baseline. Appendix E.1 explains S2--S16 in detail and shows the correlations between firms’ performances and flare length (based on the 16 different graphs) are always positive and statistically significant, with comparable magnitudes.

\paragraph{Survivorship.}
Appendix E.2 shows the results are robust to (i) the elimination of firms that exited our sample before 2005 and (ii) conditioning on the balanced panel.

\paragraph{Subsampling by Sector and Industry.}
These findings are not an artifact of aggregation or driven by a few specific sectors and industries. Appendix E.3 plots revenues and flares by economic sector defined by Standard and Poor's (S\&P), a credit-rating agency. Appendix E.3 also studies the technology sector more deeply at the SIC-code level, with a focus on computers and semiconductor industries. The positive correlations are preserved within each sector and industry. 

\paragraph{Panel Data Regressions}

Whereas our analysis in section 5.2 focuses on the relationships between the firms' flares in the whole graph for 1976--2005 and their eventual performances in 2005, Appendix F shows our findings hold more generally---at different points in time, with many years of lags, and in terms of out-of-sample predictions.

\subsection{Comparison with Other Measures}

This section compares flare length with other measures, including more conventional network-centrality measures and the Jaffe measure of technological distance.

\paragraph{Centrality Measures.} Flare length is the focus of our quantitative analysis because (i) long flares are the most salient feature of our Mapper graph and (ii) our model suggests the length of unique technological trajectories may reflect the firms’ profitability and capabilities. Nevertheless, flare length is not the only way to measure locations on a graph. Measures of network centrality offer more conventional alternatives. Appendix G.1 shows five centrality measures (degree, closeness, harmonic, betweenness, and eigenvector centralities) correlate with the firms' financial performances less strongly than our flare-based measures.

\paragraph{Jaffe's Technological Distance.} 

Both our Mapper graph and Jaffe's (1989) measure of technological distance use patent count and almost identical dissimilarity functions. Hence, one might expect Jaffe's measure to produce similar results. When we regress revenue, EBIT, and market value on the Jaffe distance, however, the fit is nearly zero in many cases (columns 1, 4, and 7 of the table in Appendix G.2). It achieves a reasonable fit when patent count is also included (columns 2, 5, and 8), but its coefficient estimate is statistically insignificant and difficult to interpret (i.e., negative) in most cases. Finally, the inclusion of our flares and islands further improves the adjusted $R^2$, but the coefficient on Jaffe's measure remains insignificant and lacks cohesive patterns.

%%% Local Variables:
%%% mode: latex
%%% TeX-master: "TDAPatent"
%%% End:

\section{Conclusion}

This paper proposes a new method to map, describe, and characterize firms' inventive activities. The shape graph from the Mapper procedure helps us understand where firms and industries are located, how they connect with each other (or not), and how their innovative activities evolve over time. In the past, economists' ability to answer these basic, descriptive questions---and hence the ability to ask and answer deeper, causal/policy questions \textit{that presuppose reliable descriptions or stylized facts}---have been constrained by the ``curse of dimensionality'' of the technological space. With the new tool, we can start revisiting and answering some of the long-standing questions in economics, including the rate and direction of inventive activity. Because its underlying mathematics is general, we believe this method is potentially useful for describing and characterizing other high-dimensional data in economics as well, such as product characteristics and international trade.

%%% Local Variables:
%%% mode: latex
%%% TeX-master: "TDAPatent"
%%% End:

% For JBES: include or exclude the Appendix by '%'
\clearpage
\pagenumbering{arabic} \renewcommand*{\thepage}{A-\arabic{page}} \appendix 
\section*{Appendix A \ Details of the Economic Model}

\subsection*{A.1 \ Demand and Supply}

We omit time subscripts in this section because it is about the micro-foundation of equation \ref{eq - period profit} in the static part of the model.

\paragraph{Demand.}

Each market $m$ is populated by a mass $M_{m}$ of consumers indexed by $h$.
Consumer $h$'s utility from buying and consuming product $i$ (offered by firm $i$) is%
\begin{equation}
u_{h,i}=\theta p_{i}+\xi _{i}+\varepsilon _{h,i},  \label{eq - utility}
\end{equation}%
where $\theta $ is the (dis)taste for paying the price $p_{i}$, $\xi _{i}$
is product quality,\footnote{We do not distinguish between observed and unobserved qualities
because we do not observe any.} and $\varepsilon _{h,i}$ is $h$'s individual
taste for product $i$, which is independently and identically distributed
(i.i.d.) type-1 extreme value. Each consumer chooses up to one product from
the set of available products in $m$, which are supplied by $i=1,2,...,N_{m}$
firms---each of which produces a single product---to maximize utility:%
\begin{equation}
d_{h}=\arg \max_{i}\left \{ u_{h,0},u_{h,1},u_{h,2},...,u_{h,N_{m}}\right \}
,  \label{eq - utility maximization}
\end{equation}%
where $u_{h,0}\equiv \varepsilon _{h,0}$ is the utility from the outside option of not buying
anything. The demand for firm/product $i$ is%
\begin{equation}
q_{i}=M_{m}\times s_{i}=M_{m}\times \frac{\exp \left( \delta _{i}\right) }{%
\sum_{i^{\prime }=0}^{N_{m}}\exp \left( \delta _{i^{\prime }}\right) },
\label{eq - demand}
\end{equation}%
where $\delta _{i}\equiv \theta _{h}p_{i}+\xi _{i}$ is the deterministic
part of utility.

\paragraph{Supply.}

Firm $i$ chooses price $p_{i}$ to maximize profit,%
\begin{equation}
\pi _{i}=\left( p_{i}-c_{i}\right) \times q_{i},  \label{eq - profit}
\end{equation}%
where $c_{i}\in \left( 0,\infty \right) $ is its constant marginal
cost. We assume all of the $N_{m}$ active firms simultaneously choose
prices, and focus on the Nash equilibrium of this Bertrand competition with
differentiated products.

\subsection*{A.2 \ Birth and Death of Firms}

\paragraph{Birth of Firms.}

In every period, $N_{t}^{PE}$ potential entrants are born,
each with the initial endowment of cash, $cash_{i,t}=fin_{i,t}=\phi _{i}$,
as well as the initial levels of productivity for R\&D, marketing, and
operations, $\left( \mathbf{\omega}_{i,t}^{x},\omega _{i,t}^{\xi },\omega
_{i,t}^{c}\right) $, where $\mathbf{\omega}_{i,t}^{x}\equiv \left( \omega
_{i,t,k}^{x}\right) _{k=1}^{K}$. These productivity levels evolve
according to an exogenous first-order Markov process. These initial values
are drawn from some distribution, such as log-normal one ($\log \left( \phi
_{i}\right) \sim N\left( \mu ^{\phi },\sigma ^{\phi }\right) $), whereas the
new firm's technological state $x_{i,t}=\left( 0,0,...,0\right) $ and
quality $\xi _{it}=0$ must start from zero, and its cost $c_{i,t}=\bar{c}\in
\left( 0,\infty \right) $ from the highest (i.e., least competitive) level.
Subsequently, each of them could become an actual entrant by investing in $%
x_{i,t}$ (to enter one of the many markets), as well as $\xi _{i,t}$ and $c_{i,t}$ (to offer a competitive product).

\paragraph{Liquidation of Firms.}

A firm is liquidated (i.e., permanently ceases all activities) if $%
cash_{i,t}<0$, which is possible when it is hit by a sufficiently negative
financial shock. Let $N_{t}^{X}$ denote the number
of such permanent exits.

\section*{Appendix B \ Raw Data: Where Do Firms Patent?}

Let us illustrate with examples what the firms' patent portfolios look like. Figure \ref{Figure - bubble (examples)} visualizes the evolution of patenting activities at six major firms. Each plot lists the 430 USPTO patent classes on the vertical axis, and the year of application (for R\&D patents) or acquisition (for M\&A patents) on the horizontal axis. The circle size represents the number of patents in each class-year.

\begin{figure}[htb!!!!]
\caption{Acquiring a String of Pearls}%

\begin{subfigure}{0.5\textwidth}
\caption{Cisco Systems}%
\centering
\includegraphics[width=0.95\linewidth]{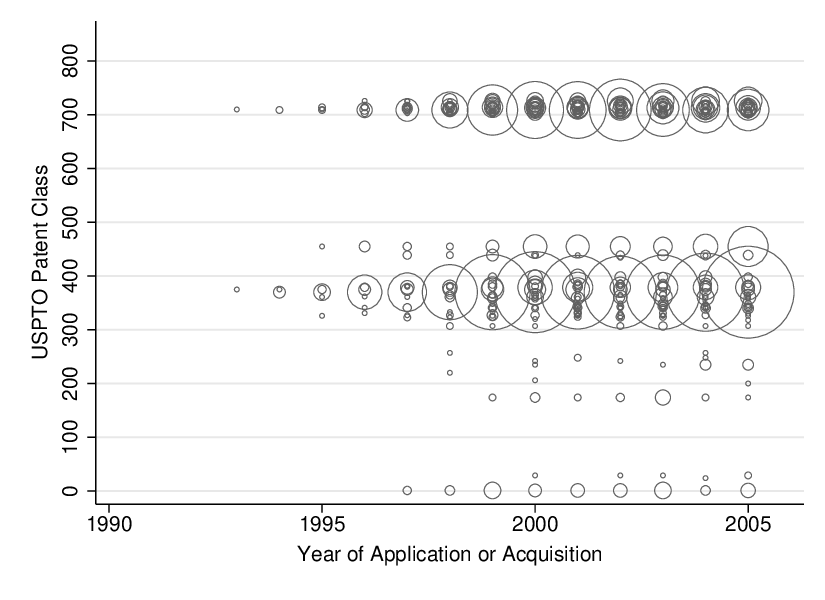}
\end{subfigure}
\begin{subfigure}{0.5\textwidth}
\caption{Seagate Technology}%
\centering
\includegraphics[width=0.95\linewidth]{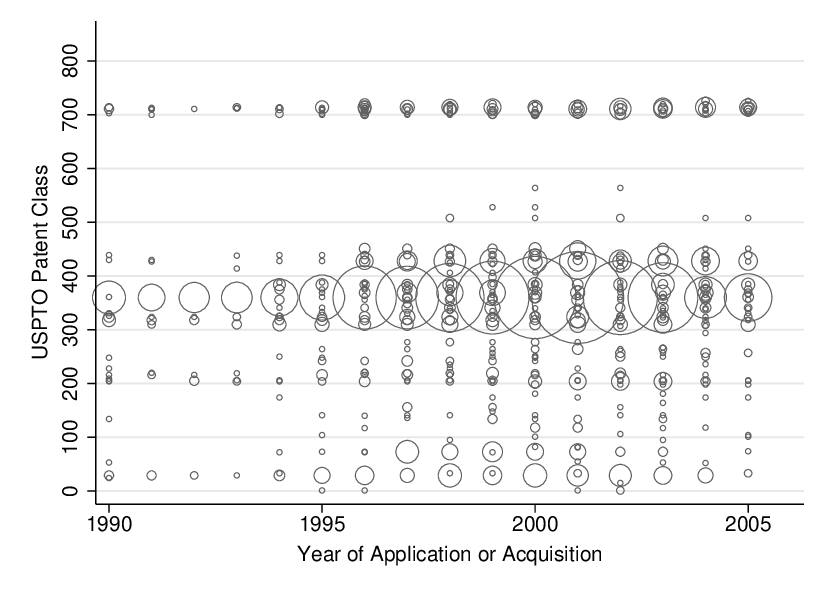}
\end{subfigure}
\begin{subfigure}{0.5\textwidth}

\caption{Pfizer}%
\centering
\includegraphics[width=0.95\linewidth]{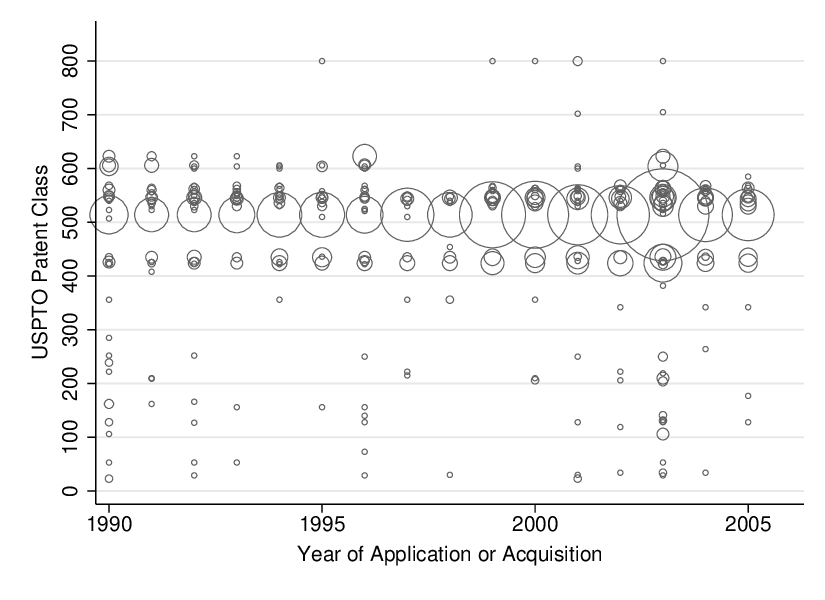}
\end{subfigure}
\begin{subfigure}{0.5\textwidth}
\caption{Medtronic}%
\centering
\includegraphics[width=0.95\linewidth]{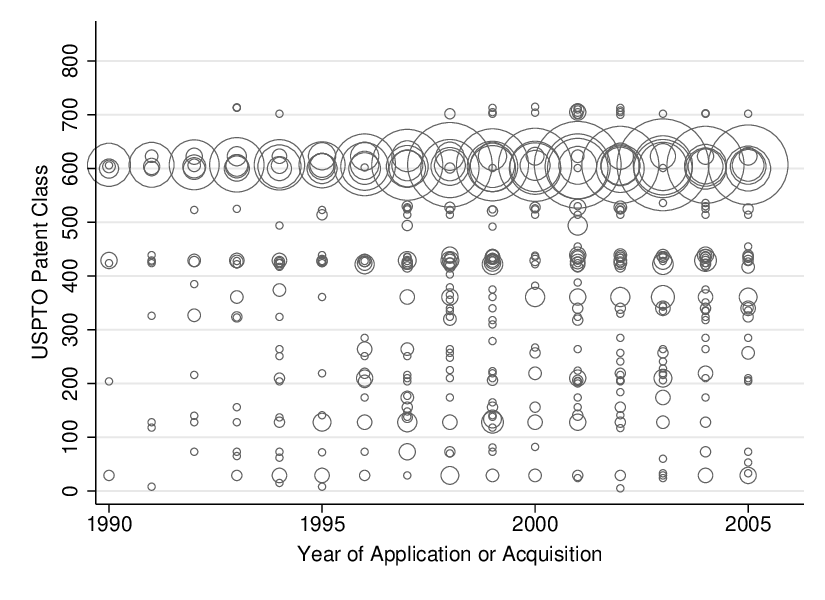}
\end{subfigure}
\begin{subfigure}{0.5\textwidth}

\caption{GE}%
\centering
\includegraphics[width=0.95\linewidth]{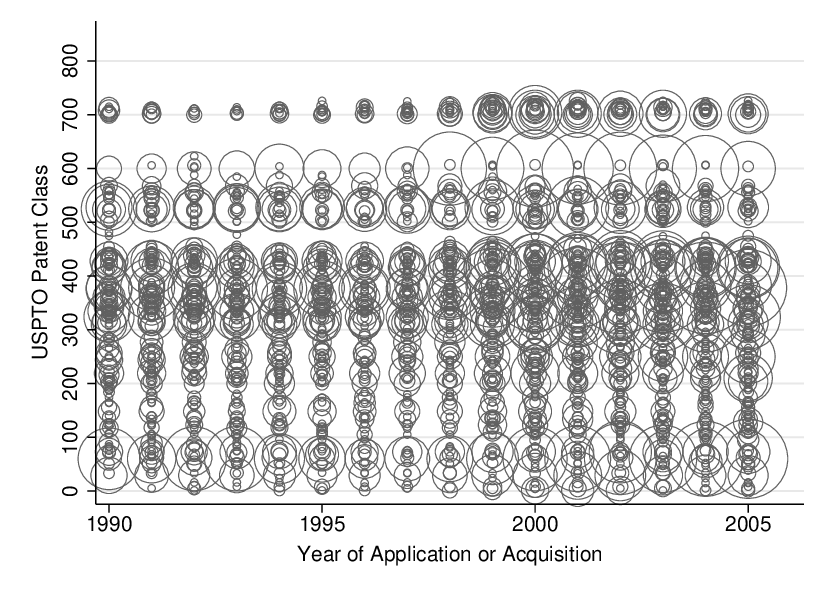}
\end{subfigure}
\begin{subfigure}{0.5\textwidth}
\caption{IBM}%
\centering
\includegraphics[width=0.95\linewidth]{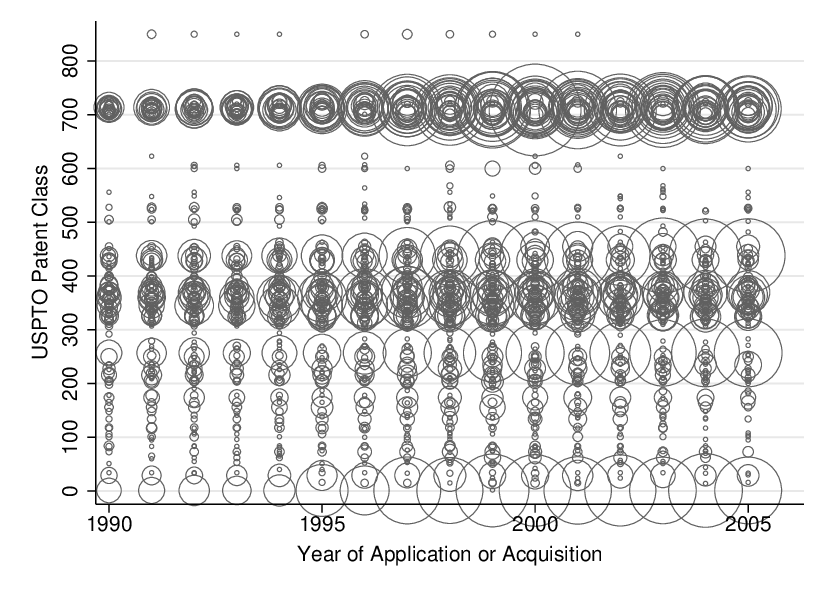}
\end{subfigure}
\caption*{\footnotesize {%
\textit{Note}: The circle size represents the number of patents in each class-year. Based on our method and analysis in sections 3 and 4, the ``flare lengths'' (our measure of the lengths of unique technological trajectories) of these firms' portfolios are: 3 (Cisco), 2 (Seagate), 1 (Pfizer), 2 (Medtronic), 4 (GE), and $\infty$ (IBM).}}%
\label{Figure - bubble (examples)}
\end{figure}%

The top panels show two IT firms. Cisco Systems makes network equipment (e.g., routers) and is famous for its active use of M\&As to acquire new products and talents; it acquired the largest number of target firms with patents in our sample. Nevertheless, most of Cisco's patents are obtained by in-house R\&D and are concentrated in classes 370 (multiplex communications) and 709 (electrical computers and digital processing systems: multicomputer data transferring). Seagate Technology makes hard disk drives (HDDs) and is another example of specialized IT firms. Its main patent class is 360 (dynamic magnetic information storage or retrieval), which is central to the HDD technology, but its portfolio gradually diversified as the firm intensified efforts to manufacture key components as well, including heads, media, and their interface.\footnote{See \cite{IgamiSubrahmanyam2019} for the details of patents and innovation in the HDD industry.}

The middle panels show two health care firms. The pharmaceutical industry is R\&D-intensive, but the patent portfolio of Pfizer looks simpler than the IT examples. Most of the drug patents are in classes 424 and 514 (drug, bio-affecting, and body treating compositions), and drug makers hardly patent elsewhere. By contrast, medical devices rely on a variety of technologies, even though their main classes are relatively few (600--607). The plot shows Medtronic, a leading medical-device maker, is active in many areas.

The bottom panels present extreme cases, for a reference. GE, a conglomerate, has one of the most diversified portfolios in our sample, with patents in more than 300 classes. The picture becomes too messy for human eyes to draw insights. Finally, IBM has by far the largest number of patents in our sample, but its portfolio looks more organized than GE's, because its activities are more focused. Most of the computers and electronics technologies are in the 300s and the early 700s, which are where IBM's portfolio is concentrated. 

These examples suggest the portfolio aspect of patents and technologies is interesting and contains potentially important information. However, the high dimensionality of technological space makes conventional data analysis difficult.

\paragraph{Comparison of HP, Dell, and Qualcomm}

Figure \ref{Figure - bubble (HP, Dell, Qualcomm)} shows the raw-data patterns for the three IT firms that we discuss in the first half of section 5.2.

HP and Dell are among the largest computer makers, and their main patent classes are similar, but their approaches to R\&D are different. HP is a traditional computer maker, whereas Dell's success is usually attributed to its unique business model in which the company sells directly to consumers and most of the manufacturing is outsourced to third-party suppliers in Asia. Such ``business-model innovations'' do not represent patentable inventions in most cases. Hence, patent statistics (and their topological representations) do not reflect Dell's ``uniqueness'' in this sense.

Qualcomm, a manufacturer of telecommunication chips, exemplifies this point with a unique portfolio (length 3) despite having relatively few patents and seemingly simple distribution across classes.

\begin{figure}[htb!!!!]
    \caption{Raw Data on Selected Technology Firms}%
    \begin{subfigure}{0.3\textwidth}
        \caption{Hewlett Packard}%
        \centering
        \includegraphics[width=\textwidth]{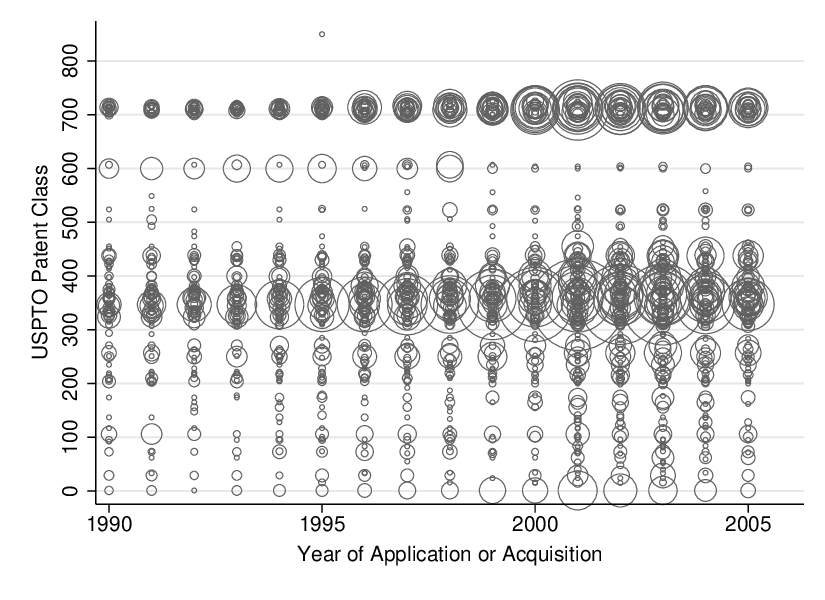}
    \end{subfigure}
    \hfill
    \begin{subfigure}{0.3\textwidth}
        \caption{Dell}%
        \centering
        \includegraphics[width=\textwidth]{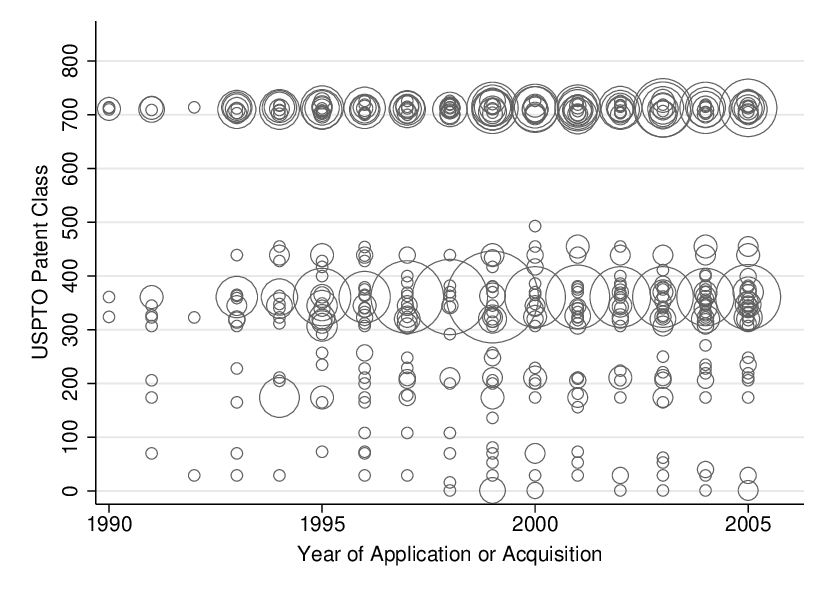}
    \end{subfigure}
    \hfill
    \begin{subfigure}{0.3\textwidth}
        \caption{Qualcomm}%
        \centering
        \includegraphics[width=\textwidth]{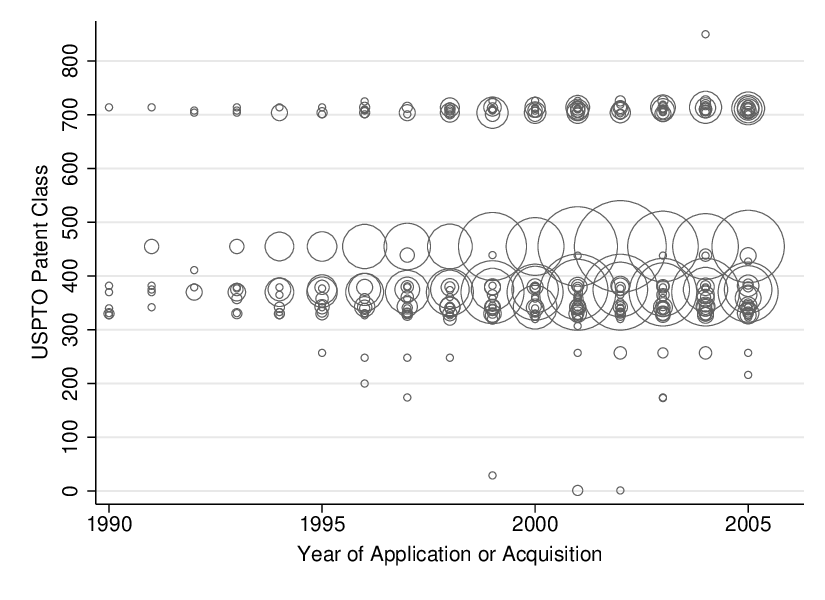}
    \end{subfigure}
    \caption*{\footnotesize {%
    \textit{Note}: The circle size represents the number of patents in each class-year. The flare lengths of these firms' portfolios are: 6 (HP), 1 (Dell), and 3 (Qualcomm).}}%
    \label{Figure - bubble (HP, Dell, Qualcomm)}
\end{figure}%

\section*{Appendix C \ Topological Data Analysis}

This section explains the idea of TDA, the Mapper algorithm, our specifications, and our original method for detecting and measuring flares.

\subsection*{C.1 \ Introduction to TDA}

Most data-analysis techniques in economics and elsewhere concern the evaluation of parameters or other quantities that characterize the system (the data-generating process, or DGP).\footnote{This and the next paragraphs borrow expositions from \cite{EpsteinCarlssonEdelsbrunner2011} and \cite{SizemoreEtAl2018}.} However, not all aspects of a system are readily summarized by numerical quantities. In particular, the “shape” of the data (i.e., the properties that remain invariant under “stretching” and “shrinking,” e.g., loops and branching patterns) could constitute a significant insight about real phenomena.

Shape is a somewhat nebulous concept and may appear too intuitive to define precisely and describe quantitatively, but the unique strength of TDA is its ability to capture and summarize such information in a useful, small representation of the data. Even though it is not among the usual tools for empirical economists, topology as an area of pure mathematics has existed for more than a century, and provides a theoretical foundation for the analysis of shapes. The adaptation of topological techniques to real data has been undertaken only recently (\cite{edelsbrunner2000topological}, \cite{zomorodian2005computing}, \cite{carlsson2009topology}, \cite{edelsbrunner2010computational}). Nevertheless, TDA has already been successfully applied to an increasing number of fields, including biology, chemistry, and materials science (e.g., \cite{nicolau2011topology}, \cite{hiraoka2016hierarchical}). See \cite{chazal2017introduction} for a brief introduction.

Among the techniques in TDA, the study of persistent homology has emerged as the most popular.\footnote{\cite{EpsteinCarlssonEdelsbrunner2011} explain the popularity of homology groups by pointing out that they offer an attractive combination of strong explanatory power, a clear intuitive meaning, and a low computational cost. Because the notion of shape within (finite) datasets is inevitably stochastic, and because homology is sensitive to noise in the data, \textit{persistent} homology is used to quantify the stability of geometric features with respect to perturbations, so that real phenomena could be distinguished from artifacts of noise.} However, its application to high-dimensional data is constrained by the computational cost of constructing combinatorial models (e.g., \v{C}ech complex, Alpha complex, Rips complex, etc.), which requires one to check higher-order intersections of the balls in that space and to store all the information. Various methods have been proposed to address this “curse of dimensionality,” but persistent homology can handle only tens of dimensions in the current state of the art. By contrast, Mapper can easily handle thousands and even millions of dimensions, by focusing on the global topology of the data and providing simplified representations of their shape via nonlinear transformations.\footnote{For example, \cite{rizvi2017single} use Mapper to study single-cell gene expression, where the number of dimensions equals the number of expressed genes (up to 10,000).} Thus, whereas persistent homology offers a fine-grained characterization of cavities in relatively low-dimensional data, Mapper enables a relatively coarse characterization of \textit{very} high-dimensional data, which makes it particularly suitable for our empirical context.

Since \cite{singh2007topological} introduced Mapper, it has been applied to study an RNA folding pathway \parencite{yao2009topological}, the DNA microarray data of breast cancer \parencite{nicolau2011topology}, cellular differentiation and development \parencite{rizvi2017single}, and the organization of whole-brain activity maps \parencite{saggar2018towards}. Methodologically, \cite{lum2013extracting} is the most closely related work to ours, because they also propose a flare-detection algorithm. Their method uses global graph-theoretic properties that are applicable to any graph, without using any additional information from the Mapper algorithm.\footnote{Specifically, their flare detection algorithm uses the $0$-dimensional persistent homology \parencite{edelsbrunner2000topological} of the graph filtered by an eccentricity measure on its nodes. An eccentricity measure tends to give a higher value to nodes that are ``eccentric'' (on tips of flares) compared with central nodes (on the trunks).}
By contrast, our algorithm takes advantage of particularities of our Mapper graph, where each node is a set of firm-years. We ensure each flare that we identify is associated with a specific firm. Hence, it can be interpreted as \textit{a flare of that firm}.

\subsection*{C.2 \ Illustrated Example of the Mapper Procedure}
\label{sec:mapper_algorithm}

Figure \ref{fig:mapper_proc} illustrates the Mapper procedure with a simple example. Let us start with data $L$ given by the points in two-dimensional space. Our goal is to obtain a simplified representation of $L$ while preserving its topological features, such as holes and branches. In step 1, we project $L$ onto the horizontal axis (i.e., $d=1$). This operation reduces the dimensionality of the data by eliminating the second dimension (i.e., information on the vertical axis in this case). In step 2, we cover these points on the horizontal axis by four equal-sized intervals (i.e., cover elements) $C_1, C_2, C_3$, and $C_4$ (i.e., $J=4$) with overlaps.\footnote{The degree of overlap is approximately 20\% in the pictured example.} In step 3, we look at each interval $C_j$, and cluster adjacent points \textit{in the original data space} with two dimensions. In step 4, we represent these clusters by nodes, and connect them with edges whenever adjacent clusters share the same points within their overlapping regions.

\begin{figure}[htb!!!!]\centering%
\caption{Illustration of the Mapper Procedure}
\includegraphics[width=0.95\textwidth]{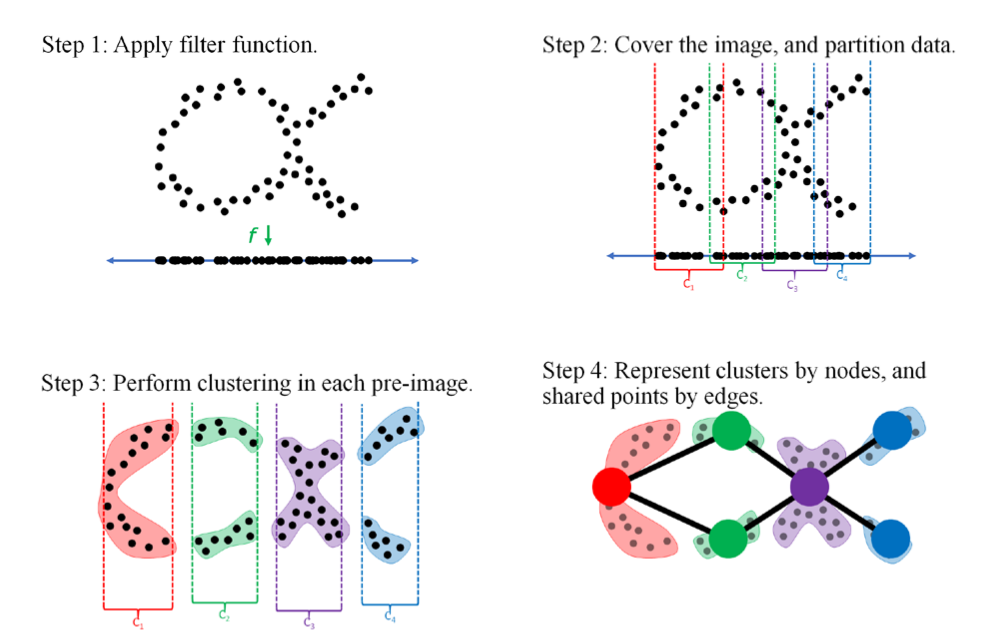}
\label{fig:mapper_proc}
\end{figure}%

The resulting graph is much simpler than the original data and amenable to graph-theoretic analyses, but it still preserves the ``global structure'' of $L$ (i.e., topological features that span multiple local regions, such as loops and long branches/flares). By contrast, using conventional techniques for dimensionality reduction alone would be similar to performing only step 1. Likewise, directly performing clustering in the original data would be the same as skipping steps 1 and 2, which would probably generate a single big cluster for the entire data in this case. Neither approach would be able to recover the \textit{shape} of the data (i.e., a collection of global structures). For this particular example, the usefulness of the Mapper graph is limited, as the original data itself is only two-dimensional and can be readily visualized. However, for more complicated high-dimensional data, a simplified graph representation offers a helpful visual aid. 

One way to interpret the Mapper procedure is to view it as a kind of local clustering together with ``global reconstruction'' (i.e., replication of global structures). The choice of the filter function and cover determines the local regions $f^{-1}(C_j) \subset L$ of the data. Then, the clustering algorithm is applied only locally, to each local region. The construction of the graph $G$ recovers some of the global information by connecting nodes (each of which is a cluster of points in $L$) whenever they share points in the original data.

\subsection*{C.3 \ Definitions, Proofs, and Computation}

This section formally presents the definitions, proofs, and computational methods for section 5.1.

\paragraph{Preparation.}

Let us review some basic concepts from graph theory. In general, a \emph{graph} $G=(V,E)$ is a set $V$ of nodes (vertices) and a set $E$ of edges. We assume that each edge $e \in E$ of $G$ is assigned the weight $w(e) = 1$.\footnote{The theory can be extended to handle positive weights $w(e) > 0$ that are different across edges.} For $u, v \in G$, the \emph{length} $\ell(p)$ of a path $p$ from $u$ to $v$ is the sum of the weights of the edges of $p$. The \emph{distance} $d_G(u,v)$ between $u$ and $v$ is the minimum length of all paths $p$ in $G$ from $u$ to $v$. For simplicity, we write $d(u,v)$ for $d_G(u,v)$.

For a graph $G$ and a subset $V'$ of the nodes of $G$, the \emph{full subgraph} of $G$ with nodes $V'$, denoted by $G[V']$, is the graph with the set of nodes $V'$ and edges consisting of all edges of $G$ whose endpoints are both in $V'$. It is the maximal subgraph of $G$ with set of nodes $V'$.

\begin{definition}[Ball]
  Let $r \in \mathbb{R}$ and $u \in G$. The (closed) ball $B_r(u)$ in $G$ is
  \[
    B_r(u) = G[\{v \in G \suchthat d(u,v) \leq r\}].
  \]
  In words, it is the full subgraph of $G$ of all nodes at most distance $r$ from $u$. 
\end{definition}
 % For ease of notation, we drop the curly brackets in denoting full subgraphs and write: $B_r(u) = G[v \in G \suchthat d(u,v) \leq r]$ instead.
% More generally, for any subset $V'$ of vertices in $G$, we define
% \[
%   B_r(V') = G[v \in G \suchthat d(u,v) \leq r \text{ for some } u \in V'].
% \]

Now, consider a Mapper graph $G=(V,E)$ of our data. From the construction of the Mapper graph, each node $v \in V$ will consist of points (firm-years) of the form $l_{i,t}$. To simplify, we adopt the following notation, because we want to consider firms and not firm-years for the analysis.
\begin{notation}
  In the setting above, firm $i$ is said to be in node $v$, or, equivalently, $v$ contains firm $i$ if node $v$ contains an observation of firm $i$ at some time $t$, that is, $l_{i,t} \in v$ for some $t$. In this situation, we write $i \in v$.
\end{notation}

For each firm $i$, we want to determine whether $i$ appears as a flare in $G$. One way to extract flares is to use global graph-theoretic properties of $G$, as in the method proposed in \cite{lum2013extracting} using $0$-persistence of eccentricity (or centrality). Instead, we start with the requirement that we only consider a structure to be a ``flare of $i$'' if each node in the flare contains $i$. This way, we focus on a smaller graph $G_i$ defined below, which contains only nodes that involve $i$, and look for flares therein.\footnote{More generally, one may consider a flare that involves multiple firms. We restrict our attention to single-firm flares in this paper because they are the most salient feature of our Mapper graphs.} We see later that this perspective simplifies computations.
\begin{definition}[Induced subgraph $G_i$ of firm $i$]
  Let $i$ be a firm.
  Define $G_i$ to be
  \[
    G_i = G[\{v \in G \suchthat i \in v \}].
  \]  
\end{definition}
That is, $G_i$ is the full subgraph of $G$ formed by nodes that contain firm $i$.  
We decompose the nodes of $G_i$ into ``interior'' and ``boundary.''
\begin{definition}[Interior and boundary of  $G_i$]
  \leavevmode
  \begin{enumerate}
  \item The \emph{interior} $F_i$ of $i$ in $G$ is defined to be
    $
    F_i = G[\{v \in G_i \suchthat B_1(v) \subseteq G_i\}].
    $
  \item The \emph{boundary} of $i$ in $G$ is $G_i \setminus F_i$.
  \end{enumerate}  
\end{definition}
In words, the interior $F_i$ contains all nodes $v$ of $G_i$ such that $G_i$ contains all neighbors of $v$ (i.e., the ball of radius $1$ around $v$).
Lemma~\ref{lem:pathexit} shows that the boundary $G_i\setminus F_i$ indeed serves as a ``boundary'' for $F_i$: to get outside of $G_i$, one always needs to go through the boundary.

Figure \ref{fig:interior_boundary} illustrates the definitions of interior and boundary. The pink region represents firm $i$'s subgraph $G_i$, the green nodes are in the interior $F_i$, and the purple nodes are in the boundary $G_i\setminus F_i$.

\begin{figure}[htb!!!!]\centering%
 \caption{Interior and Boundary}
 \includegraphics[width=0.4\textwidth]{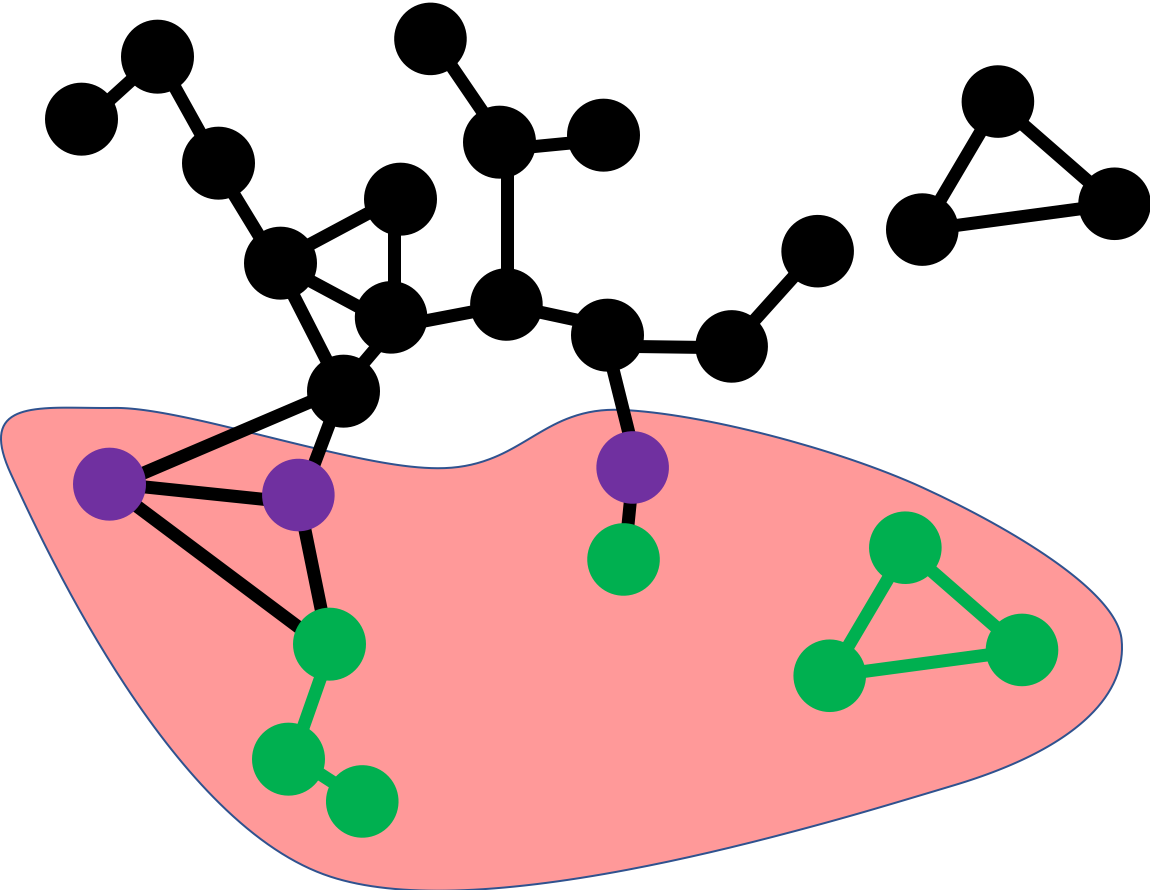}
 \label{fig:interior_boundary}
\end{figure}%

% We can be easily check that
% $
% F_i = G_i[v \in G_i \suchthat B_1(v) \subseteq G_i].
% $
% Furthermore,

Next, let us define flares and islands in graph-theoretic terms.
\begin{definition}[Flares and Islands]
  A connected component $R$ of the interior $F_i$ of firm $i$ is said to be an \emph{island of firm $i$} if
  $R$ is also a connected component of $G$, and said to be a \emph{flare of firm $i$}, otherwise.
\end{definition}
For example, two flares and one island (the triangle on the right) exist in Figure \ref{fig:interior_boundary}. In the following, we refine these notions using numerical indices. As defined above, a flare may not always ``look like'' what one may imagine to be a flare.

\paragraph{Measuring Flares.}

We introduce the following definition and proposition, which serve as the foundations for defining our concept of flare length.
\begin{definition}[Exit distance]
  \label{defn:exit}
  Let $u \in F_i$ be a node in the interior of firm $i$. The exit distance of $u$ in $F_i$ is
  \[
    e_i(u) = \min\{d(u,w) \suchthat {w\in G\setminus F_i}\}.
  \]
  In the case in which no path exists from $u$ to any $w \in G\setminus F_i$, we put $e_i(u) = \infty$.
\end{definition}

\begin{restatable}{proposition}{thmexit}
  \label{thm:exit}
  Let $u \in F_i$. Then,
  \[
    e_i(u) = \min\{d_{G_i}(u,v) \suchthat {v \in G_i\setminus F_i}\},
  \]
   where $d_{G_i}(u,v)$ is the distance between $u$ and $v$ in $G_i$.
\end{restatable}

To prove Proposition~\ref{thm:exit}, we first show the boundary $G_i\setminus F_i$ indeed serves as a ``boundary'' for $F_i$: to get outside of $G_i$, one always needs to go through the boundary.
\begin{lemma}
  \label{lem:pathexit}
  Let $u \in F_i$ and $w \in G\setminus G_i$, and let $p$ be a path from $u$ to $w$. Then, the path $p$ passes through some node $v \in G_i \setminus F_i$.
\end{lemma}
\begin{proof}
  Let $p$ be such a path from $u\in F_i$ to $w\in G\setminus G_i$, which passes through the nodes
  \[
    u=v_0, v_1, v_2,\hdots v_{n-1}, v_n = w
  \]
  in that order.

  Suppose, to the contrary, that all $v_j$ are not in the boundary $G_i\setminus F_i$. We show by induction that $v_j \in F_i$ for all $j \in \{0,\hdots,n\}$. First, $v_0=u \in F_i$ is clear. Suppose $v_j \in F_i$. Because $v_{j+1} \in B_1(v_j) \subseteq G_i$ by definition of the interior $F_i$, and because $v_{j+1} \notin G_i\setminus F_i$ by assumption, we see $v_{j+1} \in F_i$. Thus, by induction, $v_j \in F_i$ for all $j \in \{0,\hdots,n\}$. In particular, $v_n = w \in F_i$, which is a contradiction, because $w \in G\setminus G_i \subseteq G\setminus F_i$.
  
  Therefore, some $v_j$ exists in the boundary $G_i \setminus F_i$.
\end{proof}

Now we prove Proposition~\ref{thm:exit}.

\begin{proof}
  It is clear that
  \[
    \min\{d(u,w) \suchthat {w\in G\setminus F_i}\} \leq \min\{d_{G_i}(u,v) \suchthat {v\in G_i\setminus F_i}\}.
  \]
  Suppose the minimum of the left-hand side is achieved by a $w \in G\setminus F_i$, and let $d(u,w) = \ell(p)$, the length of a minimum path $p$ in $G$ from $u \in F_i$ to $w \in G\setminus F_i$. Let $v$ be the first node $v \in G_i \setminus F_i$ that $p$ passes through. Note such $v$ exists by Lemma~\ref{lem:pathexit}.

  In the case in which $v \neq w$, truncate $p$ to the path $p'$ from $u$ to $v$. By choice of $v$, $p'$ is fully contained in $G_i$, and $\ell(p') < \ell(p)$ because we only have positive weights and $p'$ has strictly fewer edges than $p$. 
  It follows that
  \[
    \min\{d(u,w) \suchthat {w\in G\setminus F_i}\} = \ell(p) > \ell(p') \geq \min\{d_{G_i}(u,v) \suchthat {v\in G_i\setminus F_i}\},
  \]
  because $p'$ is a path from $u$ to $v$ that is contained in $G_i$. This is a contradiction.

  Thus, $v=w$, and it follows that
  \[
    \min\{d(u,w) \suchthat {w\in G\setminus F_i}\} = \ell(p) \geq \min\{d_{G_i}(u,v) \suchthat {v\in G_i\setminus F_i}\},
  \]
  which shows the required equality.
\end{proof}

Using Proposition~\ref{thm:exit}, we can compute $e_i(u)$ using only the information of $G_i$, because the distance $d_{G_i}(u,v)$ is the minimum length of all \ul{paths in $G_i$} from $u$ to $v$. By contrast, directly using Definition~\ref{defn:exit} would necessitate the computation of $d(u,w)$, the  minimum length of all \ul{paths in $G$} from $u$ to $w$.

We use the exit distance $e_i(u)$ to refine our notion of flares.
\begin{definition}[Flare index]
  For a connected component $R$ of $F_i$ (a flare or island of firm $i$), the \emph{flare index} of $R$ is defined to be 
  \[
    \lambda_i(R) = \max_{u\in R} e_i(u).
  \]  
\end{definition}

We immediately obtain the following characterization of islands using $\lambda_i$.
\begin{lemma}
  Let $R$ be a connected component of $F_i$.
  Then, $\lambda_i(R) = \infty$ if and only if $R$ is an island of firm $i$.
\end{lemma}
\begin{proof}
  Immediate from the definitions.
\end{proof}

Finally, to aggregate all the information, we define flare signature.
\begin{definition}[Flare signature]
  Let $F_i = R_1 \sqcup R_2 \sqcup \hdots \sqcup R_S$ be a decomposition of $F_i$ into its connected components. The \emph{flare signature} of $i$ is the multiset
  \[
    \vec{\lambda}_i = \{\{\lambda_i(R_s) \suchthat s =1,\hdots,S\}\}.
  \]
  Note that if $F_i$ is empty, we simply put the empty multiset as the flare signature of $i$.  
\end{definition}

We link the flare signature to the following ``types.''
\begin{enumerate}
\item \textbf{$\vec{\lambda}_i$ is empty.} This case occurs if and only if $F_i = \emptyset$, meaning every node containing firm $i$ neighbors at least one node not containing $i$.  We call this case \textbf{Type 0: no flare or island}.
\item \textbf{$\vec{\lambda}_i$ contains only finite elements}. In this case, each connected component $R$ of $F_i$ is connected to some point $w \in G \setminus F_i$, meaning each $R$ itself cannot be a connected component of $G$. Thus, each $R$ is not an island; it is a flare. We call this case \textbf{Type 1: flares only}.
\item \textbf{$\vec{\lambda}_i$ contains finite elements, and some copies of $\infty$}. This case corresponds to \textbf{Type 2: flares and islands}.
\item \textbf{$\vec{\lambda}_i$ contains only copies of $\infty$}. This case corresponds to \textbf{Type 3: islands only}.
\end{enumerate}

The flare signature is defined as a multiset of flare indices. Sometimes, having one number describing how much firm $i$ looks like a flare in the Mapper graph may be convenient. Thus, we define the following.

\begin{definition}[{Flare length}]
  The \emph{flare length} (or just \emph{length}, for short) of firm $i$ is
  \[
    \lambda_i = \left\{
      \begin{array}{ll}
        0 & \text{if } \vec{\lambda}_i \text{ is empty,}\\
        \mathop{\mathrm{finmax}}(\vec{\lambda}_i) & \text{if } \vec{\lambda}_i \text{ has at least one finite element,}\\
        \infty & \text{otherwise,}
      \end{array}
    \right.        
  \]
  where $\mathop{\mathrm{finmax}}(\vec{\lambda}_i)$ is the maximum among all finite elements of $\vec{\lambda}_i$.
\end{definition}
Type 0 gets {flare length} $0$, type 3 is sent to index $\infty$, and types 1 and 2 occupy the range in between, where the {flare length} of a firm is determined by the ``longest'' flare of firm $i$. 

\paragraph{Computation of Flare Signatures.}
Let $G=(V,E)$ be the Mapper graph of our data $L$.
For each firm $i$, the computation of the subgraph $G_i$ involving $i$ can be done by iterating through all nodes $v \in V$ and checking membership of firm $i$ in $v$. The interior-boundary decomposition of $G_i$ can be computed by considering the boundary first. For each $v \in G_i$, we simply check if $v$ has a neighbor that is not in $G_i$; if so, $v$ is part of the boundary $G_i\setminus F_i$. The nodes of $G_i$ not in the boundary are then automatically part of the interior.

Next, let us consider the computation of the flare signature $\vec{\lambda}_i$ of firm $i$.
First, we need a decomposition of $F_i$ into its connected components:
\[
  F_i = R_1 \sqcup R_2 \sqcup \hdots \sqcup R_S,
\]
which can be done, for example, via a breadth-first search.
%, with a time complexity $\bigO(|V|+|E|)$.
For each connected component $R_{s}$ of $F_i$, its flare index is given by
\[
  \lambda_i(R_{s}) = \max_{u\in R_{s}} e_i(u).
\]
Because we need to do the same for each connected component $R_{s}$ of $F_i$, we compute $e_i(u)$ for all $u \in F_i$. By Proposition~\ref{thm:exit}, the exit distance is
\[
  e_i(u) = \min\{d_{G_i}(u,v) \suchthat {v \in G_i\setminus F_i}\},
\]
which can be computed using a multi-source version of Dijkstra's shortest-path algorithm, with sources $G_i\setminus F_i$.

\subsection*{C.4 \ How the Time Dimension Interacts with Flares}

One might wonder how our definition of flare length---which does not explicitly incorporate the time dimension---may (or may not) capture product market competition---which occurs (only) among firms that are located in the same market \textit{in the same period}.

First, note both current and \textit{past} locations of rivals affect the focal firm's profit in our model. The dynamics of the potential demand in equation \ref{eq - market size, transition} imply the ``existing'' markets (i.e., those in which some firms have previously operated) are less profitable than new markets. Because of this intertemporal linkage, entering the market in the trails of other firms is not a particularly attractive strategy even if it is not currently populated by rivals. Therefore, measuring flares based on the entire graph $G$ (i.e., without distinguishing time periods) makes sense from the perspective of ``followers.''

Second, from the perspective of the pioneering firm, its profit at the time of entry into a new market is not affected by whether some other firms follow its footsteps in later years. Hence, in principle, our measurement of flares should ignore the fact that its trajectory intersects with those of the latecomers. In practice, this issue does not seem to affect our measures because most of our firm-year observations are either surrounded by many contemporaneous rivals or breaking into new regions in unique trajectories, as the case studies in section 4.3 show. Based on these empirical patterns, our ``timeless'' measurement of flares provides a reasonable approximation to the firms' unique technological trajectories.

%%% Local Variables:
%%% mode: latex
%%% TeX-master: "TDAPatent"
%%% End:

\section*{Appendix D \ Comparison with Jaffe-style Clustering}

This section supplements the comparison of Mapper with Jaffe (1989). We explain their methodological similarities and differences, as well as present an alternative map of technological space based on his data-transformation convention. 

\paragraph{\protect Methodological Differences.} Whereas section 4.4 highlights the differences in results, Table \ref{Table - Jaffe comparison} clarifies two methodological differences.

\begin{table}[tbh]
\caption{Comparison with Jaffe (1989)}
\begin{center}
\fontsize{9pt}{11pt}\selectfont%
\begin{tabular}{lccc}
\hline \hline
Procedures & Ours & Ours  & Jaffe (1989) \\ 
& (main text) & (this Appendix) &  \\ \hline
1. Re-scaling & Log & Share & Share \\ 
2. Distance metric & Cosine & Cosine & Cosine \\ 
3. Clustering & Local & Local & Global \\ 
4. Reconstruction & Edges & Edges & None \\
5. Final output & Graph & Graph & Clusters \\ \hline \hline
\end{tabular}
\begin{minipage}{450pt}
{\fontsize{9pt}{9pt}\selectfont \smallskip  \textit{Note}: One can use other distance metrics in the Mapper procedures, including Euclidean, correlation, min-complement, and Mahalanobis. See various sensitivity analyses in Appendix E.}
\end{minipage}
\end{center}
\label{Table - Jaffe comparison}
\end{table}

First, we take a logarithm of patent count, $l_{i,t,c}=\ln(\tilde{p}_{i,t,c}+1)$, whereas he takes a share of each class within a firm-year, $l_{i,t,c}=\frac{\tilde{p}_{i,t,c}}{\sum_{c} \tilde{p}_{i,t,c}}$. These rescaling protocols transform the metric space itself and lead to significant differences in the outputs. Hence, how one pre-processes raw data is an important, substantive choice. Nevertheless, this difference is secondary in terms of methodology, because it is a matter of data pre-processing rather than the analytical procedure itself. As we demonstrate in this section, we can easily switch to Jaffe's share-based measure while sticking to our overall framework.

The second and more important difference is that Jaffe performs clustering at the global level to generate a list of mutually exclusive clusters of firms, whereas our “clusters” are local and retain connections through edges between them (which reflect the existence of commonly shared members). In other words, his algorithm is a big \textit{discretization} operation, whereas ours is designed to recover the \textit{continuum} of firms and industries in the data. Uncovering the original, continuous data patterns is important because industry boundaries could be fluid especially when innovative activities are concerned. In the following, we demonstrate how our method can help reveal the global shape of the data and generate additional insights beyond what Jaffe-style clustering does.

\begin{figure}[htb!!!!]\centering%
\caption{Mapper Graph Based on Jaffe's Measure}%
\includegraphics[width=0.60\textwidth]{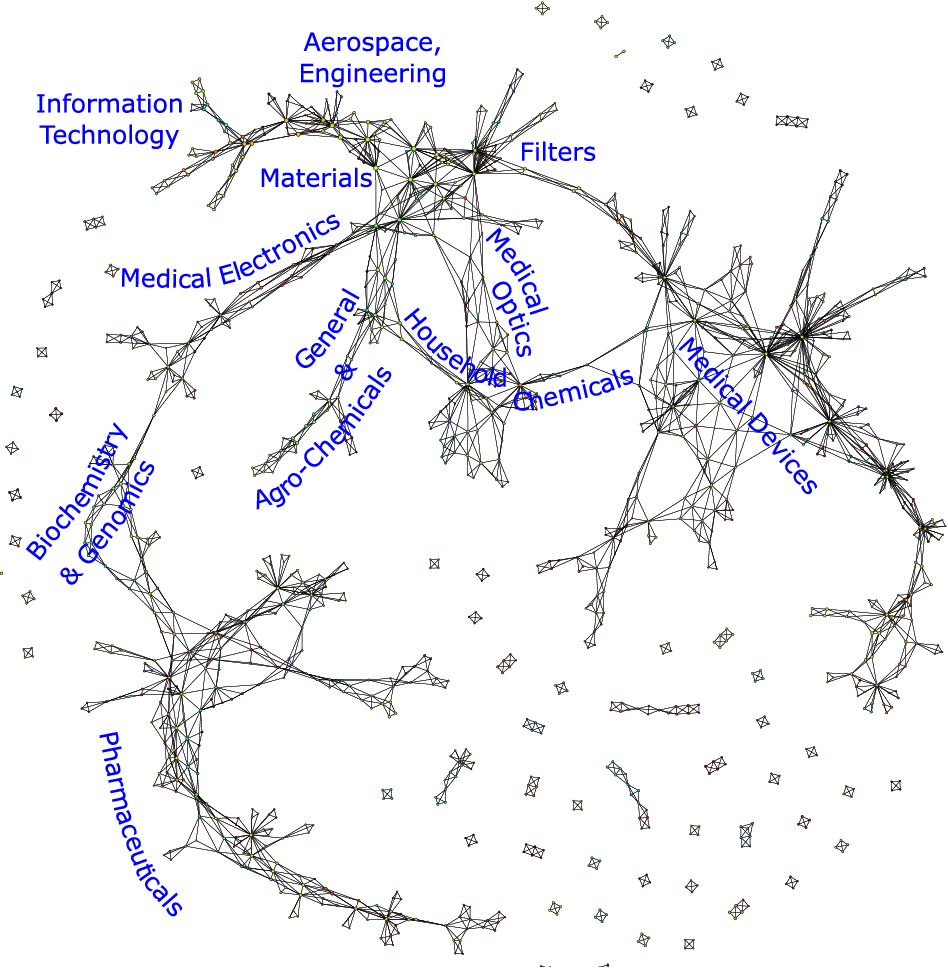}
\caption*{\footnotesize {%
\textit{Note}: Node colors represent the average year of the firm-years in that cluster, with earlier years in blue and later years in red. This figure is a shape-graph representation of 333 major firms' R\&D patents in 1976--2005 based on shares, cosine distance, $n=40$, and $o=0.5$.}}%
\label{Figure - mapper(cos_sumone_n40)}
\end{figure}%

\paragraph{\protect Mapper Graph Based on Jaffe's Measure.} Our approach preserves the underlying continuity in the data. Figure \ref{Figure - mapper(cos_sumone_n40)} is the Mapper graph of the same data, based on Jaffe's rescaling protocol (percentage shares) and distance metric. Unlike the 21 mutually exclusive groups from the global clustering method (Table \ref{Table - K-Medoids Clustering} in section 4.4), the shape graph recovers a \textit{continuum of industries} from the data. Indeed, its main insight is that industries are connected, sometimes in unanticipated ways. 

\paragraph{\protect ``Shrinking'' High-Tech Industries.} Many firms populate the upper-north-west corner of the graph. This high-tech region is so densely populated that disentangling it is difficult (see Figure \ref{Figure - mapper(cos_sumone_n40_details)}, panel a). These firms conduct R\&D in relatively many patent classes. Raw patent counts (and their logged version in section 4.3) preserve the uniqueness of each firm’s portfolio. However, after their conversion into percentage shares (and hence the loss of information on volumes in absolute terms), most portfolios end up looking alike. Thus, the non-share-based Mapper graphs of section 4.3 seem more informative about high-tech industries.

\paragraph{\protect Biomedical Super Flare.} By contrast, the share-based Mapper graph maps biomedical areas more clearly and reveals interesting technological connections between industries. Pharmaceutical companies live in their own world (in the south-west corner of Figure \ref{Figure - mapper(cos_sumone_n40)}), patenting only in a few drug-related classes. Nevertheless, they are not completely isolated, because biochemistry and medical electronics firms stretch from the northern ``heartland'' of engineering, materials, and general chemicals. The detailed maps in Figure \ref{Figure - mapper(cos_sumone_n40_details)} (panels a and b) show medical-equipment manufacturers (e.g., Perkin Elmer and Beckman Coulter) and genomics-based drug developers (e.g., Amgen and Genzyme) connect with pharmaceutical companies (e.g., Merck and Pfizer), collectively forming a long ``archipelago'' of biomedical industries. These connections are intuitive because genomics firms rely on measurement and data processing to develop new drugs. Uncovering them from Table \ref{Table - K-Medoids Clustering} alone would be difficult because it classifies general and agro-chemicals in cluster 6 and biochemicals and medical electronics in cluster 15.\footnote{Both clusters prominently feature Monsanto as a member, but its unique trajectory does not conform to the patterns of any other firms in either cluster (except Bayer, which acquired it in 2018). Figure \ref{Figure - mapper(cos_sumone_n40_details)} shows Bayer did not move much throughout the sample period, whereas Monsanto made a long trip from the crowded center of materials and chemicals industries to Bayer's location. The fact that Bayer acquired Monsanto in 2018 might suggest patent portfolios are a useful predictor of competitive positions and mergers. See \cite{EC2017}.}

\begin{figure}[htb!!!!]
\caption{Mapper Graph Based on Jaffe's Measure (Details)}%

\begin{subfigure}{1\textwidth}
\centering
\includegraphics[width=0.9\linewidth]{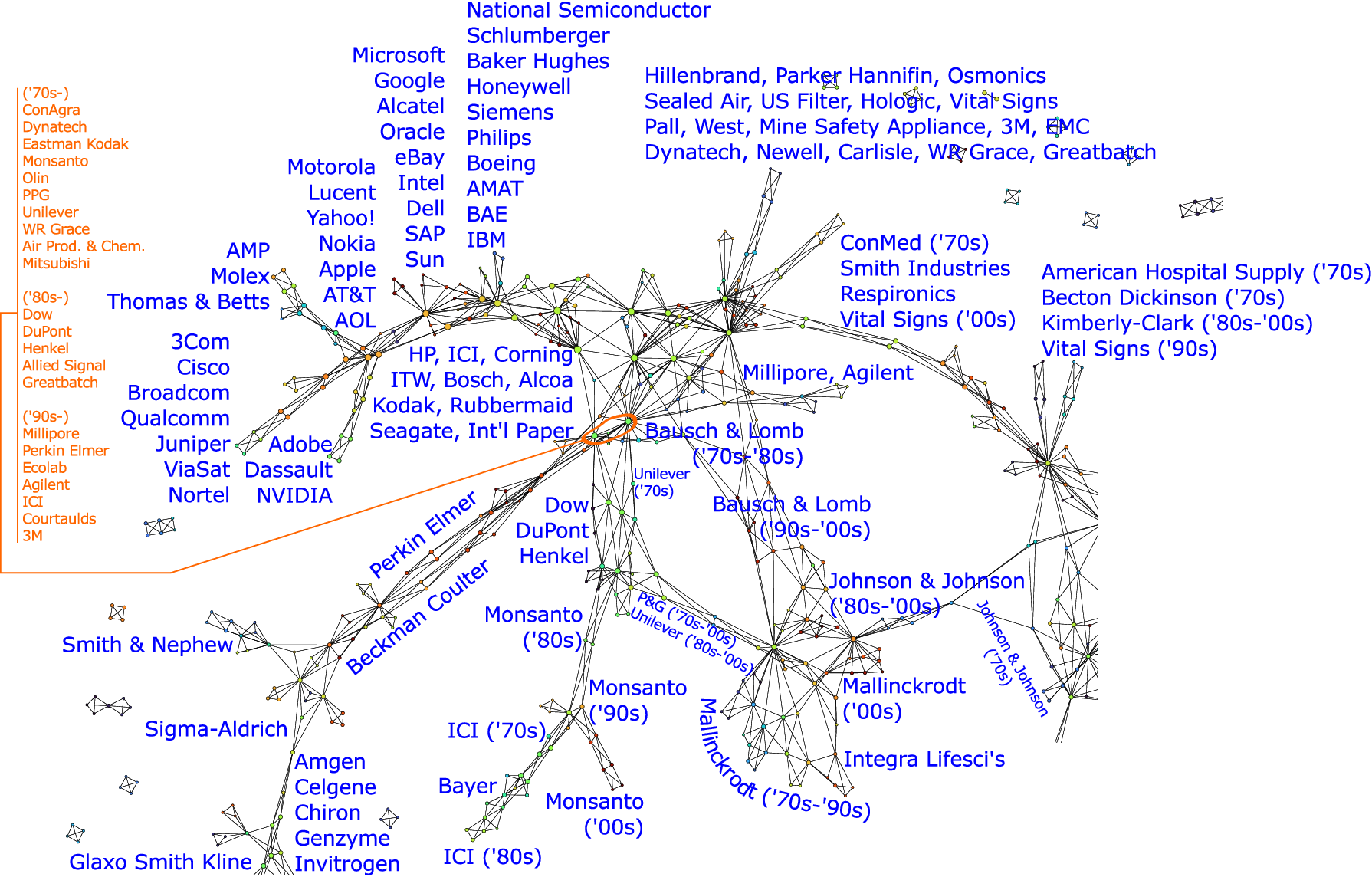}
\caption{IT, Engineering, Materials, and Chemicals}%
\end{subfigure}

\begin{subfigure}{0.5\textwidth}
\centering
\includegraphics[width=0.9\linewidth]{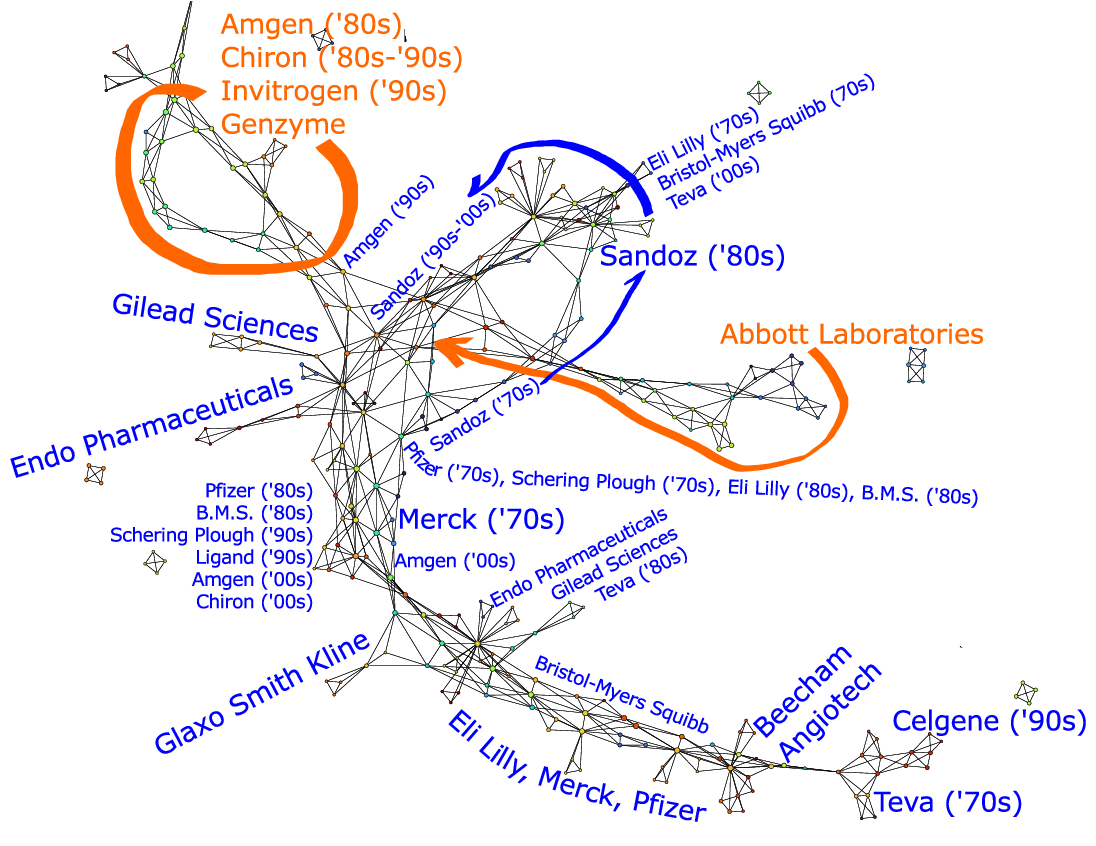}
\caption{Biomedicals and Pharmaceuticals}%
\end{subfigure}
\begin{subfigure}{0.5\textwidth}
\centering
\includegraphics[width=0.9\linewidth]{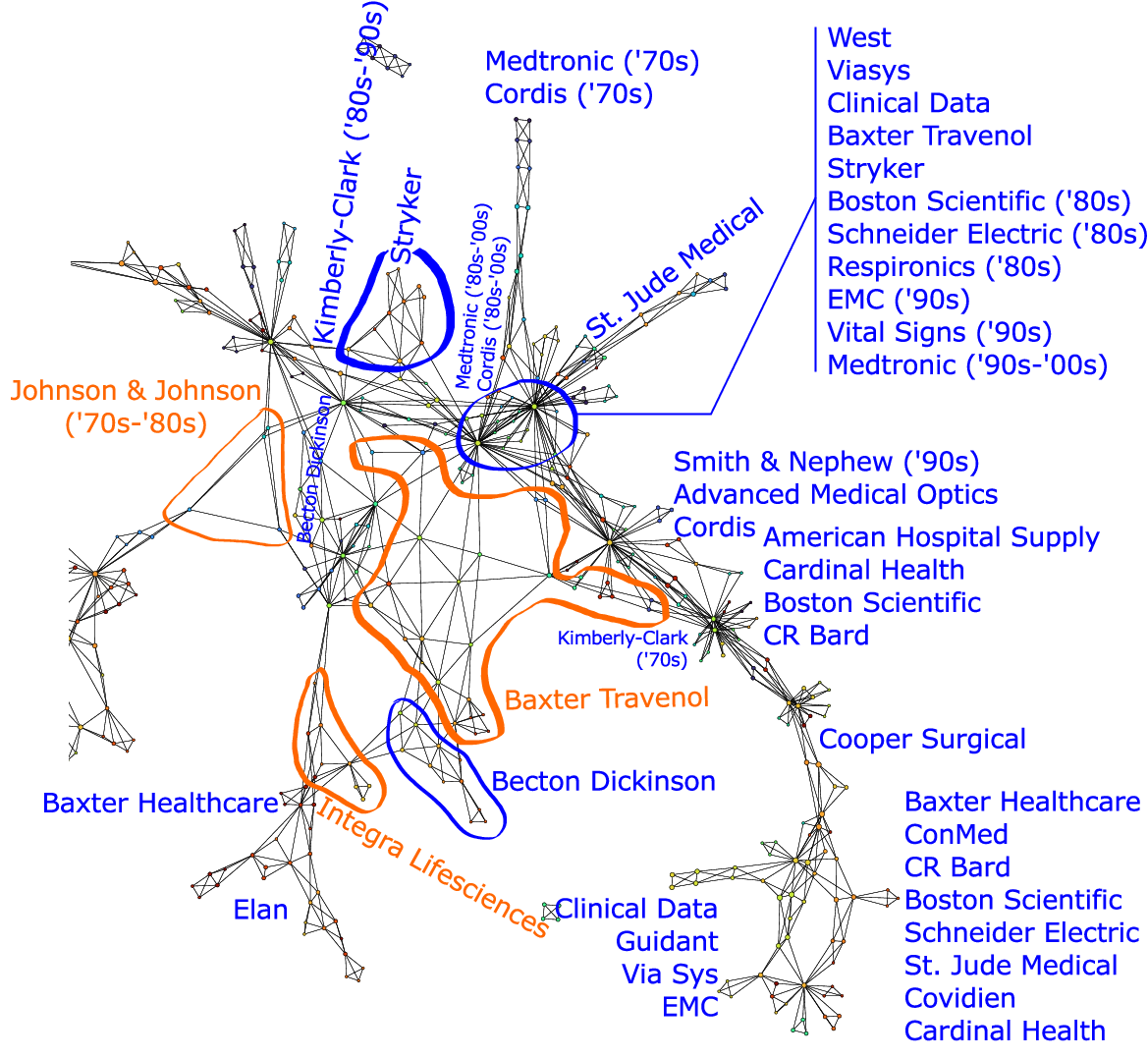}
\caption{Medical Devices}%
\end{subfigure}
\caption*{\footnotesize {%
\textit{Note}: These figures are enlarged and more detailed versions of the Mapper graph in Figure \ref{Figure - mapper(cos_sumone_n40)}.}}%
\label{Figure - mapper(cos_sumone_n40_details)}
\end{figure}%

\paragraph{\protect Two Bridges to Medical Devices.} Medical-device manufacturers occupy a large territory in the eastern half of Figure \ref{Figure - mapper(cos_sumone_n40)}. The Mapper graph reveals somewhat surprising ways in which this industry connects with others. Specifically, two types of firms bridge between medical devices and the engineering heartland.

One bridge consists of household chemicals and contact lenses. Figure \ref{Figure - mapper(cos_sumone_n40_details)} (panels a and c) shows household names, such as Unilever, P\&G, and Bausch \& Lomb, were close to the center of materials and general chemicals in the 1970s and the 1980s. But then their R\&D efforts moved in the south-east direction to form their own peninsulas by the 1990s and the 2000s. J\&J has a major health care division and bridges between household chemicals and medical devices.

The other bridge is located in the north and builds on dense clusters of less well-known firms specializing in aerodynamics and filters (e.g., Sealed Air, U.S. Filter, and Mine Safety Appliance). It then extends in the south-east direction and connects with more obviously medical-device-related names, such as Respironics and Vital Signs. The two groups of firms are seemingly unrelated at first glace, but their underlying technologies are common: breathing requires clean air, and the monitoring of vital signs concerns fluid dynamics. Thus, technologically speaking, mine safety and medical devices are closer neighbors than what a conventional industry-classification system would suggest. By contrast, the global clustering in Table \ref{Table - K-Medoids Clustering} is not particularly informative about these connections:
P\&G and J\&J appear in cluster 4; the aerodynamics-and-filters firms appear separately in cluster 8; and medical devices are split into clusters 4 and 20.

\clearpage

\paragraph{\protect K-Means Clustering.} Whereas Jaffe (1989) uses k-means clustering, we use its variant, k-medoids clustering, in section 4.4. Table \ref{Table - K-means clustering} shows K-means clustering of our data leads to an extreme result in which a single cluster contains more than 70\% of all firm-years, because so many firm-years are located in the densely populated neighborhood of electronics and engineering (i.e., the lower middle part of Figure \ref{Figure - mapper(cos_log_details_1)}).

\begin{table}[tbh]
\caption{K-Means Clustering}
\begin{center}
\fontsize{9pt}{11pt}\selectfont%
\begin{tabular}{lccc}
\hline \hline
Cluster & Number of & Number of & Representative \\ 
& firm-years & unique firms & firms \\ \hline
1 & 5,118 & 303 & (Too many firms to list) \\ 
2 & 438 & 47 & Tellabs, 3Com, Ericsson, Qualcomm, Broadcom \\ 
3 & 421 & 35 & Baxter Travenol, Cordis, C.R.Bard, Medtronic, St. Jude Medical
\\ 
4 & 154 & 19 & Amgen, Chiron, Celgene, Genzyme, Invitrogen \\ 
5 & 53 & 8 & BAE Systems, Trimble Navigation, Lockheed Martin \\ 
6 & 44 & 5 & Leggett \& Platt, Hillenbrand, Stryker \\ 
7 & 36 & 9 & FLIR Systems, Veeco Instruments, Titan, Lockheed Martin \\ 
8 & 29 & 4 & Federal Signal, Zero \\ 
9 & 18 & 4 & Morgan Crucible, Solectron, Emhart \\ 
10 & 18 & 3 & Veeco Instruments, Power-One \\ 
11 & 15 & 4 & RPM, Cookson \\ 
12 & 15 & 4 & Roper Industries, Varian \\ 
13 & 15 & 3 & Newell, Carlisle, Avant! \\ 
14 & 12 & 2 & SPS Technologies, Carpenter Technology \\ 
15 & 12 & 1 & Zebra Technologies \\ 
16 & 10 & 1 & Verifone Systems \\ 
17 & 9 & 2 & Carpenter Technology, Lucent \\ 
18 & 8 & 3 & Magne Tek, Franklin Electric \\ 
19 & 5 & 1 & Roper Industries \\ 
20 & 3 & 2 & Terex, Meggitt \\ 
21 & 1 & 1 & Datum \\ 
Total & 6,434 & 461 &  \\ \hline \hline
\end{tabular}
\begin{minipage}{450pt}
{\fontsize{9pt}{9pt}\selectfont \smallskip  \textit{Note}: The number of clusters (21) follows Jaffe's original specification. The total number of unique firms exceeds 333, because many firms appear in multiple clusters.}
\end{minipage}
\end{center}
\label{Table - K-means clustering}
\end{table}

\clearpage

%%% Local Variables:
%%% mode: latex
%%% TeX-master: "TDAPatent"
%%% End:

\section*{Appendix E \ Sensitivity Analysis}

This section reports the details of the sensitivity analysis in section 5.4.

\subsection*{E.1 \ Alternative Specifications of Mapper}

The 15 alternative specifications in Table \ref{Table - List of Specifications} (in the main text) cover all of the five ``tuning parameters'' in section 4.2: (i) the filter function (S2--S3), (ii) resolution (S4--S5), (iii) clustering method (S6--S10), (iv) the dissimilarity function (S11--S14), and (v) overlap (S15--S16). 

S2 uses three-dimensional PCA as $f$ with $n=20$, which drastically increases the effective resolution level to $J=n^{d}=8,000$ and leads to a larger, finer graph. S3 uses MDS as $f$ and the Euclidean distance as $\delta$. S3 generates a smaller graph with fewer flares and more lattice-like parts, which make topological structures less visible. S4 and S5 use lower and higher resolution $J\in\{225,625\}$, respectively, which generates correspondingly coarser and finer graphs. S6--S10 alter the details of hierarchical clustering (HC) methods, whereas S11--S14 use alternative distance metrics, but all of them produce graphs that are similar to S1 both qualitatively and quantitatively. Finally, S15--S16 change the overlap to $o\in\{0.3,0.7\}$, which result in more and less fragmented graphs, respectively, as expected.

 Tables \ref{Table - Sensitivity of Revenue Reg}--\ref{Table - Sensitivity of Mcap Reg} summarize the regressions using the flare/island measures based on all of the 16 different specifications. The results are surprisingly homogeneous: the correlations between firms’ performances and flare length are always positive and statistically significant, with comparable magnitudes.

\begin{table}[tbh]
\caption{Sensitivity Analysis of Revenue Regression}
\begin{center}
\fontsize{9pt}{11pt}\selectfont%
\begin{tabular}{ccccccccc}
\hline \hline
LHS variable: & \multicolumn{8}{c}{Log(Revenue)} \\ 
\cline{2-4}\cline{3-4}\cline{5-7}\cline{8-9}
Mapper specification: & (1) & (2) & (3) & (4) & (5) & (6) & (7) & (8) \\ 
& Baseline &  &  &  &  &  &  &  \\ \hline
\multicolumn{1}{l}{Flare length} & $0.34$ & $0.36$ & $0.50$ & $0.30$ & $0.28$
& $0.33$ & $0.34$ & $0.30$ \\ 
\multicolumn{1}{l}{} & $\left( 0.08\right) $ & $\left( 0.07\right) $ & $%
\left( 0.15\right) $ & $\left( 0.07\right) $ & $\left( 0.07\right) $ & $%
\left( 0.09\right) $ & $\left( 0.10\right) $ & $\left( 0.11\right) $ \\ 
\multicolumn{1}{l}{Islands only} & $0.96$ & $-$ & $-$ & $1.57$ & $1.77$ & $%
1.52$ & $1.53$ & $-$ \\ 
\multicolumn{1}{l}{} & $\left( 0.84\right) $ & $\left( -\right) $ & $\left(
-\right) $ & $\left( 0.42\right) $ & $\left( 0.84\right) $ & $\left(
1.43\right) $ & $\left( 1.43\right) $ & $\left( -\right) $ \\ 
\multicolumn{1}{l}{Log(Patents)} & $0.28$ & $0.24$ & $0.31$ & $0.28$ & $0.28$
& $0.31$ & $0.31$ & $0.33$ \\ 
\multicolumn{1}{l}{} & $\left( 0.04\right) $ & $\left( 0.05\right) $ & $%
\left( 0.04\right) $ & $\left( 0.04\right) $ & $\left( 0.04\right) $ & $%
\left( 0.04\right) $ & $\left( 0.04\right) $ & $\left( 0.04\right) $ \\ 
\multicolumn{1}{l}{Constant} & $6.08$ & $6.17$ & $5.96$ & $6.08$ & $6.09$ & $%
5.98$ & $5.97$ & $5.89$ \\ 
\multicolumn{1}{l}{} & $\left( 0.22\right) $ & $\left( 0.22\right) $ & $%
\left( 0.22\right) $ & $\left( 0.21\right) $ & $\left( 0.22\right) $ & $%
\left( 0.22\right) $ & $\left( 0.22\right) $ & $\left( 0.22\right) $ \\ 
\multicolumn{1}{l}{$R^{2}$} & $0.346$ & $0.352$ & $0.329$ & $0.361$ & $0.343$
& $0.332$ & $0.331$ & $0.322$ \\ 
\multicolumn{1}{l}{Adjusted $R^{2}$} & $0.340$ & $0.348$ & $0.324$ & $0.355$
& $0.337$ & $0.325$ & $0.324$ & $0.317$ \\ 
\multicolumn{1}{l}{Number of observations} & $328$ & $328$ & $328$ & $328$ & 
$328$ & $328$ & $328$ & $328$ \\ 
\multicolumn{1}{l}{Filter function $f$} & 2D-PCA & 3D-PCA & 2D-MDS & 2D-PCA
& 2D-PCA & 2D-PCA & 2D-PCA & 2D-PCA \\ 
\multicolumn{1}{l}{Resolution $J$} & $400$ & $8,000$ & $400$ & $225$ & $625$
& $400$ & $400$ & $400$ \\ 
\multicolumn{1}{l}{Clustering: HC linkage} & Single & Single & Single & 
Single & Single & Weighted & Average & Complete \\ 
\multicolumn{1}{l}{Clustering: \# clusters} & First gap & First gap & First
gap & First gap & First gap & First gap & First gap & First gap \\ 
\multicolumn{1}{l}{Dissimilarity function $\delta $} & Cosine & Cosine & 
Euclid & Cosine & Cosine & Cosine & Cosine & Cosine \\ 
\multicolumn{1}{l}{Overlap $o$} & $0.5$ & $0.5$ & $0.5$ & $0.5$ & $0.5$ & $%
0.5$ & $0.5$ & $0.5$ \\ \hline \hline
LHS variable: & \multicolumn{8}{c}{Log(Revenue)} \\ 
\cline{2-3}\cline{3-3}\cline{4-6}\cline{7-9}
Mapper specification: & (9) & (10) & (11) & (12) & (13) & (14) & (15) & (16)
\\ 
&  &  &  &  &  &  &  &  \\ \hline
\multicolumn{1}{l}{Flare length} & $0.38$ & $0.34$ & $0.32$ & $0.33$ & $0.28$
& $0.24$ & $0.37$ & $0.53$ \\ 
\multicolumn{1}{l}{} & $\left( 0.08\right) $ & $\left( 0.09\right) $ & $%
\left( 0.08\right) $ & $\left( 0.08\right) $ & $\left( 0.08\right) $ & $%
\left( 0.15\right) $ & $\left( 0.09\right) $ & $\left( 0.11\right) $ \\ 
\multicolumn{1}{l}{Islands only} & $0.97$ & $1.59$ & $1.04$ & $1.25$ & $2.00$
& $-$ & $0.88$ & $1.82$ \\ 
\multicolumn{1}{l}{} & $\left( 0.84\right) $ & $\left( 1.43\right) $ & $%
\left( 1.02\right) $ & $\left( 0.73\right) $ & $\left( 0.74\right) $ & $%
\left( -\right) $ & $\left( 0.34\right) $ & $\left( 0.61\right) $ \\ 
\multicolumn{1}{l}{Log(Patents)} & $0.28$ & $0.29$ & $0.28$ & $0.28$ & $0.29$
& $0.36$ & $0.29$ & $0.27$ \\ 
\multicolumn{1}{l}{} & $\left( 0.04\right) $ & $\left( 0.04\right) $ & $%
\left( 0.04\right) $ & $\left( 0.04\right) $ & $\left( 0.04\right) $ & $%
\left( 0.04\right) $ & $\left( 0.04\right) $ & $\left( 0.04\right) $ \\ 
\multicolumn{1}{l}{Constant} & $6.08$ & $6.03$ & $6.06$ & $6.08$ & $6.05$ & $%
5.78$ & $5.98$ & $6.13$ \\ 
\multicolumn{1}{l}{} & $\left( 0.22\right) $ & $\left( 0.22\right) $ & $%
\left( 0.22\right) $ & $\left( 0.22\right) $ & $\left( 0.22\right) $ & $%
\left( 0.22\right) $ & $\left( 0.21\right) $ & $\left( 0.22\right) $ \\ 
\multicolumn{1}{l}{$R^{2}$} & $0.346$ & $0.335$ & $0.338$ & $0.346$ & $0.339$
& $0.310$ & $0.344$ & $0.355$ \\ 
\multicolumn{1}{l}{Adjusted $R^{2}$} & $0.339$ & $0.329$ & $0.332$ & $0.339$
& $0.333$ & $0.306$ & $0.338$ & $0.349$ \\ 
\multicolumn{1}{l}{Number of observations} & $328$ & $328$ & $328$ & $328$ & 
$328$ & $328$ & $328$ & $328$ \\ 
\multicolumn{1}{l}{Filter function $f$} & 2D-PCA & 2D-PCA & 2D-PCA & 2D-PCA
& 2D-PCA & 2D-PCA & 2D-PCA & 2D-PCA \\ 
\multicolumn{1}{l}{Resolution $J$} & $400$ & $400$ & $400$ & $400$ & $400$ & 
$400$ & $400$ & $400$ \\ 
\multicolumn{1}{l}{Clustering: HC linkage} & Single & Single & Single & 
Single & Single & Single & Single & Single \\ 
\multicolumn{1}{l}{Clustering: \# clusters} & Mid gap & Last gap & First gap
& First gap & First gap & First gap & First gap & First gap \\ 
\multicolumn{1}{l}{Dissimilarity function $\delta $} & Cosine & Cosine & 
Euclid & Correl. & Min-c. & Mahal. & Cosine & Cosine \\ 
\multicolumn{1}{l}{Overlap $o$} & $0.5$ & $0.5$ & $0.5$ & $0.5$ & $0.5$ & $%
0.5$ & $0.3$ & $0.7$ \\ \hline \hline
\end{tabular}
\begin{minipage}{475pt}
{\fontsize{9pt}{9pt}\selectfont \smallskip  \textit{Note}: See sections 4.2 and 5.4 for the meaning and discussion of these Mapper specifications, respectively. Standard errors are in parentheses.}
\end{minipage}
\end{center}
\label{Table - Sensitivity of Revenue Reg}
\end{table}

\begin{table}[tbh]
\caption{Sensitivity Analysis of EBIT Regression}
\begin{center}
\fontsize{9pt}{11pt}\selectfont%
\begin{tabular}{ccccccccc}
\hline \hline
LHS variable: & \multicolumn{8}{c}{Log(EBIT)} \\ 
\cline{2-4}\cline{3-4}\cline{5-7}\cline{8-9}
Mapper specification: & (1) & (2) & (3) & (4) & (5) & (6) & (7) & (8) \\ 
& Baseline &  &  &  &  &  &  &  \\ \hline
\multicolumn{1}{l}{Flare length} & $0.33$ & $0.23$ & $0.33$ & $0.27$ & $0.20$
& $0.27$ & $0.27$ & $0.27$ \\ 
\multicolumn{1}{l}{} & $\left( 0.08\right) $ & $\left( 0.08\right) $ & $%
\left( 0.16\right) $ & $\left( 0.07\right) $ & $\left( 0.07\right) $ & $%
\left( 0.10\right) $ & $\left( 0.11\right) $ & $\left( 0.11\right) $ \\ 
\multicolumn{1}{l}{Islands only} & $0.94$ & $-$ & $-$ & $1.14$ & $1.40$ & $%
1.22$ & $1.20$ & $-$ \\ 
\multicolumn{1}{l}{} & $\left( 0.88\right) $ & $\left( -\right) $ & $\left(
-\right) $ & $\left( 0.44\right) $ & $\left( 0.89\right) $ & $\left(
1.51\right) $ & $\left( 1.51\right) $ & $\left( -\right) $ \\ 
\multicolumn{1}{l}{Log(Patents)} & $0.29$ & $0.31$ & $0.35$ & $0.31$ & $0.32$
& $0.33$ & $0.34$ & $0.35$ \\ 
\multicolumn{1}{l}{} & $\left( 0.05\right) $ & $\left( 0.05\right) $ & $%
\left( 0.05\right) $ & $\left( 0.04\right) $ & $\left( 0.05\right) $ & $%
\left( 0.05\right) $ & $\left( 0.05\right) $ & $\left( 0.04\right) $ \\ 
\multicolumn{1}{l}{Constant} & $3.97$ & $3.89$ & $3.75$ & $3.91$ & $3.88$ & $%
3.83$ & $3.81$ & $3.77$ \\ 
\multicolumn{1}{l}{} & $\left( 0.24\right) $ & $\left( 0.25\right) $ & $%
\left( 0.24\right) $ & $\left( 0.23\right) $ & $\left( 0.24\right) $ & $%
\left( 0.24\right) $ & $\left( 0.24\right) $ & $\left( 0.24\right) $ \\ 
\multicolumn{1}{l}{$R^{2}$} & $0.345$ & $0.327$ & $0.317$ & $0.345$ & $0.328$
& $0.326$ & $0.323$ & $0.321$ \\ 
\multicolumn{1}{l}{Adjusted $R^{2}$} & $0.338$ & $0.322$ & $0.313$ & $0.339$
& $0.322$ & $0.319$ & $0.316$ & $0.317$ \\ 
\multicolumn{1}{l}{Number of observations} & $301$ & $301$ & $301$ & $301$ & 
$301$ & $301$ & $301$ & $301$ \\ 
\multicolumn{1}{l}{Filter function $f$} & 2D-PCA & 3D-PCA & 2D-MDS & 2D-PCA
& 2D-PCA & 2D-PCA & 2D-PCA & 2D-PCA \\ 
\multicolumn{1}{l}{Resolution $J$} & $400$ & $8,000$ & $400$ & $225$ & $625$
& $400$ & $400$ & $400$ \\ 
\multicolumn{1}{l}{Clustering: HC linkage} & Single & Single & Single & 
Single & Single & Weighted & Average & Complete \\ 
\multicolumn{1}{l}{Clustering: \# clusters} & First gap & First gap & First
gap & First gap & First gap & First gap & First gap & First gap \\ 
\multicolumn{1}{l}{Dissimilarity function $\delta $} & Cosine & Cosine & 
Euclid & Cosine & Cosine & Cosine & Cosine & Cosine \\ 
\multicolumn{1}{l}{Overlap $o$} & $0.5$ & $0.5$ & $0.5$ & $0.5$ & $0.5$ & $%
0.5$ & $0.5$ & $0.5$ \\ \hline \hline
LHS variable: & \multicolumn{8}{c}{Log(EBIT)} \\ 
\cline{2-3}\cline{3-3}\cline{4-6}\cline{7-9}
Mapper specification: & (9) & (10) & (11) & (12) & (13) & (14) & (15) & (16)
\\ 
&  &  &  &  &  &  &  &  \\ \hline
\multicolumn{1}{l}{Flare length} & $0.33$ & $0.33$ & $0.25$ & $0.31$ & $0.21$
& $0.18$ & $0.28$ & $0.38$ \\ 
\multicolumn{1}{l}{} & $\left( 0.09\right) $ & $\left( 0.10\right) $ & $%
\left( 0.08\right) $ & $\left( 0.08\right) $ & $\left( 0.08\right) $ & $%
\left( 0.17\right) $ & $\left( 0.10\right) $ & $\left( 0.12\right) $ \\ 
\multicolumn{1}{l}{Islands only} & $0.87$ & $1.35$ & $0.70$ & $1.15$ & $1.48$
& $-$ & $0.74$ & $1.32$ \\ 
\multicolumn{1}{l}{} & $\left( 0.88\right) $ & $\left( 1.50\right) $ & $%
\left( 1.08\right) $ & $\left( 0.77\right) $ & $\left( 0.78\right) $ & $%
\left( -\right) $ & $\left( 0.37\right) $ & $\left( 0.65\right) $ \\ 
\multicolumn{1}{l}{Log(Patents)} & $0.31$ & $0.31$ & $0.32$ & $0.30$ & $0.33$
& $0.38$ & $0.33$ & $0.32$ \\ 
\multicolumn{1}{l}{} & $\left( 0.05\right) $ & $\left( 0.05\right) $ & $%
\left( 0.05\right) $ & $\left( 0.05\right) $ & $\left( 0.05\right) $ & $%
\left( 0.05\right) $ & $\left( 0.04\right) $ & $\left( 0.05\right) $ \\ 
\multicolumn{1}{l}{Constant} & $3.92$ & $3.91$ & $3.88$ & $3.96$ & $3.86$ & $%
3.65$ & $3.79$ & $3.90$ \\ 
\multicolumn{1}{l}{} & $\left( 0.24\right) $ & $\left( 0.24\right) $ & $%
\left( 0.24\right) $ & $\left( 0.24\right) $ & $\left( 0.24\right) $ & $%
\left( 0.24\right) $ & $\left( 0.23\right) $ & $\left( 0.24\right) $ \\ 
\multicolumn{1}{l}{$R^{2}$} & $0.337$ & $0.334$ & $0.328$ & $0.342$ & $0.326$
& $0.311$ & $0.331$ & $0.333$ \\ 
\multicolumn{1}{l}{Adjusted $R^{2}$} & $0.331$ & $0.327$ & $0.321$ & $0.336$
& $0.320$ & $0.306$ & $0.324$ & $0.326$ \\ 
\multicolumn{1}{l}{Number of observations} & $301$ & $301$ & $301$ & $301$ & 
$301$ & $301$ & $301$ & $301$ \\ 
\multicolumn{1}{l}{Filter function $f$} & 2D-PCA & 2D-PCA & 2D-PCA & 2D-PCA
& 2D-PCA & 2D-PCA & 2D-PCA & 2D-PCA \\ 
\multicolumn{1}{l}{Resolution $J$} & $400$ & $400$ & $400$ & $400$ & $400$ & 
$400$ & $400$ & $400$ \\ 
\multicolumn{1}{l}{Clustering: HC linkage} & Single & Single & Single & 
Single & Single & Single & Single & Single \\ 
\multicolumn{1}{l}{Clustering: \# clusters} & Mid gap & Last gap & First gap
& First gap & First gap & First gap & First gap & First gap \\ 
\multicolumn{1}{l}{Dissimilarity function $\delta $} & Cosine & Cosine & 
Euclid & Correl. & Min-c. & Mahal. & Cosine & Cosine \\ 
\multicolumn{1}{l}{Overlap $o$} & $0.5$ & $0.5$ & $0.5$ & $0.5$ & $0.5$ & $%
0.5$ & $0.3$ & $0.7$ \\ \hline \hline
\end{tabular}
\begin{minipage}{475pt}
{\fontsize{9pt}{9pt}\selectfont \smallskip  \textit{Note}: See sections 4.2 and 5.4 for the meaning and discussion of these Mapper specifications, respectively. Standard errors are in parentheses.}
\end{minipage}
\end{center}
\label{Table - Sensitivity of EBIT Reg}
\end{table}

\begin{table}[tbh]
\caption{Sensitivity Analysis of Market-Value Regression}
\begin{center}
\fontsize{9pt}{11pt}\selectfont%
\begin{tabular}{ccccccccc}
\hline \hline
LHS variable: & \multicolumn{8}{c}{Log(Market value)} \\ 
\cline{2-4}\cline{3-4}\cline{5-7}\cline{8-9}
Mapper specification: & (1) & (2) & (3) & (4) & (5) & (6) & (7) & (8) \\ 
& Baseline &  &  &  &  &  &  &  \\ \hline
\multicolumn{1}{l}{Flare length} & $0.28$ & $0.18$ & $0.13$ & $0.24$ & $0.15$
& $0.27$ & $0.25$ & $0.27$ \\ 
\multicolumn{1}{l}{} & $\left( 0.08\right) $ & $\left( 0.08\right) $ & $%
\left( 0.16\right) $ & $\left( 0.07\right) $ & $\left( 0.07\right) $ & $%
\left( 0.10\right) $ & $\left( 0.11\right) $ & $\left( 0.11\right) $ \\ 
\multicolumn{1}{l}{Islands only} & $0.70$ & $-$ & $-$ & $1.10$ & $0.91$ & $%
0.76$ & $0.71$ & $-$ \\ 
\multicolumn{1}{l}{} & $\left( 0.90\right) $ & $\left( -\right) $ & $\left(
-\right) $ & $\left( 0.45\right) $ & $\left( 0.91\right) $ & $\left(
1.54\right) $ & $\left( 1.54\right) $ & $\left( -\right) $ \\ 
\multicolumn{1}{l}{Log(Patents)} & $0.34$ & $0.36$ & $0.42$ & $0.34$ & $0.37$
& $0.36$ & $0.37$ & $0.37$ \\ 
\multicolumn{1}{l}{} & $\left( 0.05\right) $ & $\left( 0.05\right) $ & $%
\left( 0.05\right) $ & $\left( 0.04\right) $ & $\left( 0.05\right) $ & $%
\left( 0.05\right) $ & $\left( 0.05\right) $ & $\left( 0.04\right) $ \\ 
\multicolumn{1}{l}{Constant} & $6.20$ & $6.12$ & $5.93$ & $6.19$ & $6.09$ & $%
6.12$ & $6.08$ & $6.08$ \\ 
\multicolumn{1}{l}{} & $\left( 0.24\right) $ & $\left( 0.25\right) $ & $%
\left( 0.24\right) $ & $\left( 0.23\right) $ & $\left( 0.24\right) $ & $%
\left( 0.24\right) $ & $\left( 0.24\right) $ & $\left( 0.23\right) $ \\ 
\multicolumn{1}{l}{$R^{2}$} & $0.342$ & $0.330$ & $0.321$ & $0.349$ & $0.329$
& $0.335$ & $0.331$ & $0.332$ \\ 
\multicolumn{1}{l}{Adjusted $R^{2}$} & $0.336$ & $0.326$ & $0.317$ & $0.343$
& $0.323$ & $0.328$ & $0.325$ & $0.328$ \\ 
\multicolumn{1}{l}{Number of observations} & $325$ & $325$ & $325$ & $325$ & 
$325$ & $325$ & $325$ & $325$ \\ 
\multicolumn{1}{l}{Filter function $f$} & 2D-PCA & 3D-PCA & 2D-MDS & 2D-PCA
& 2D-PCA & 2D-PCA & 2D-PCA & 2D-PCA \\ 
\multicolumn{1}{l}{Resolution $J$} & $400$ & $8,000$ & $400$ & $225$ & $625$
& $400$ & $400$ & $400$ \\ 
\multicolumn{1}{l}{Clustering: HC linkage} & Single & Single & Single & 
Single & Single & Weighted & Average & Complete \\ 
\multicolumn{1}{l}{Clustering: \# clusters} & First gap & First gap & First
gap & First gap & First gap & First gap & First gap & First gap \\ 
\multicolumn{1}{l}{Dissimilarity function $\delta $} & Cosine & Cosine & 
Euclid & Cosine & Cosine & Cosine & Cosine & Cosine \\ 
\multicolumn{1}{l}{Overlap $o$} & $0.5$ & $0.5$ & $0.5$ & $0.5$ & $0.5$ & $%
0.5$ & $0.5$ & $0.5$ \\ \hline \hline
LHS variable: & \multicolumn{8}{c}{Log(Market value)} \\ 
\cline{2-3}\cline{3-3}\cline{4-6}\cline{7-9}
Mapper specification: & (9) & (10) & (11) & (12) & (13) & (14) & (15) & (16)
\\ 
&  &  &  &  &  &  &  &  \\ \hline
\multicolumn{1}{l}{Flare length} & $0.29$ & $0.30$ & $0.17$ & $0.25$ & $0.20$
& $0.16$ & $0.26$ & $0.36$ \\ 
\multicolumn{1}{l}{} & $\left( 0.09\right) $ & $\left( 0.10\right) $ & $%
\left( 0.08\right) $ & $\left( 0.08\right) $ & $\left( 0.08\right) $ & $%
\left( 0.16\right) $ & $\left( 0.10\right) $ & $\left( 0.12\right) $ \\ 
\multicolumn{1}{l}{Islands only} & $0.68$ & $0.84$ & $0.15$ & $0.89$ & $1.11$
& $-$ & $0.55$ & $1.15$ \\ 
\multicolumn{1}{l}{} & $\left( 0.91\right) $ & $\left( 1.53\right) $ & $%
\left( 1.11\right) $ & $\left( 0.79\right) $ & $\left( 0.79\right) $ & $%
\left( -\right) $ & $\left( 0.37\right) $ & $\left( 0.67\right) $ \\ 
\multicolumn{1}{l}{Log(Patents)} & $0.34$ & $0.35$ & $0.38$ & $0.35$ & $0.36$
& $0.41$ & $0.37$ & $0.35$ \\ 
\multicolumn{1}{l}{} & $\left( 0.05\right) $ & $\left( 0.05\right) $ & $%
\left( 0.05\right) $ & $\left( 0.05\right) $ & $\left( 0.05\right) $ & $%
\left( 0.05\right) $ & $\left( 0.05\right) $ & $\left( 0.05\right) $ \\ 
\multicolumn{1}{l}{Constant} & $6.18$ & $6.18$ & $6.07$ & $6.17$ & $6.13$ & $%
5.94$ & $6.08$ & $6.18$ \\ 
\multicolumn{1}{l}{} & $\left( 0.24\right) $ & $\left( 0.24\right) $ & $%
\left( 0.24\right) $ & $\left( 0.24\right) $ & $\left( 0.24\right) $ & $%
\left( 0.24\right) $ & $\left( 0.23\right) $ & $\left( 0.24\right) $ \\ 
\multicolumn{1}{l}{$R^{2}$} & $0.340$ & $0.338$ & $0.328$ & $0.339$ & $0.332$
& $0.321$ & $0.335$ & $0.339$ \\ 
\multicolumn{1}{l}{Adjusted $R^{2}$} & $0.334$ & $0.332$ & $0.322$ & $0.332$
& $0.326$ & \thinspace $0.317$ & $0.329$ & $0.333$ \\ 
\multicolumn{1}{l}{Number of observations} & $325$ & $325$ & $325$ & $325$ & 
$325$ & $325$ & $325$ & $325$ \\ 
\multicolumn{1}{l}{Filter function $f$} & 2D-PCA & 2D-PCA & 2D-PCA & 2D-PCA
& 2D-PCA & 2D-PCA & 2D-PCA & 2D-PCA \\ 
\multicolumn{1}{l}{Resolution $J$} & $400$ & $400$ & $400$ & $400$ & $400$ & 
$400$ & $400$ & $400$ \\ 
\multicolumn{1}{l}{Clustering: HC linkage} & Single & Single & Single & 
Single & Single & Single & Single & Single \\ 
\multicolumn{1}{l}{Clustering: \# clusters} & Mid gap & Last gap & First gap
& First gap & First gap & First gap & First gap & First gap \\ 
\multicolumn{1}{l}{Dissimilarity function $\delta $} & Cosine & Cosine & 
Euclid & Correl. & Min-c. & Mahal. & Cosine & Cosine \\ 
\multicolumn{1}{l}{Overlap $o$} & $0.5$ & $0.5$ & $0.5$ & $0.5$ & $0.5$ & $%
0.5$ & $0.3$ & $0.7$ \\ \hline \hline
\end{tabular}
\begin{minipage}{475pt}
{\fontsize{9pt}{9pt}\selectfont \smallskip  \textit{Note}: See sections 4.2 and 5.4 for the meaning and discussion of these Mapper specifications, respectively. Standard errors are in parentheses.}
\end{minipage}
\end{center}
\label{Table - Sensitivity of Mcap Reg}
\end{table}

\clearpage

\subsection*{E.2 \ Survivorship}

Tables \ref{Table - Regressions (Survivors thru 2005)} and \ref{Table - Regressions (Balanced 1976-2005)} show the results are robust to (i) the elimination of firms that exited our sample before 2005 and (ii) conditioning on the balanced panel of firms.

\begin{table}[tbh!!!]
\caption{Flares, Counts, and Performances (Survivors through 2005)}
\begin{center}
\fontsize{9pt}{11pt}\selectfont%
\begin{tabular}{cccccccccccc}
\hline \hline
LHS variable: & \multicolumn{3}{c}{Log(Revenue)} &  & \multicolumn{3}{c}{
Log(EBIT)} &  & \multicolumn{3}{c}{Log(Market value)} \\ 
\cline{2-4}\cline{3-4}\cline{6-8}\cline{10-12}
& (1) & (2) & (3) &  & (4) & (5) & (6) &  & (7) & (8) & (9) \\ \hline
\multicolumn{1}{l}{Flare length} & $0.65$ & $-$ & $0.35$ &  & $0.64$ & $-$ & 
$0.33$ &  & $0.61$ & $-$ & $0.26$ \\ 
\multicolumn{1}{l}{} & $\left( 0.07\right) $ & $\left( -\right) $ & $\left(
0.08\right) $ &  & $\left( 0.07\right) $ & $\left( -\right) $ & $\left(
0.09\right) $ &  & $\left( 0.07\right) $ & $\left( -\right) $ & $\left(
0.09\right) $ \\ 
\multicolumn{1}{l}{Islands only} & $2.49$ & $-$ & $1.17$ &  & $2.61$ & $-$ & 
$1.21$ &  & $2.57$ & $-$ & $1.04$ \\ 
\multicolumn{1}{l}{} & $\left( 1.04\right) $ & $\left( -\right) $ & $\left(
1.01\right) $ &  & $\left( 1.11\right) $ & $\left( -\right) $ & $\left(
1.08\right) $ &  & $\left( 1.11\right) $ & $\left( -\right) $ & $\left(
1.07\right) $ \\ 
\multicolumn{1}{l}{Log(Patents)} & $-$ & $0.41$ & $0.27$ &  & $-$ & $0.42$ & 
$0.29$ &  & $-$ & $0.42$ & $0.32$ \\ 
\multicolumn{1}{l}{} & $\left( -\right) $ & $\left( 0.04\right) $ & $\left(
0.05\right) $ &  & $\left( -\right) $ & $\left( 0.04\right) $ & $\left(
0.05\right) $ &  & $\left( -\right) $ & $\left( 0.04\right) $ & $\left(
0.05\right) $ \\ 
\multicolumn{1}{l}{Constant} & $7.47$ & $5.66$ & $6.15$ &  & $5.46
$ & $3.62$ & $4.07$ &  & $8.01$ & $6.11$ & $6.48$ \\ 
\multicolumn{1}{l}{} & $\left( 0.11\right) $ & $\left( 0.24\right) $ & $%
\left( 0.26\right) $ &  & $\left( 0.12\right) $ & $\left( 0.25\right) $ & $%
\left( 0.27\right) $ &  & $\left( 0.11\right) $ & $\left( 0.24\right) $ & $%
\left( 0.27\right) $ \\ 
\multicolumn{1}{l}{$R^{2}$} & $0.285$ & $0.317$ & $0.365$ &  & $0.266$ & $0.313$
& $0.351$ &  & $0.238$ & $0.314$ & $0.339$ \\ 
\multicolumn{1}{l}{Adjusted $R^{2}$} & $0.279$ & $0.314$ & $0.357$ &  & $0.260$
& $0.310$ & $0.343$ &  & $0.232$ & $0.311$ & $0.332$ \\ 
\multicolumn{1}{l}{Number of observations} & $256$ & $256$ & $256$ &  & $238$
& $238$ & $238$ &  & $255$ & $255$ & $255$ \\ \hline \hline
\end{tabular}
\begin{minipage}{450pt}
{\fontsize{9pt}{9pt}\selectfont \smallskip  \textit{Note}: This table is the same as Table \ref{Table - Regressions} except for conditioning on the availability of financial data in 2005.}
\end{minipage}
\end{center}
\label{Table - Regressions (Survivors thru 2005)}
\end{table}

\begin{table}[tbh!!!]
\caption{Flares, Counts, and Performances (Balanced Panel, 1976--2005)}
\begin{center}
\fontsize{9pt}{11pt}\selectfont%
\begin{tabular}{cccccccccccc}
\hline \hline
LHS variable: & \multicolumn{3}{c}{Log(Revenue)} &  & \multicolumn{3}{c}{
Log(EBIT)} &  & \multicolumn{3}{c}{Log(Market value)} \\ 
\cline{2-4}\cline{3-4}\cline{6-8}\cline{10-12}
& (1) & (2) & (3) &  & (4) & (5) & (6) &  & (7) & (8) & (9) \\ \hline
\multicolumn{1}{l}{Flare length} & $0.54$ & $-$ & $0.25$ &  & $0.55$ & $-$ & 
$0.27$ &  & $0.58$ & $-$ & $0.26$ \\ 
\multicolumn{1}{l}{} & $\left( 0.07\right) $ & $\left( -\right) $ & $\left(
0.08\right) $ &  & $\left( 0.08\right) $ & $\left( -\right) $ & $\left(
0.10\right) $ &  & $\left( 0.08\right) $ & $\left( -\right) $ & $\left(
0.09\right) $ \\ 
\multicolumn{1}{l}{Islands only} & $1.39$ & $-$ & $0.10$ &  & $2.01$ & $-$ & 
$0.73$ &  & $2.47$ & $-$ & $0.99$ \\ 
\multicolumn{1}{l}{} & $\left( 1.16\right) $ & $\left( -\right) $ & $\left(
1.08\right) $ &  & $\left( 1.38\right) $ & $\left( -\right) $ & $\left(
1.32\right) $ &  & $\left( 1.30\right) $ & $\left( -\right) $ & $\left(
1.20\right) $ \\ 
\multicolumn{1}{l}{Log(Patents)} & $-$ & $0.44$ & $0.32$ &  & $-$ & $0.45$ & 
$0.32$ &  & $-$ & $0.50$ & $0.37$ \\ 
\multicolumn{1}{l}{} & $\left( -\right) $ & $\left( 0.05\right) $ & $\left(
0.06\right) $ &  & $\left( -\right) $ & $\left( 0.06\right) $ & $\left(
0.08\right) $ &  & $\left( -\right) $ & $\left( 0.05\right) $ & $\left(
0.07\right) $ \\
\multicolumn{1}{l}{Constant} & $8.12$ & $5.76$ & $6.29$ &  & $5.92
$ & $3.52$ & $4.12$ &  & $8.37$ & $5.71$ & $6.28$ \\ 
\multicolumn{1}{l}{} & $\left( 0.13\right) $ & $\left( 0.34\right) $ & $%
\left( 0.38\right) $ &  & $\left( 0.16\right) $ & $\left( 0.41\right) $ & $%
\left( 0.46\right) $ &  & $\left( 0.15\right) $ & $\left( 0.37\right) $ & $%
\left( 0.42\right) $ \\ 
\multicolumn{1}{l}{$R^{2}$} & $0.350$ & $0.431$ & $0.476$ &  & $0.297$ & $0.355$
& $0.395$ &  & $0.343$ & $0.440$ & $0.478$ \\ 
\multicolumn{1}{l}{Adjusted $R^{2}$} & $0.338$ & $0.425$ & $0.462$ &  & $0.284$
& $0.349$ & $0.378$ &  & $0.331$ & $0.435$ & $0.463$ \\ 
\multicolumn{1}{l}{Number of observations} & $112$ & $112$ & $112$ &  & $109$
& $109$ & $109$ &  & $112$ & $112$ & $112$ \\ \hline \hline
\end{tabular}
\begin{minipage}{450pt}
{\fontsize{9pt}{9pt}\selectfont \smallskip  \textit{Note}: This table is the same as Table \ref{Table - Regressions} except for conditioning on the availability of financial data in 1976--2005.}
\end{minipage}
\end{center}
\label{Table - Regressions (Balanced 1976-2005)}
\end{table}

\clearpage

\subsection*{E.3 \ Subsampling by Sector and Industry}

Figure \ref{Figure - bubble plots)} (a) plots each firm's revenue in 2005 (on the vertical axis) against the flare length of its patents in 1976--2005 (on the horizontal axis). The circle size reflects the total count of patents in 1976--2005. The maximum finite flare length of all firms is 8; the figure shows infinitely long flares (i.e., islands-only type) at length 10 for ease of visualization. Two patterns emerge. First, the upper-triangle-like shape of the scatter plot suggests long flares always entail high revenues, but the reverse is not true. Some high-revenue firms show short or no flares. Second, the prevalence of large circles in the upper region suggests large portfolios are frequently associated with both high revenues and long flares. However, some firms have many patents but only short flares of length 2 or 3. Thus, long flares predict high revenues and many patents, but not all ``large'' firms exhibit long flares. Panels (b) and (c) show similar patterns for profit and market value, respectively.

\begin{figure}[htb!!!!]
    \caption{Flares and Financial Performances}
    \begin{subfigure}{0.3\textwidth}
        \centering
        \caption{Revenue}
        \includegraphics[width=\textwidth]{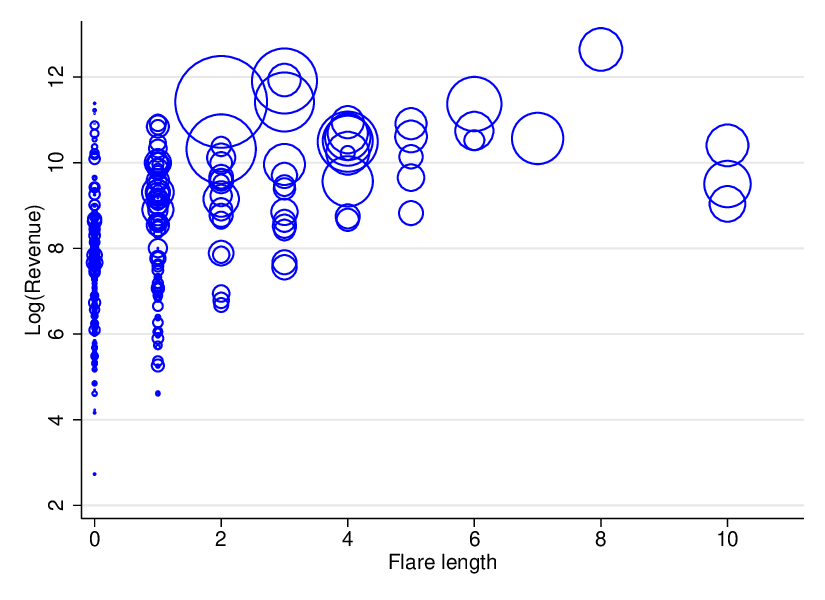}
    \end{subfigure}
    \hfill
    \begin{subfigure}{0.3\textwidth}
        \centering
        \caption{EBIT}
        \includegraphics[width=\textwidth]{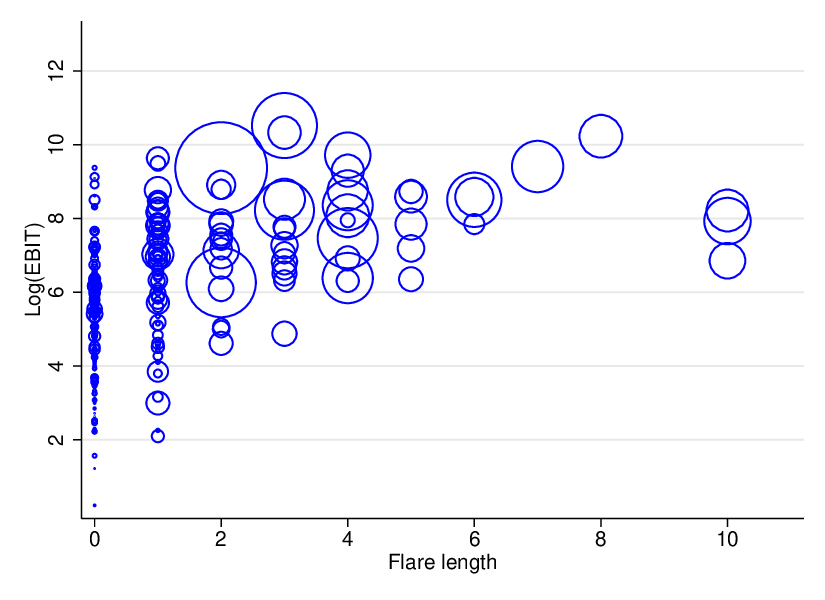}
    \end{subfigure}
    \hfill
    \begin{subfigure}{0.3\textwidth}
        \centering
        \caption{Market value}
        \includegraphics[width=\textwidth]{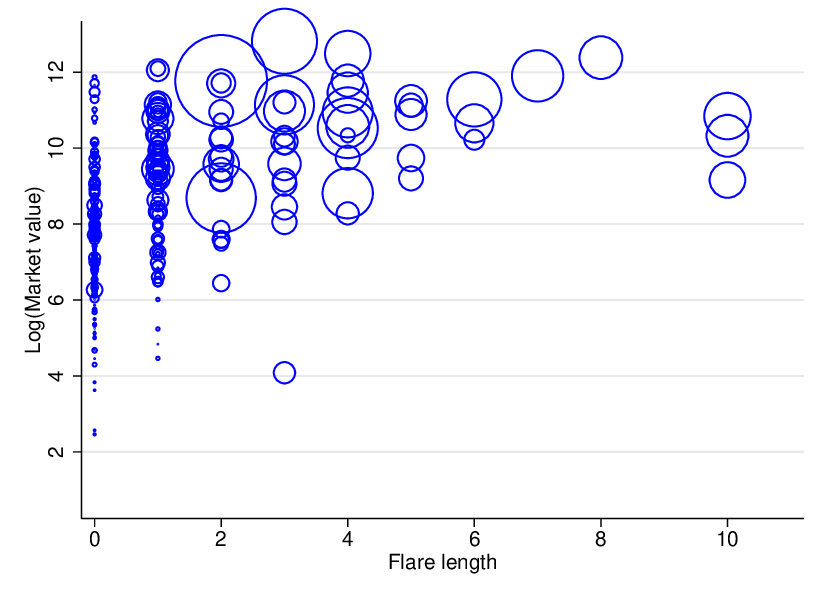}
    \end{subfigure}
    \caption*{\footnotesize {%
        \textit{Note}: The center of each circle represents the firm's revenue in 2005 and the flare length of its patent portfolio in 1976--2005 (based on cosine distance). The circle size reflects the firm's total patent count across all classes and all years. Infinitely long flares (i.e., islands-only type) are shown at length 10 for illustration purposes.}}%
    \label{Figure - bubble plots)}
\end{figure}%

These patterns are not an artifact of aggregation or driven by a few specific sectors and industries. Figure \ref{Figure - revenues & flares by sector} plots revenues and flares by economic sector defined by Standard and Poor's (S\&P), a credit-rating agency. Figure \ref{Figure - revenues & flares by SIC code} studies the technology sector more deeply at the SIC-code level, with a focus on computers and semiconductor industries. These additional scatter plots show the positive correlations are preserved within each sector and industry. 

\begin{figure}[htb!!!!]
    \caption{Revenues and Flares by Sector}%
    \begin{subfigure}{0.3\textwidth}
        \caption{Technology}%
        \centering
        \includegraphics[width=\textwidth]{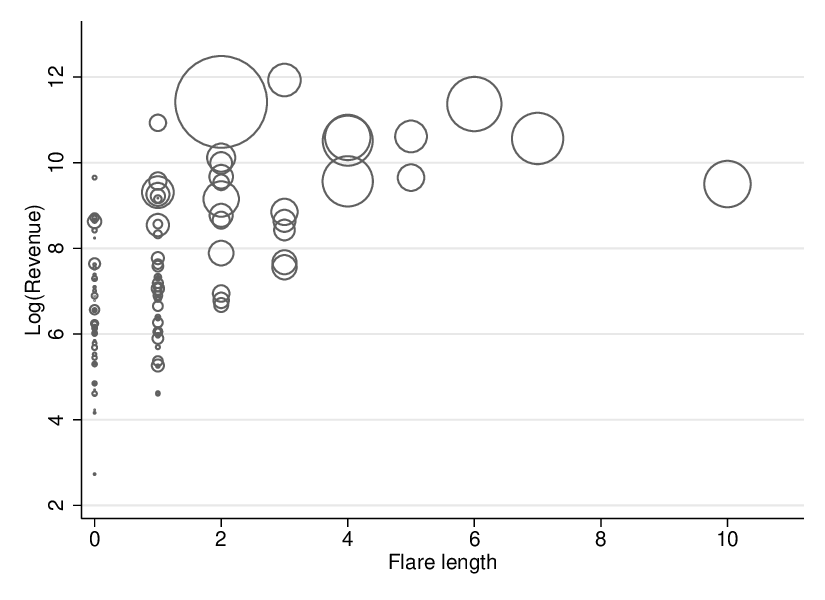}
    \end{subfigure}
    \hfill
    \begin{subfigure}{0.3\textwidth}
        \caption{Capital Goods}%
        \centering
        \includegraphics[width=\textwidth]{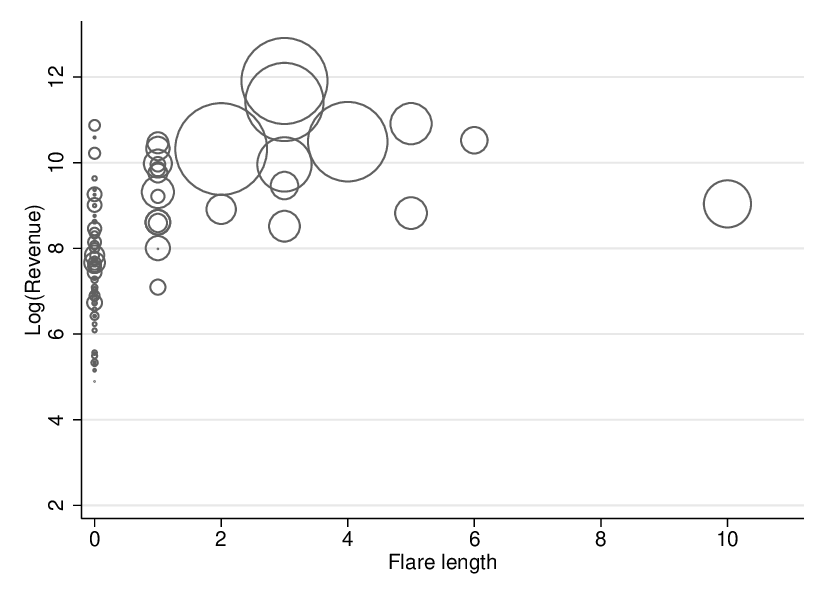}
    \end{subfigure}
    \hfill
    \begin{subfigure}{0.3\textwidth}
        \caption{Health Care}%
        \centering
        \includegraphics[width=\textwidth]{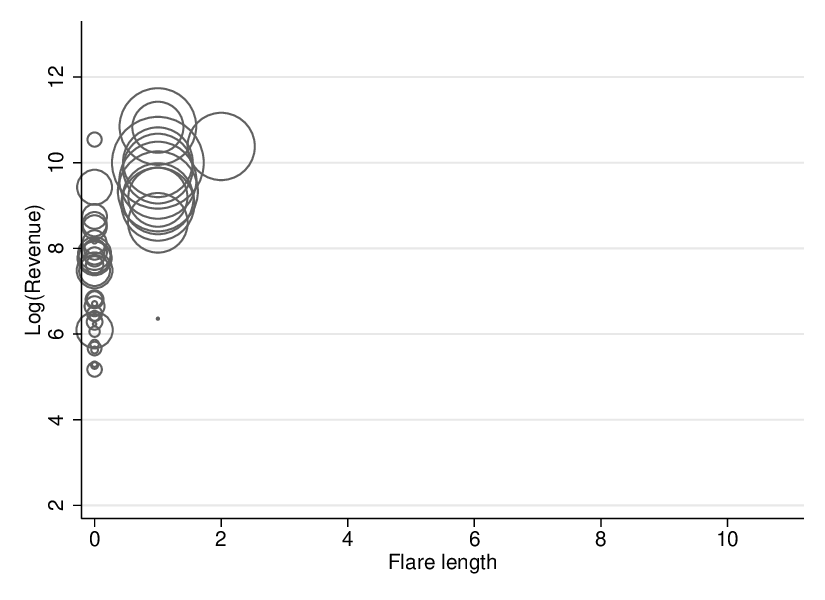}
    \end{subfigure}
    
    \begin{subfigure}{0.3\textwidth}
        \caption{Consumer Goods}%
        \centering
        \includegraphics[width=\textwidth]{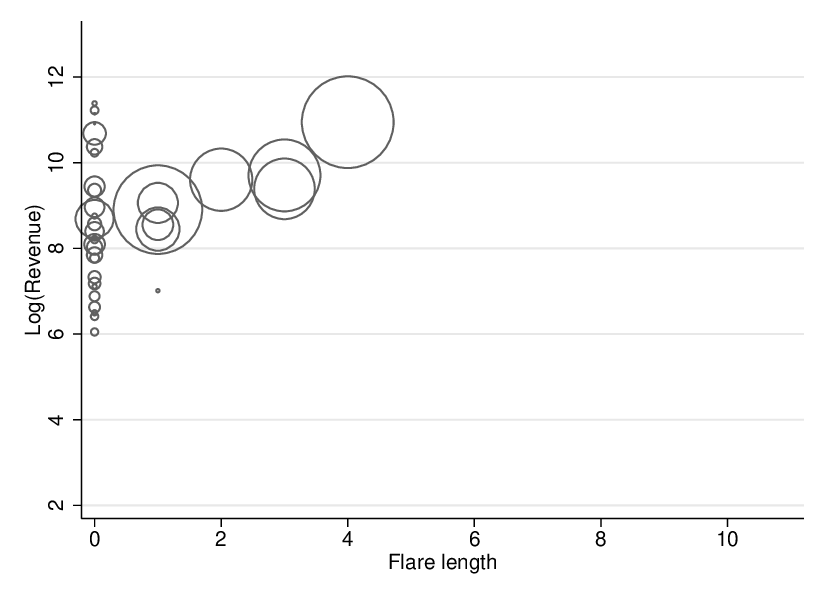}
    \end{subfigure}
    \hfill
    \begin{subfigure}{0.3\textwidth}
        \caption{Basic Materials}%
        \centering
        \includegraphics[width=\textwidth]{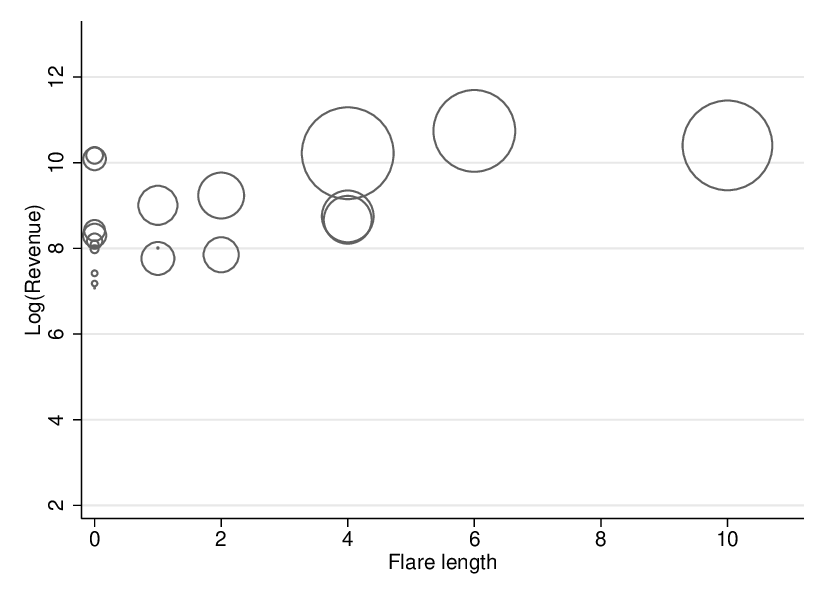}
    \end{subfigure}
    \hfill
    \begin{subfigure}{0.3\textwidth}
        \caption{Others}%
        \centering
        \includegraphics[width=\textwidth]{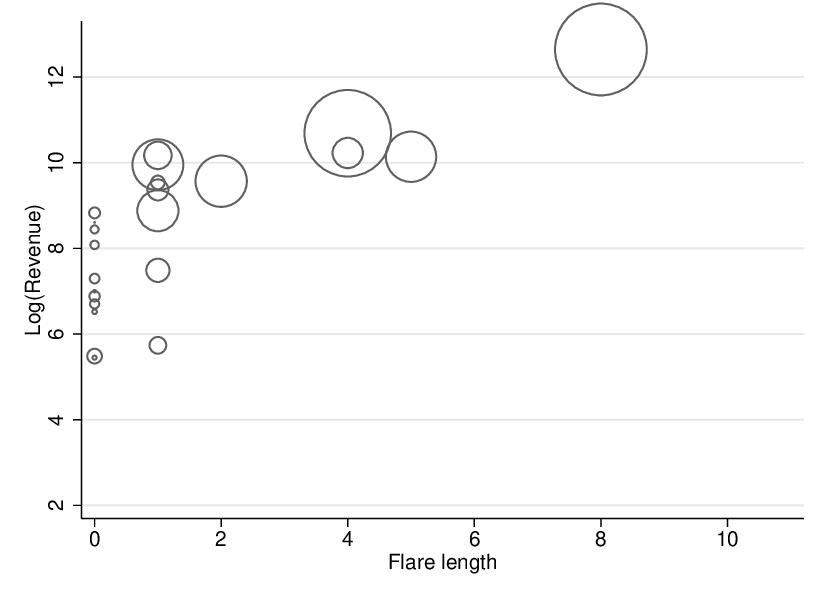}
    \end{subfigure}
    \caption*{\footnotesize {%
        \textit{Note}: ``Consumer goods'' include the S\&P consumer-cyclicals and consumer-staples sectors. ``Others'' include the S\&P energy, communication services, transport, and utilities sectors.}}%
    \label{Figure - revenues & flares by sector}
\end{figure}%

\begin{figure}[htb!!!!]
    \caption{Revenues and Flares by SIC Code}%
    \begin{subfigure}{0.5\textwidth}
        \caption{Computers and Peripherals}%
        \centering
        \includegraphics[width=0.9\linewidth]{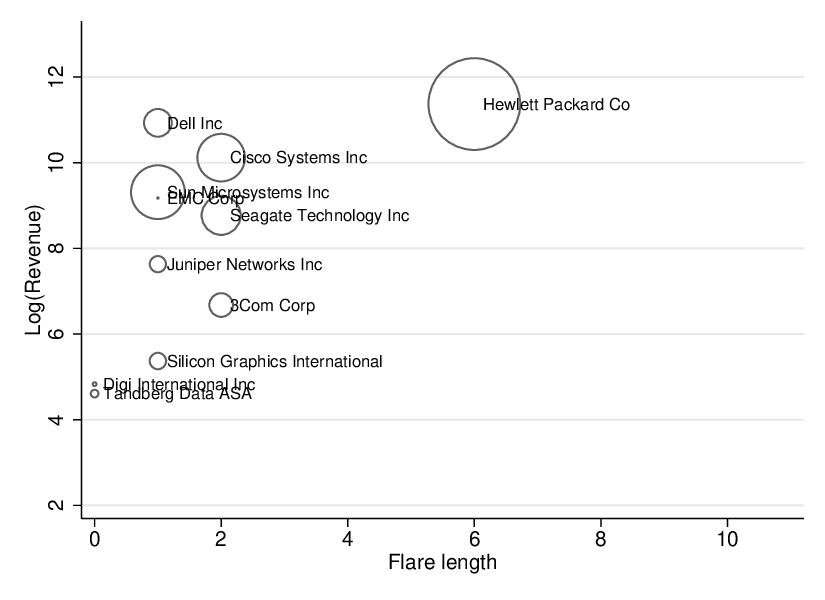}
    \end{subfigure}
    \begin{subfigure}{0.5\textwidth}
        \caption{Semiconductors}%
        \centering
        \includegraphics[width=0.9\linewidth]{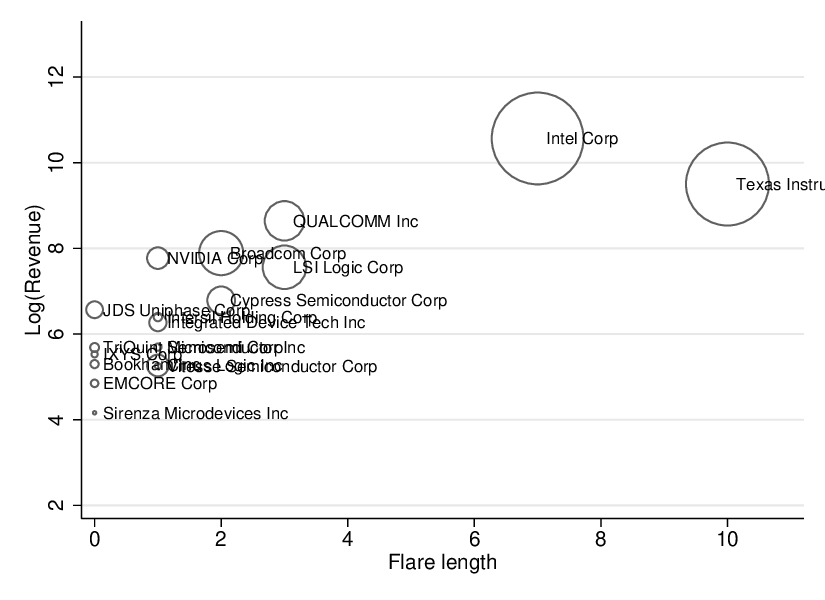}
    \end{subfigure}
    \caption*{\footnotesize {%
        \textit{Note}: For computers and their peripherals, we use 3570 (computer and office equipment), 3571 (electronic computers), 3572 (computer storage devices), 3575 (computer terminals), and 3576 (computer communications equipment). For semiconductors, we use SIC code 3674 (semiconductors and related devices).}}%
    \label{Figure - revenues & flares by SIC code}
\end{figure}%

\clearpage

\subsection*{E.4 \ Patents Acquired by M\&As}

How does the picture change if we incorporate M\&A patents as well? Figure \ref{Figure - mapper(cos_log_m2)} shows another graph based on both R\&D and M\&A patents. Because only 11.4\% of all patents are obtained by M\&As, the overall pattern looks familiar: 1,253 nodes, 3,084 edges, and 31 connected components. The average degree is 4.92, the average number of nodes per firm is 9.44, and the average flare length is 0.87. The main change is that slightly more connections are formed. Thus, M\&A patents seem to marginally expand the firms' coverage areas, ``fill in the gaps'' between firms, and make their eventual portfolios more similar to each other than the R\&D-only versions are. This tendency seems particularly strong in IT-related industries. By contrast, engineering conglomerates, pharmaceuticals, and chemical firms exhibit relatively small changes. They are already clustered together and densely connected in the previous graph; hence, M\&A patents can add only so many connections. 

\begin{figure}[htb!!!!]
\caption{Mapper Graph of Both R\&D and M\&A Patents}%
\centering
\includegraphics[width=0.5\linewidth]{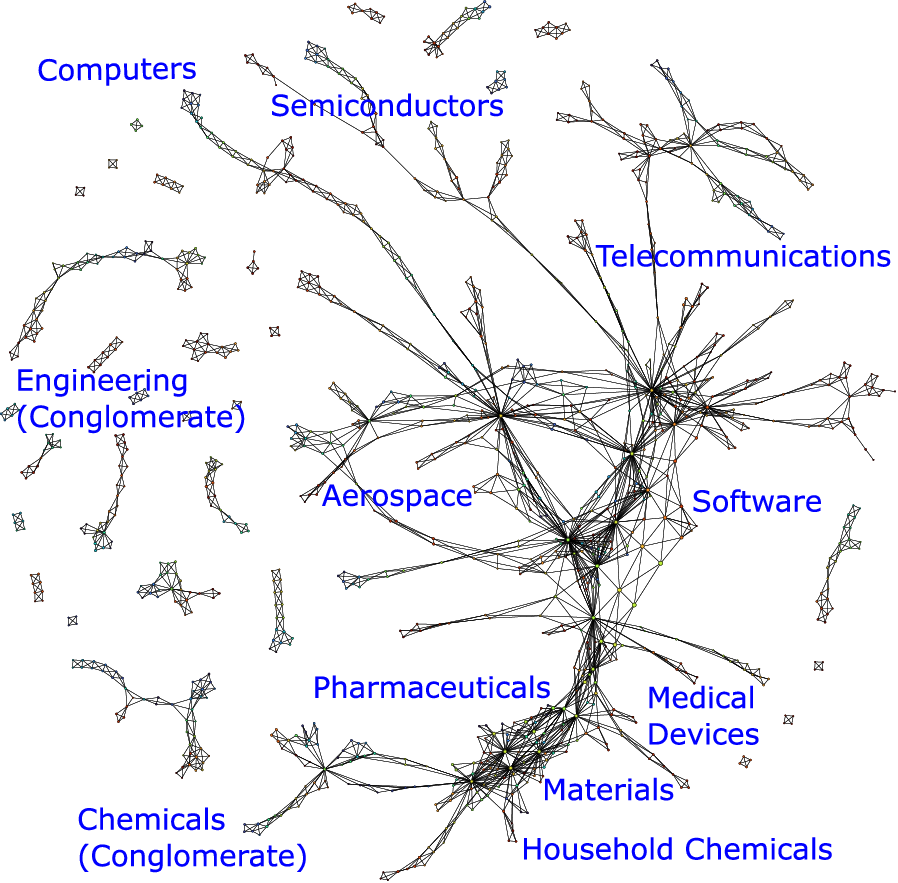}
\caption*{\footnotesize {%
\textit{Note}: This version uses both R\&D and M\&A patents, whereas other figures use only R\&D patents.}}%
\label{Figure - mapper(cos_log_m2)}
\end{figure}%

\clearpage

\clearpage

\section*{Appendix F \ Panel Data and Out-of-Sample Predictions}

Our main output in section 4 is the Mapper graph of the entire dataset. Accordingly, our regressions in section 5 study the statistical relationships between the firms' flares in the whole graph for 1976--2005 and their eventual performances in 2005. By contrast, this section investigates whether these relationships hold more generally---at different points in time, with many years of lags, with firm fixed effects, and in terms of out-of-sample predictions.

\paragraph{Panel-Data Regressions.}

We run the panel-data regressions of the form:
\begin{equation}
\ln (y_{i,t})=\alpha _{1}+\alpha _{2}\lambda_{i,t-\tau}+\alpha _{3}\mathbb{I}\left \{ \lambda_{i,t-\tau}=\infty \right \}
+\alpha _{4} \ln (p_{i,t-\tau}) +\varepsilon _{i,t},
\label{eq - panel regression}
\end{equation}%
where $y_{i,t}$ is firm $i$'s revenue (or other
performance metrics) in year $t$, $\lambda_{i,t}$ is the flare length of its patenting history from 1976 to $t$, $p_{i,t}$ is firm $i$'s patent applications in year $t$ (i.e., $p_{i,t}=\sum_{c}p_{i,t,c}$), and $\tau=0,1,2,...$ is the length of time lag.

\begin{table}[tbh]
\caption{Summary Statistics of 1980--2005 Panel Data}
\label{Table - Panel Sumstats}
\begin{center}
\fontsize{9pt}{11pt}\selectfont%
\begin{tabular}{lcccccc}
\hline \hline
Variable & Unit of & Number of & Mean & Standard & Minimum & Maximum \\ 
& measurement & observations &  & deviation &  &  \\ \hline
Revenue & USD \$1,000 & $5,599$ & $18,729$ & $199,999$ & $0$ & $4,129,493$
\\ 
EBIT & USD \$1,000 & $5,599$ & $1,141$ & $7,830$ & $-68,027$ & $227,497$ \\ 
Market value & USD \$1,000 & $5,183$ & $11,456$ & $32,550$ & $0.65$ & $%
508,330$ \\ 
R\&D patents & Count & $5,607$ & $101$ & $280$ & $0$ & $4,420$ \\ 
Log(Revenue) & $-$ & $5,598$ & $7.37$ & $2.00$ & $-2.04$ & $15.23$ \\ 
Log(EBIT) & $-$ & $5,129$ & $5.28$ & $1.94$ & $-3.91$ & $12.33$ \\ 
Log(Market value) & $-$ & $5,183$ & $7.62$ & $1.96$ & $-0.43$ & $13.14$ \\ 
Log(R\&D patents + 1) & $-$ & $5,607$ & $2.83$ & $1.96$ & $0$ & $8.39$ \\ 
Flare length & See main text & $5,607$ & $0.64$ & $1.11$ & $0$ & $11$ \\ 
Islands only & 0 or 1 & $5,607$ & $0.10$ & $0.30$ & $0$ & $1$ \\ \hline \hline
\end{tabular}
\begin{minipage}{450pt}
{\fontsize{9pt}{9pt}\selectfont \smallskip  \textit{Note}: Each firm's flare length and islands-only indicator in each year are based on the Mapper graph of all firms' patent data from 1976 to that year (1980, 1981, ..., 2005).}
\end{minipage}
\end{center}
\end{table}

Table \ref{Table - Panel Sumstats} summarizes the panel data we use in this section. The effective sample period starts in 1980 instead of 1976 because we use the five-year moving window to construct the Mapper graph (see section 4.1). Tables \ref{Table - Panel Revenue Reg}--\ref{Table - Panel Mcap Reg} report the results of revenue, EBIT, and market-value regressions, respectively. Each table shows three sets of results in the top, middle, and bottom panels, which correspond to regressions with (i) only the main regressors (lagged flare length and islands-only dummy), (ii) the main regressors and patent count, and (iii) the main regressors and firm fixed effects (this specification replaces $\alpha_{1}$ with $\alpha_{1,i}$ in equation \ref{eq - panel regression}), respectively. The 11 columns reflect different lags, $\tau=0,1,2,...,10$.

\begin{table}[tbh]
\caption{Revenue Regression Using 1980--2005 Panel Data}
\label{Table - Panel Revenue Reg}
\begin{center}
\fontsize{9pt}{11pt}\selectfont%
\begin{tabular}{cccccccccccc}
\hline \hline
LHS variable: & \multicolumn{11}{c}{Log(Revenue)} \\ 
\cline{2-4}\cline{3-4}\cline{5-7}\cline{8-12}
Lag of RHS variables: & $0$ & $1$ & $2$ & $3$ & $4$ & $5$ & $6$ & $7$ & $8$
& $9$ & $10$ \\ \hline
\multicolumn{1}{l}{Lagged flare length} & $0.90$ & $0.88$ & $0.87$ & $0.87$
& $0.89$ & $0.91$ & $0.92$ & $0.90$ & $0.91$ & $0.93$ & $0.92$ \\ 
\multicolumn{1}{l}{} & $\left( 0.02\right) $ & $\left( 0.02\right) $ & $%
\left( 0.02\right) $ & $\left( 0.02\right) $ & $\left( 0.02\right) $ & $%
\left( 0.03\right) $ & $\left( 0.03\right) $ & $\left( 0.03\right) $ & $%
\left( 0.03\right) $ & $\left( 0.03\right) $ & $\left( 0.03\right) $ \\ 
\multicolumn{1}{l}{Lagged islands only} & $2.13$ & $2.08$ & $2.03$ & $1.98$
& $1.95$ & $1.89$ & $1.86$ & $1.82$ & $1.77$ & $1.73$ & $1.65$ \\ 
\multicolumn{1}{l}{} & $\left( 0.08\right) $ & $\left( 0.07\right) $ & $%
\left( 0.07\right) $ & $\left( 0.07\right) $ & $\left( 0.07\right) $ & $%
\left( 0.07\right) $ & $\left( 0.07\right) $ & $\left( 0.07\right) $ & $%
\left( 0.07\right) $ & $\left( 0.08\right) $ & $\left( 0.08\right) $ \\ 
\multicolumn{1}{l}{Constant} & $6.57$ & $6.66$ & $6.75$ & $6.83$ & $6.89$ & $%
6.95$ & $7.01$ & $7.09$ & $7.15$ & $7.19$ & $7.26$ \\ 
\multicolumn{1}{l}{} & $\left( 0.03\right) $ & $\left( 0.03\right) $ & $%
\left( 0.03\right) $ & $\left( 0.03\right) $ & $\left( 0.03\right) $ & $%
\left( 0.03\right) $ & $\left( 0.03\right) $ & $\left( 0.03\right) $ & $%
\left( 0.03\right) $ & $\left( 0.04\right) $ & $\left( 0.04\right) $ \\ 
\multicolumn{1}{l}{Firm fixed effects} & No & No & No & No & No & No & No & 
No & No & No & No \\ 
\multicolumn{1}{l}{$R^{2}$} & $0.289$ & $0.289$ & $0.290$ & $0.286$ & $0.285$
& $0.286$ & $0.283$ & $0.275$ & $0.273$ & $0.274$ & $0.263$ \\ 
\multicolumn{1}{l}{Adjusted $R^{2}$} & $0.288$ & $0.289$ & $0.289$ & $0.285$
& $0.285$ & $0.285$ & $0.282$ & $0.274$ & $0.273$ & $0.273$ & $0.263$ \\ 
\multicolumn{1}{l}{Number of observations} & $5,598$ & $5,286$ & $4,973$ & $%
4,662$ & $4,354$ & $4,048$ & $3,749$ & $3,455$ & $3,170$ & $2,903$ & $2,651$
\\ \hline \hline
LHS variable: & \multicolumn{11}{c}{Log(Revenue)} \\ \cline{2-12}
Lag of RHS variables: & $0$ & $1$ & $2$ & $3$ & $4$ & $5$ & $6$ & $7$ & $8$
& $9$ & $10$ \\ \hline
\multicolumn{1}{l}{Lagged flare length} & $0.36$ & $0.33$ & $0.32$ & $0.30$
& $0.30$ & $0.30$ & $0.29$ & $0.26$ & $0.26$ & $0.26$ & $0.24$ \\ 
& $\left( 0.02\right) $ & $\left( 0.02\right) $ & $\left( 0.03\right) $ & $%
\left( 0.03\right) $ & $\left( 0.03\right) $ & $\left( 0.03\right) $ & $%
\left( 0.03\right) $ & $\left( 0.03\right) $ & $\left( 0.04\right) $ & $%
\left( 0.04\right) $ & $\left( 0.04\right) $ \\ 
\multicolumn{1}{l}{Lagged islands only} & $0.80$ & $0.74$ & $0.69$ & $0.64$
& $0.60$ & $0.55$ & $0.50$ & $0.44$ & $0.39$ & $0.34$ & $0.26$ \\ 
\multicolumn{1}{l}{} & $\left( 0.08\right) $ & $\left( 0.07\right) $ & $%
\left( 0.08\right) $ & $\left( 0.08\right) $ & $\left( 0.08\right) $ & $%
\left( 0.08\right) $ & $\left( 0.08\right) $ & $\left( 0.08\right) $ & $%
\left( 0.09\right) $ & $\left( 0.09\right) $ & $\left( 0.09\right) $ \\ 
\multicolumn{1}{l}{Lagged log patents} & $0.49$ & $0.50$ & $0.49$ & $0.49$ & 
$0.48$ & $0.48$ & $0.48$ & $0.48$ & $0.48$ & $0.48$ & $0.48$ \\ 
\multicolumn{1}{l}{} & $\left( 0.01\right) $ & $\left( 0.01\right) $ & $%
\left( 0.01\right) $ & $\left( 0.01\right) $ & $\left( 0.02\right) $ & $%
\left( 0.02\right) $ & $\left( 0.02\right) $ & $\left( 0.02\right) $ & $%
\left( 0.02\right) $ & $\left( 0.02\right) $ & $\left( 0.02\right) $ \\ 
\multicolumn{1}{l}{Constant} & $5.66$ & $5.76$ & $5.86$ & $5.96$ & $6.06$ & $%
6.15$ & $6.24$ & $6.32$ & $6.40$ & $6.48$ & $6.56$ \\ 
\multicolumn{1}{l}{} & $\left( 0.04\right) $ & $\left( 0.04\right) $ & $%
\left( 0.04\right) $ & $\left( 0.04\right) $ & $\left( 0.04\right) $ & $%
\left( 0.04\right) $ & $\left( 0.04\right) $ & $\left( 0.04\right) $ & $%
\left( 0.04\right) $ & $\left( 0.04\right) $ & $\left( 0.05\right) $ \\ 
\multicolumn{1}{l}{Firm fixed effects} & No & No & No & No & No & No & No & 
No & No & No & No \\ 
\multicolumn{1}{l}{$R^{2}$} & $0.413$ & $0.420$ & $0.422$ & $0.418$ & $0.416$
& $0.413$ & $0.411$ & $0.406$ & $0.405$ & $0.400$ & $0.394$ \\ 
\multicolumn{1}{l}{Adjusted $R^{2}$} & $0.412$ & $0.420$ & $0.422$ & $0.418$
& $0.415$ & $0.413$ & $0.410$ & $0.406$ & $0.404$ & $0.400$ & $0.393$ \\ 
\multicolumn{1}{l}{Number of observations} & $5,598$ & $5,286$ & $4,973$ & $%
4,662$ & $4,354$ & $4,048$ & $3,749$ & $3,455$ & $3,170$ & $2,903$ & $2,651$
\\ \hline \hline
LHS variable: & \multicolumn{11}{c}{Log(Revenue)} \\ 
\cline{2-3}\cline{3-3}\cline{4-6}\cline{7-12}
Lag of RHS variables: & $0$ & $1$ & $2$ & $3$ & $4$ & $5$ & $6$ & $7$ & $8$
& $9$ & $10$ \\ \hline
\multicolumn{1}{l}{Lagged flare length} & $0.15$ & $0.13$ & $0.11$ & $0.09$
& $0.06$ & $0.05$ & $0.03$ & $0.02$ & $0.02$ & $0.01$ & $-0.01$ \\ 
\multicolumn{1}{l}{} & $\left( 0.02\right) $ & $\left( 0.02\right) $ & $%
\left( 0.02\right) $ & $\left( 0.02\right) $ & $\left( 0.02\right) $ & $%
\left( 0.02\right) $ & $\left( 0.02\right) $ & $\left( 0.02\right) $ & $%
\left( 0.02\right) $ & $\left( 0.02\right) $ & $\left( 0.02\right) $ \\ 
\multicolumn{1}{l}{Lagged islands only} & $-0.10$ & $-0.13$ & $-0.17$ & $%
-0.21$ & $-0.24$ & $-0.25$ & $-0.25$ & $-0.23$ & $-0.22$ & $-0.21$ & $-0.26$
\\ 
\multicolumn{1}{l}{} & $\left( 0.06\right) $ & $\left( 0.06\right) $ & $%
\left( 0.05\right) $ & $\left( 0.05\right) $ & $\left( 0.05\right) $ & $%
\left( 0.05\right) $ & $\left( 0.05\right) $ & $\left( 0.05\right) $ & $%
\left( 0.04\right) $ & $\left( 0.04\right) $ & $\left( 0.04\right) $ \\ 
\multicolumn{1}{l}{Constant} & $7.28$ & $7.38$ & $7.48$ & $7.57$ & $7.66$ & $%
7.74$ & $7.82$ & $7.89$ & $7.95$ & $8.01$ & $8.09$ \\ 
\multicolumn{1}{l}{} & $\left( 0.02\right) $ & $\left( 0.02\right) $ & $%
\left( 0.02\right) $ & $\left( 0.02\right) $ & $\left( 0.02\right) $ & $%
\left( 0.02\right) $ & $\left( 0.02\right) $ & $\left( 0.02\right) $ & $%
\left( 0.02\right) $ & $\left( 0.02\right) $ & $\left( 0.02\right) $ \\ 
\multicolumn{1}{l}{Firm fixed effects} & Yes & Yes & Yes & Yes & Yes & Yes & 
Yes & Yes & Yes & Yes & Yes \\ 
\multicolumn{1}{l}{Within-firm $R^{2}$} & $0.029$ & $0.026$ & $0.026$ & $%
0.026$ & $0.023$ & $0.023$ & $0.021$ & $0.019$ & $0.019$ & $0.018$ & $0.022$
\\ 
\multicolumn{1}{l}{Between-firm $R^{2}$} & $0.313$ & $0.277$ & $0.197$ & $%
0.079$ & $0.002$ & $0.008$ & $0.049$ & $0.074$ & $0.078$ & $0.128$ & $0.225$
\\ 
\multicolumn{1}{l}{Overall $R^{2}$} & $0.141$ & $0.105$ & $0.068$ & $0.028$
& $0.003$ & $0.000$ & $0.008$ & $0.018$ & $0.016$ & $0.031$ & $0.080$ \\ 
\multicolumn{1}{l}{Number of observations} & $5,598$ & $5,286$ & $4,973$ & $%
4,662$ & $4,354$ & $4,048$ & $3,749$ & $3,455$ & $3,170$ & $2,903$ & $2,651$
\\ 
\multicolumn{1}{l}{Number of firms} & $317$ & $314$ & $311$ & $308$ & $306$
& $299$ & $294$ & $285$ & $267$ & $253$ & $241$ \\ \hline \hline
\end{tabular}
\begin{minipage}{475pt}
{\fontsize{9pt}{9pt}\selectfont \smallskip  \textit{Note}: Standard errors are in parentheses.}
\end{minipage}
\end{center}
\end{table}

\begin{table}[tbh]
\caption{EBIT Regression Using 1980--2005 Panel Data}
\label{Table - Panel EBIT Reg}
\begin{center}
\fontsize{9pt}{11pt}\selectfont%
\begin{tabular}{cccccccccccc}
\hline \hline
LHS variable: & \multicolumn{11}{c}{Log(EBIT)} \\ 
\cline{2-4}\cline{3-4}\cline{5-7}\cline{8-12}
Lag of RHS variables: & $0$ & $1$ & $2$ & $3$ & $4$ & $5$ & $6$ & $7$ & $8$
& $9$ & $10$ \\ \hline
\multicolumn{1}{l}{Lagged flare length} & $0.82$ & $0.80$ & $0.79$ & $0.78$
& $0.82$ & $0.83$ & $0.85$ & $0.84$ & $0.85$ & $0.88$ & $0.87$ \\ 
\multicolumn{1}{l}{} & $\left( 0.02\right) $ & $\left( 0.02\right) $ & $%
\left( 0.02\right) $ & $\left( 0.02\right) $ & $\left( 0.03\right) $ & $%
\left( 0.03\right) $ & $\left( 0.03\right) $ & $\left( 0.03\right) $ & $%
\left( 0.03\right) $ & $\left( 0.04\right) $ & $\left( 0.04\right) $ \\ 
\multicolumn{1}{l}{Lagged islands only} & $1.73$ & $1.70$ & $1.67$ & $1.62$
& $1.62$ & $1.61$ & $1.64$ & $1.59$ & $1.56$ & $1.51$ & $1.43$ \\ 
\multicolumn{1}{l}{} & $\left( 0.08\right) $ & $\left( 0.08\right) $ & $%
\left( 0.08\right) $ & $\left( 0.08\right) $ & $\left( 0.08\right) $ & $%
\left( 0.08\right) $ & $\left( 0.08\right) $ & $\left( 0.08\right) $ & $%
\left( 0.08\right) $ & $\left( 0.08\right) $ & $\left( 0.09\right) $ \\ 
\multicolumn{1}{l}{Constant} & $4.56$ & $4.64$ & $4.71$ & $4.79$ & $4.83$ & $%
4.88$ & $4.93$ & $4.99$ & $5.04$ & $5.08$ & $5.14$ \\ 
\multicolumn{1}{l}{} & $\left( 0.03\right) $ & $\left( 0.03\right) $ & $%
\left( 0.03\right) $ & $\left( 0.03\right) $ & $\left( 0.03\right) $ & $%
\left( 0.03\right) $ & $\left( 0.04\right) $ & $\left( 0.04\right) $ & $%
\left( 0.04\right) $ & $\left( 0.04\right) $ & $\left( 0.04\right) $ \\ 
\multicolumn{1}{l}{Firm fixed effects} & No & No & No & No & No & No & No & 
No & No & No & No \\ 
\multicolumn{1}{l}{$R^{2}$} & $0.244$ & $0.237$ & $0.237$ & $0.230$ & $0.232$
& $0.228$ & $0.228$ & $0.217$ & $0.218$ & $0.214$ & $0.205$ \\ 
\multicolumn{1}{l}{Adjusted $R^{2}$} & $0.243$ & $0.237$ & $0.236$ & $0.229$
& $0.232$ & $0.228$ & $0.228$ & $0.216$ & $0.218$ & $0.213$ & $0.204$ \\ 
\multicolumn{1}{l}{Number of observations} & $5,129$ & $4,846$ & $4,563$ & $%
4,281$ & $4,007$ & $3,738$ & $3,472$ & $3,211$ & $2,954$ & $2,705$ & $2,476$
\\ \hline \hline
LHS variable: & \multicolumn{11}{c}{Log(EBIT)} \\ \cline{2-12}
Lag of RHS variables: & $0$ & $1$ & $2$ & $3$ & $4$ & $5$ & $6$ & $7$ & $8$
& $9$ & $10$ \\ \hline
\multicolumn{1}{l}{Lagged flare length} & $0.23$ & $0.20$ & $0.19$ & $0.17$
& $0.18$ & $0.16$ & $0.16$ & $0.12$ & $0.14$ & $0.14$ & $0.14$ \\ 
& $\left( 0.02\right) $ & $\left( 0.03\right) $ & $\left( 0.03\right) $ & $%
\left( 0.03\right) $ & $\left( 0.03\right) $ & $\left( 0.03\right) $ & $%
\left( 0.04\right) $ & $\left( 0.04\right) $ & $\left( 0.04\right) $ & $%
\left( 0.04\right) $ & $\left( 0.05\right) $ \\ 
\multicolumn{1}{l}{Lagged islands only} & $0.29$ & $0.24$ & $0.20$ & $0.13$
& $0.13$ & $0.11$ & $0.12$ & $0.05$ & $0.06$ & $-0.01$ & $-0.08$ \\ 
\multicolumn{1}{l}{} & $\left( 0.08\right) $ & $\left( 0.08\right) $ & $%
\left( 0.08\right) $ & $\left( 0.08\right) $ & $\left( 0.08\right) $ & $%
\left( 0.08\right) $ & $\left( 0.09\right) $ & $\left( 0.09\right) $ & $%
\left( 0.09\right) $ & $\left( 0.10\right) $ & $\left( 0.10\right) $ \\ 
\multicolumn{1}{l}{Lagged log patents} & $0.54$ & $0.54$ & $0.54$ & $0.54$ & 
$0.53$ & $0.53$ & $0.53$ & $0.54$ & $0.53$ & $0.53$ & $0.53$ \\ 
\multicolumn{1}{l}{} & $\left( 0.01\right) $ & $\left( 0.01\right) $ & $%
\left( 0.02\right) $ & $\left( 0.02\right) $ & $\left( 0.02\right) $ & $%
\left( 0.02\right) $ & $\left( 0.02\right) $ & $\left( 0.02\right) $ & $%
\left( 0.02\right) $ & $\left( 0.02\right) $ & $\left( 0.02\right) $ \\ 
\multicolumn{1}{l}{Constant} & $3.57$ & $3.65$ & $3.74$ & $3.83$ & $3.92$ & $%
3.99$ & $4.07$ & $4.14$ & $4.21$ & $4.28$ & $4.36$ \\ 
\multicolumn{1}{l}{} & $\left( 0.04\right) $ & $\left( 0.04\right) $ & $%
\left( 0.04\right) $ & $\left( 0.04\right) $ & $\left( 0.04\right) $ & $%
\left( 0.04\right) $ & $\left( 0.04\right) $ & $\left( 0.05\right) $ & $%
\left( 0.05\right) $ & $\left( 0.05\right) $ & $\left( 0.05\right) $ \\ 
\multicolumn{1}{l}{Firm fixed effects} & No & No & No & No & No & No & No & 
No & No & No & No \\ 
\multicolumn{1}{l}{$R^{2}$} & $0.403$ & $0.403$ & $0.400$ & $0.397$ & $0.391$
& $0.384$ & $0.379$ & $0.371$ & $0.368$ & $0.359$ & $0.349$ \\ 
\multicolumn{1}{l}{Adjusted $R^{2}$} & $0.402$ & $0.402$ & $0.399$ & $0.397$
& $0.390$ & $0.383$ & $0.378$ & $0.371$ & $0.367$ & $0.358$ & $0.348$ \\ 
\multicolumn{1}{l}{Number of observations} & $5,129$ & $4,846$ & $4,563$ & $%
4,281$ & $4,007$ & $3,738$ & $3,472$ & $3,211$ & $2,954$ & $2,705$ & $2,476$
\\ \hline \hline
LHS variable: & \multicolumn{11}{c}{Log(EBIT)} \\ 
\cline{2-3}\cline{3-3}\cline{4-6}\cline{7-12}
Lag of RHS variables: & $0$ & $1$ & $2$ & $3$ & $4$ & $5$ & $6$ & $7$ & $8$
& $9$ & $10$ \\ \hline
\multicolumn{1}{l}{Lagged flare length} & $0.14$ & $0.12$ & $0.12$ & $0.08$
& $0.08$ & $0.05$ & $0.05$ & $0.04$ & $0.06$ & $0.05$ & $0.02$ \\ 
\multicolumn{1}{l}{} & $\left( 0.02\right) $ & $\left( 0.02\right) $ & $%
\left( 0.02\right) $ & $\left( 0.02\right) $ & $\left( 0.02\right) $ & $%
\left( 0.02\right) $ & $\left( 0.02\right) $ & $\left( 0.02\right) $ & $%
\left( 0.02\right) $ & $\left( 0.03\right) $ & $\left( 0.03\right) $ \\ 
\multicolumn{1}{l}{Lagged islands only} & $-0.21$ & $-0.21$ & $-0.24$ & $%
-0.32$ & $-0.28$ & $-0.29$ & $-0.22$ & $-0.21$ & $-0.17$ & $-0.19$ & $-0.24$
\\ 
\multicolumn{1}{l}{} & $\left( 0.07\right) $ & $\left( 0.07\right) $ & $%
\left( 0.07\right) $ & $\left( 0.06\right) $ & $\left( 0.06\right) $ & $%
\left( 0.07\right) $ & $\left( 0.07\right) $ & $\left( 0.06\right) $ & $%
\left( 0.06\right) $ & $\left( 0.06\right) $ & $\left( 0.06\right) $ \\ 
\multicolumn{1}{l}{Constant} & $5.22$ & $5.30$ & $5.37$ & $5.48$ & $5.53$ & $%
5.62$ & $5.67$ & $5.73$ & $5.77$ & $5.83$ & $5.91$ \\ 
\multicolumn{1}{l}{} & $\left( 0.02\right) $ & $\left( 0.02\right) $ & $%
\left( 0.02\right) $ & $\left( 0.02\right) $ & $\left( 0.02\right) $ & $%
\left( 0.02\right) $ & $\left( 0.03\right) $ & $\left( 0.03\right) $ & $%
\left( 0.03\right) $ & $\left( 0.03\right) $ & $\left( 0.03\right) $ \\ 
\multicolumn{1}{l}{Firm fixed effects} & Yes & Yes & Yes & Yes & Yes & Yes & 
Yes & Yes & Yes & Yes & Yes \\ 
\multicolumn{1}{l}{Within-firm $R^{2}$} & $0.026$ & $0.023$ & $0.026$ & $%
0.025$ & $0.024$ & $0.018$ & $0.013$ & $0.012$ & $0.013$ & $0.014$ & $0.014$
\\ 
\multicolumn{1}{l}{Between-firm $R^{2}$} & $0.221$ & $0.177$ & $0.128$ & $%
0.006$ & $0.017$ & $0.013$ & $0.003$ & $0.011$ & $0.003$ & $0.004$ & $0.072$
\\ 
\multicolumn{1}{l}{Overall $R^{2}$} & $0.082$ & $0.059$ & $0.046$ & $0.007$
& $0.011$ & $0.000$ & $0.000$ & $0.000$ & $0.004$ & $0.000$ & $0.011$ \\ 
\multicolumn{1}{l}{Number of observations} & $5,129$ & $4,846$ & $4,563$ & $%
4,281$ & $4,007$ & $3,738$ & $3,472$ & $3,211$ & $2,954$ & $2,705$ & $2,476$
\\ 
\multicolumn{1}{l}{Number of firms} & $314$ & $309$ & $306$ & $301$ & $299$
& $293$ & $285$ & $277$ & $260$ & $247$ & $234$ \\ \hline \hline
\end{tabular}
\begin{minipage}{475pt}
{\fontsize{9pt}{9pt}\selectfont \smallskip  \textit{Note}: Standard errors are in parentheses.}
\end{minipage}
\end{center}
\end{table}

\begin{table}[tbh]
\caption{Market-Value Regression Using 1980--2005 Panel Data}
\label{Table - Panel Mcap Reg}
\begin{center}
\fontsize{9pt}{11pt}\selectfont%
\begin{tabular}{cccccccccccc}
\hline \hline
LHS variable: & \multicolumn{11}{c}{Log(Market value)} \\ 
\cline{2-4}\cline{3-4}\cline{5-7}\cline{8-12}
Lag of RHS variables: & $0$ & $1$ & $2$ & $3$ & $4$ & $5$ & $6$ & $7$ & $8$
& $9$ & $10$ \\ \hline
\multicolumn{1}{l}{Lagged flare length} & $0.80$ & $0.79$ & $0.79$ & $0.79$
& $0.81$ & $0.83$ & $0.85$ & $0.84$ & $0.85$ & $0.88$ & $0.88$ \\ 
\multicolumn{1}{l}{} & $\left( 0.02\right) $ & $\left( 0.02\right) $ & $%
\left( 0.02\right) $ & $\left( 0.03\right) $ & $\left( 0.03\right) $ & $%
\left( 0.03\right) $ & $\left( 0.03\right) $ & $\left( 0.03\right) $ & $%
\left( 0.03\right) $ & $\left( 0.04\right) $ & $\left( 0.04\right) $ \\ 
\multicolumn{1}{l}{Lagged islands only} & $1.40$ & $1.41$ & $1.42$ & $1.43$
& $1.40$ & $1.39$ & $1.38$ & $1.38$ & $1.33$ & $1.35$ & $1.29$ \\ 
\multicolumn{1}{l}{} & $\left( 0.08\right) $ & $\left( 0.08\right) $ & $%
\left( 0.08\right) $ & $\left( 0.08\right) $ & $\left( 0.08\right) $ & $%
\left( 0.08\right) $ & $\left( 0.08\right) $ & $\left( 0.08\right) $ & $%
\left( 0.09\right) $ & $\left( 0.09\right) $ & $\left( 0.09\right) $ \\ 
\multicolumn{1}{l}{Constant} & $6.95$ & $7.01$ & $7.09$ & $7.16$ & $7.22$ & $%
7.28$ & $7.34$ & $7.40$ & $7.47$ & $7.51$ & $7.58$ \\ 
\multicolumn{1}{l}{} & $\left( 0.03\right) $ & $\left( 0.03\right) $ & $%
\left( 0.03\right) $ & $\left( 0.03\right) $ & $\left( 0.03\right) $ & $%
\left( 0.04\right) $ & $\left( 0.04\right) $ & $\left( 0.04\right) $ & $%
\left( 0.04\right) $ & $\left( 0.04\right) $ & $\left( 0.05\right) $ \\ 
\multicolumn{1}{l}{Firm fixed effects} & No & No & No & No & No & No & No & 
No & No & No & No \\ 
\multicolumn{1}{l}{$R^{2}$} & $0.208$ & $0.202$ & $0.200$ & $0.197$ & $0.193$
& $0.193$ & $0.189$ & $0.185$ & $0.183$ & $0.184$ & $0.174$ \\ 
\multicolumn{1}{l}{Adjusted $R^{2}$} & $0.208$ & $0.202$ & $0.200$ & $0.197$
& $0.193$ & $0.193$ & $0.189$ & $0.185$ & $0.183$ & $0.184$ & $0.173$ \\ 
\multicolumn{1}{l}{Number of observations} & $5,183$ & $4,970$ & $4,704$ & $%
4,416$ & $4,129$ & $3,842$ & $3,561$ & $3,285$ & $3,018$ & $2,769$ & $2,533$
\\ \hline \hline
LHS variable: & \multicolumn{11}{c}{Log(Market value)} \\ \cline{2-12}
Lag of RHS variables: & $0$ & $1$ & $2$ & $3$ & $4$ & $5$ & $6$ & $7$ & $8$
& $9$ & $10$ \\ \hline
\multicolumn{1}{l}{Lagged flare length} & $0.13$ & $0.09$ & $0.09$ & $0.08$
& $0.06$ & $0.06$ & $0.04$ & $0.04$ & $0.05$ & $0.06$ & $0.06$ \\ 
& $\left( 0.03\right) $ & $\left( 0.03\right) $ & $\left( 0.03\right) $ & $%
\left( 0.03\right) $ & $\left( 0.03\right) $ & $\left( 0.03\right) $ & $%
\left( 0.04\right) $ & $\left( 0.04\right) $ & $\left( 0.04\right) $ & $%
\left( 0.05\right) $ & $\left( 0.05\right) $ \\ 
\multicolumn{1}{l}{Lagged islands only} & $-0.23$ & $-0.25$ & $-0.24$ & $%
-0.23$ & $-0.29$ & $-0.30$ & $-0.32$ & $-0.31$ & $-0.35$ & $-0.32$ & $-0.38$
\\ 
\multicolumn{1}{l}{} & $\left( 0.08\right) $ & $\left( 0.08\right) $ & $%
\left( 0.08\right) $ & $\left( 0.08\right) $ & $\left( 0.08\right) $ & $%
\left( 0.09\right) $ & $\left( 0.09\right) $ & $\left( 0.09\right) $ & $%
\left( 0.10\right) $ & $\left( 0.10\right) $ & $\left( 0.11\right) $ \\ 
\multicolumn{1}{l}{Lagged log patents} & $0.62$ & $0.63$ & $0.62$ & $0.61$ & 
$0.62$ & $0.61$ & $0.61$ & $0.60$ & $0.59$ & $0.58$ & $0.58$ \\ 
\multicolumn{1}{l}{} & $\left( 0.01\right) $ & $\left( 0.01\right) $ & $%
\left( 0.02\right) $ & $\left( 0.02\right) $ & $\left( 0.02\right) $ & $%
\left( 0.02\right) $ & $\left( 0.02\right) $ & $\left( 0.02\right) $ & $%
\left( 0.02\right) $ & $\left( 0.02\right) $ & $\left( 0.02\right) $ \\ 
\multicolumn{1}{l}{Constant} & $5.80$ & $5.85$ & $5.95$ & $6.05$ & $6.15$ & $%
6.25$ & $6.35$ & $6.45$ & $6.54$ & $6.64$ & $6.74$ \\ 
\multicolumn{1}{l}{} & $\left( 0.04\right) $ & $\left( 0.04\right) $ & $%
\left( 0.04\right) $ & $\left( 0.04\right) $ & $\left( 0.04\right) $ & $%
\left( 0.04\right) $ & $\left( 0.04\right) $ & $\left( 0.05\right) $ & $%
\left( 0.05\right) $ & $\left( 0.05\right) $ & $\left( 0.05\right) $ \\ 
\multicolumn{1}{l}{Firm fixed effects} & No & No & No & No & No & No & No & 
No & No & No & No \\ 
\multicolumn{1}{l}{$R^{2}$} & $0.411$ & $0.413$ & $0.407$ & $0.400$ & $0.395$
& $0.387$ & $0.378$ & $0.366$ & $0.359$ & $0.346$ & $0.333$ \\ 
\multicolumn{1}{l}{Adjusted $R^{2}$} & $0.410$ & $0.412$ & $0.407$ & $0.400$
& $0.394$ & $0.387$ & $0.378$ & $0.365$ & $0.359$ & $0.345$ & $0.333$ \\ 
\multicolumn{1}{l}{Number of observations} & $5,183$ & $4,970$ & $4,704$ & $%
4,416$ & $4,129$ & $3,842$ & $3,561$ & $3,285$ & $3,018$ & $2,769$ & $2,533$
\\ \hline \hline
LHS variable: & \multicolumn{11}{c}{Log(Market value)} \\ 
\cline{2-3}\cline{3-3}\cline{4-6}\cline{7-12}
Lag of RHS variables: & $0$ & $1$ & $2$ & $3$ & $4$ & $5$ & $6$ & $7$ & $8$
& $9$ & $10$ \\ \hline
\multicolumn{1}{l}{Lagged flare length} & $0.19$ & $0.17$ & $0.16$ & $0.15$
& $0.13$ & $0.13$ & $0.12$ & $0.14$ & $0.13$ & $0.13$ & $0.10$ \\ 
\multicolumn{1}{l}{} & $\left( 0.02\right) $ & $\left( 0.02\right) $ & $%
\left( 0.02\right) $ & $\left( 0.02\right) $ & $\left( 0.02\right) $ & $%
\left( 0.02\right) $ & $\left( 0.03\right) $ & $\left( 0.03\right) $ & $%
\left( 0.03\right) $ & $\left( 0.03\right) $ & $\left( 0.03\right) $ \\ 
\multicolumn{1}{l}{Lagged islands only} & $-0.31$ & $-0.31$ & $-0.32$ & $%
-0.29$ & $-0.34$ & $-0.31$ & $-0.30$ & $-0.22$ & $-0.24$ & $-0.18$ & $-0.24$
\\ 
\multicolumn{1}{l}{} & $\left( 0.07\right) $ & $\left( 0.07\right) $ & $%
\left( 0.07\right) $ & $\left( 0.07\right) $ & $\left( 0.07\right) $ & $%
\left( 0.07\right) $ & $\left( 0.07\right) $ & $\left( 0.07\right) $ & $%
\left( 0.07\right) $ & $\left( 0.07\right) $ & $\left( 0.06\right) $ \\ 
\multicolumn{1}{l}{Constant} & $7.52$ & $7.59$ & $7.68$ & $7.76$ & $7.84$ & $%
7.92$ & $7.99$ & $8.04$ & $8.12$ & $8.17$ & $8.27$ \\ 
\multicolumn{1}{l}{} & $\left( 0.02\right) $ & $\left( 0.02\right) $ & $%
\left( 0.02\right) $ & $\left( 0.02\right) $ & $\left( 0.03\right) $ & $%
\left( 0.03\right) $ & $\left( 0.03\right) $ & $\left( 0.03\right) $ & $%
\left( 0.03\right) $ & $\left( 0.03\right) $ & $\left( 0.03\right) $ \\ 
\multicolumn{1}{l}{Firm fixed effects} & Yes & Yes & Yes & Yes & Yes & Yes & 
Yes & Yes & Yes & Yes & Yes \\ 
\multicolumn{1}{l}{Within-firm $R^{2}$} & $0.043$ & $0.039$ & $0.039$ & $%
0.035$ & $0.036$ & $0.037$ & $0.034$ & $0.032$ & $0.035$ & $0.031$ & $0.031$
\\ 
\multicolumn{1}{l}{Between-firm $R^{2}$} & $0.170$ & $0.140$ & $0.109$ & $%
0.098$ & $0.043$ & $0.043$ & $0.019$ & $0.053$ & $0.035$ & $0.054$ & $0.005$
\\ 
\multicolumn{1}{l}{Overall $R^{2}$} & $0.088$ & $0.069$ & $0.054$ & $0.049$
& $0.027$ & $0.027$ & $0.018$ & $0.035$ & $0.027$ & $0.036$ & $0.009$ \\ 
\multicolumn{1}{l}{Number of observations} & $5,183$ & $4,970$ & $4,704$ & $%
4,416$ & $4,129$ & $3,842$ & $3,561$ & $3,285$ & $3,018$ & $2,769$ & $2,533$
\\ 
\multicolumn{1}{l}{Number of firms} & $298$ & $296$ & $294$ & $291$ & $289$
& $283$ & $278$ & $269$ & $251$ & $238$ & $227$ \\ \hline \hline
\end{tabular}
\begin{minipage}{475pt}
{\fontsize{9pt}{9pt}\selectfont \smallskip  \textit{Note}: Standard errors are in parentheses.}
\end{minipage}
\end{center}
\end{table}

The top panels of these tables show results that are similar to columns 1, 4, and 7 of the baseline results in Table \ref{Table - Regressions} (in section 5.2). Likewise, the middle panels show results that are comparable to columns 3, 6, and 9 of Table \ref{Table - Regressions}. The statistical relationships are surprisingly persistent over time. Even a decade-long lag ($\tau=10$) leads to relatively minor changes in magnitude. Thus, the length of unique technological trajectory seems to capture relevant information for long-term firm dynamics.

The bottom panels of Tables \ref{Table - Panel Revenue Reg}--\ref{Table - Panel Mcap Reg} show the inclusion of firm fixed effects ``cannibalizes'' the contributions of flares and islands, which suggests high collinearity between them.\footnote{In particular, the coefficients for islands-only observations become either statistically insignificant or negative. Because only 10\% of firm-years is islands-only and they tend to be special cases (either large conglomerates or extremely niche firms), their interpretation with fixed effects is difficult.} These results suggest flares and islands are closely related to persistent firm heterogeneity, such as firms' underlying R\&D capabilities $\omega^{x}_{i,t}$.

\paragraph{Out-of-Sample Predictions.}

Table \ref{Table - Forecasting} reports the fit of out-of-sample predictions. The top panel uses the 1980--2004 data for the estimation of equation \ref{eq - panel regression} with $\tau = 1$ and assesses the fit of its predictions in the 2005 subsample. The middle and bottom panels do the same with $\tau=5$ and $\tau=10$, respectively, with correspondingly smaller estimation subsamples (1980--2000 and 1980--1995).

Three findings emerge. First, columns 1, 4, and 7 of each panel show the correlation between the predicted and actual values is 0.425--0.532. That is, the out-of-sample predictions based solely on flares and islands can achieve reasonably high correlations with the prediction targets. Second, the difference between columns 2 and 3 suggests the inclusion of flares and islands improves the prediction of revenues above and beyond what patent count alone can achieve (an increase from 0.545 to 0.571). Qualitatively similar results hold for EBIT (columns 5 and 6) and market value (8 and 9), albeit with smaller magnitudes.

Third, we find similar patterns with an alternative measure of fit, the mean squared error,
\begin{equation}
MSE=\frac{1}{N_{p}}\sum_{obs=1}^{N_{p}}\left( \ln \left( y_{obs}\right) -\ln
\left( \hat{y}_{obs}\right) \right) ^{2},
\label{eq - MSE}
\end{equation}%
where $N_p$ is the size of the prediction subsample, $obs$ is an index for observations, and $\hat{y}_{obs}$ is the predicted value of a performance measure. We can assess the contributions of flares and islands by comparing the MSEs in columns 1, 4, and 7 with their respective ``null'' predictions, which we define as the means of the estimation subsample (i.e., the fitted value in a ``regression'' with no regressors other than a constant). The ``improvement over null (\%)'' rows report the reductions in MSEs as a percentage of the null MSE. These improvements are sizeable: 13.6\%--15.5\% in the top panel, 21.3\%--33.5\% in the middle panel, and 34.0\%--43.8\% in the bottom panel.

In summary, despite being relatively coarse, discrete, slow-moving measures of firms' activities, flares and islands seem to do a surprisingly good job in out-of-sample predictions.

\begin{table}[tbh]
\caption{Out-of-Sample Predictions}
\begin{center}
\fontsize{9pt}{11pt}\selectfont%
\begin{tabular}{cccccccccccc}
\hline \hline
LHS variable: & \multicolumn{3}{c}{Log(Revenue)} &  & \multicolumn{3}{c}{
Log(EBIT)} &  & \multicolumn{3}{c}{Log(Market value)} \\ 
\cline{2-4}\cline{3-4}\cline{6-8}\cline{10-12}
& (1) & (2) & (3) &  & (4) & (5) & (6) &  & (7) & (8) & (9) \\ \hline
\multicolumn{1}{l}{Lag (1) flare length} & $0.92$ & $-$ & $0.34$ &  & $0.81$
& $-$ & $0.20$ &  & $0.81$ & $-$ & $0.09$ \\ 
\multicolumn{1}{l}{} & $\left( 0.02\right) $ & $\left( -\right) $ & $\left(
0.03\right) $ &  & $\left( 0.02\right) $ & $\left( -\right) $ & $\left(
0.03\right) $ &  & $\left( 0.02\right) $ & $\left( -\right) $ & $\left(
0.03\right) $ \\ 
\multicolumn{1}{l}{Lag (1) islands only} & $2.15$ & $-$ & $0.75$ &  & $1.74$
& $-$ & $0.26$ &  & $1.46$ & $-$ & $-0.23$ \\ 
\multicolumn{1}{l}{} & $\left( 0.08\right) $ & $\left( -\right) $ & $\left(
0.08\right) $ &  & $\left( 0.08\right) $ & $\left( -\right) $ & $\left(
0.08\right) $ &  & $\left( 0.08\right) $ & $\left( -\right) $ & $\left(
0.08\right) $ \\ 
\multicolumn{1}{l}{Lag (1) log patents} & $-$ & $0.65$ & $0.51$ &  & $-$ & $%
0.62$ & $0.54$ &  & $-$ & $0.65$ & $0.63$ \\ 
\multicolumn{1}{l}{} & $\left( -\right) $ & $\left( 0.01\right) $ & $\left(
0.02\right) $ &  & $\left( -\right) $ & $\left( 0.01\right) $ & $\left(
0.02\right) $ &  & $\left( -\right) $ & $\left( 0.01\right) $ & $\left(
0.02\right) $ \\ 
\multicolumn{1}{l}{Constant} & $6.57$ & $5.57$ & $5.66$ &  & $4.59$ & $3.56$
& $3.60$ &  & $6.95$ & $5.78$ & $5.79$ \\ 
\multicolumn{1}{l}{} & $\left( 0.03\right) $ & $\left( 0.04\right) $ & $%
\left( 0.04\right) $ &  & $\left( 0.03\right) $ & $\left( 0.04\right) $ & $%
\left( 0.04\right) $ &  & $\left( 0.03\right) $ & $\left( 0.04\right) $ & $%
\left( 0.04\right) $ \\ 
\multicolumn{1}{l}{Estimation subsample} & \multicolumn{3}{c}{1980--2004} & 
& \multicolumn{3}{c}{1980--2004} &  & \multicolumn{3}{c}{1980--2004} \\ 
\multicolumn{1}{l}{In-sample $R^{2}$} & $0.299$ & $0.412$ & $0.433$ &  & $%
0.244$ & $0.401$ & $0.409$ &  & $0.206$ & $0.414$ & $0.418$ \\ 
\multicolumn{1}{l}{In-sample adjusted $R^{2}$} & $0.299$ & $0.412$ & $0.432$
&  & $0.243$ & $0.401$ & $0.408$ &  & $0.206$ & $0.414$ & $0.418$ \\ 
\multicolumn{1}{l}{In-sample observations} & $4,726$ & $4,726$ & $4,726$ & 
& $4,590$ & $4,590$ & $4,590$ &  & $4,706$ & $4,706$ & $4,706$ \\ 
\multicolumn{1}{l}{Prediction subsample} & \multicolumn{3}{c}{2005} &  & 
\multicolumn{3}{c}{2005} &  & \multicolumn{3}{c}{2005} \\ 
\multicolumn{1}{l}{Out-of-sample correlation} & $0.477$ & $0.545$ & $0.571$
&  & $0.439$ & $0.572$ & $0.579$ &  & $0.437$ & $0.598$ & $0.599$ \\ 
\multicolumn{1}{l}{Out-of-sample MSE} & $2.978$ & $2.546$ & $2.500$ &  & $%
3.315$ & $2.681$ & $2.668$ &  & $3.476$ & $2.764$ & $2.727$ \\ 
\multicolumn{1}{l}{Out-of-sample observations} & $279$ & $279$ & $279$ &  & $%
256$ & $256$ & $256$ &  & $264$ & $264$ & $264$ \\ \hline \hline
LHS variable: & \multicolumn{3}{c}{Log(Revenue)} &  & \multicolumn{3}{c}{
Log(EBIT)} &  & \multicolumn{3}{c}{Log(Market value)} \\ 
\cline{2-4}\cline{3-4}\cline{6-8}\cline{10-12}
& (1) & (2) & (3) &  & (4) & (5) & (6) &  & (7) & (8) & (9) \\ \hline
\multicolumn{1}{l}{Lag (5) flare length} & $1.04$ & $-$ & $0.28$ &  & $0.96$
& $-$ & $0.18$ &  & $0.98$ & $-$ & $0.09$ \\ 
\multicolumn{1}{l}{} & $\left( 0.03\right) $ & $\left( -\right) $ & $\left(
0.04\right) $ &  & $\left( 0.04\right) $ & $\left( -\right) $ & $\left(
0.04\right) $ &  & $\left( 0.04\right) $ & $\left( -\right) $ & $\left(
0.05\right) $ \\ 
\multicolumn{1}{l}{Lag (5) islands only} & $1.94$ & $-$ & $0.39$ &  & $1.70$
& $-$ & $0.09$ &  & $1.54$ & $-$ & $-0.27$ \\ 
\multicolumn{1}{l}{} & $\left( 0.08\right) $ & $\left( -\right) $ & $\left(
0.09\right) $ &  & $\left( 0.08\right) $ & $\left( -\right) $ & $\left(
0.10\right) $ &  & $\left( 0.09\right) $ & $\left( -\right) $ & $\left(
0.10\right) $ \\ 
\multicolumn{1}{l}{Lag (5) log patents} & $-$ & $0.63$ & $0.53$ &  & $-$ & $%
0.60$ & $0.56$ &  & $-$ & $0.63$ & $0.63$ \\ 
\multicolumn{1}{l}{} & $\left( -\right) $ & $\left( 0.01\right) $ & $\left(
0.02\right) $ &  & $\left( -\right) $ & $\left( 0.01\right) $ & $\left(
0.02\right) $ &  & $\left( -\right) $ & $\left( 0.02\right) $ & $\left(
0.02\right) $ \\ 
\multicolumn{1}{l}{Constant} & $6.76$ & $5.95$ & $5.99$ &  & $4.66$ & $3.83$
& $3.85$ &  & $6.96$ & $6.06$ & $6.05$ \\ 
\multicolumn{1}{l}{} & $\left( 0.04\right) $ & $\left( 0.05\right) $ & $%
\left( 0.05\right) $ &  & $\left( 0.04\right) $ & $\left( 0.05\right) $ & $%
\left( 0.05\right) $ &  & $\left( 0.04\right) $ & $\left( 0.05\right) $ & $%
\left( 0.05\right) $ \\ 
\multicolumn{1}{l}{Estimation subsample} & \multicolumn{3}{c}{1980--2000} & 
& \multicolumn{3}{c}{1980--2000} &  & \multicolumn{3}{c}{1980--2000} \\ 
\multicolumn{1}{l}{In-sample $R^{2}$} & $0.298$ & $0.428$ & $0.437$ &  & $%
0.248$ & $0.396$ & $0.401$ &  & $0.220$ & $0.399$ & $0.405$ \\ 
\multicolumn{1}{l}{In-sample adjusted $R^{2}$} & $0.297$ & $0.428$ & $0.437$
&  & $0.247$ & $0.396$ & $0.400$ &  & $0.220$ & $0.399$ & $0.405$ \\ 
\multicolumn{1}{l}{In-sample observations} & $2,762$ & $2,762$ & $2,762$ & 
& $2,640$ & $2,640$ & $2,640$ &  & $2,630$ & $2,630$ & $2,6530$ \\ 
\multicolumn{1}{l}{Prediction subsample} & \multicolumn{3}{c}{2001--2005} & 
& \multicolumn{3}{c}{2001--2005} &  & \multicolumn{3}{c}{2001--2005} \\ 
\multicolumn{1}{l}{Out-of-sample correlation} & $0.532$ & $0.581$ & $0.603$
&  & $0.460$ & $0.581$ & $0.587$ &  & $0.425$ & $0.587$ & $0.586$ \\ 
\multicolumn{1}{l}{Out-of-sample MSE} & $2.465$ & $2.161$ & $2.087$ &  & $%
2.953$ & $2.316$ & $2.302$ &  & $3.386$ & $2.474$ & $2.461$ \\ 
\multicolumn{1}{l}{Out-of-sample observations} & $1,286$ & $1,286$ & $1,286$
&  & $1,098$ & $1,098$ & $1,098$ &  & $1,212$ & $1,212$ & $1,212$ \\ 
\hline \hline
LHS variable: & \multicolumn{3}{c}{Log(Revenue)} &  & \multicolumn{3}{c}{
Log(EBIT)} &  & \multicolumn{3}{c}{Log(Market value)} \\ 
\cline{2-4}\cline{3-4}\cline{6-8}\cline{10-12}
& (1) & (2) & (3) &  & (4) & (5) & (6) &  & (7) & (8) & (9) \\ \hline
\multicolumn{1}{l}{Lag (10) flare length} & $1.05$ & $-$ & $0.20$ &  & $1.00$
& $-$ & $0.06$ &  & $0.97$ & $-$ & $0.08$ \\ 
\multicolumn{1}{l}{} & $\left( 0.08\right) $ & $\left( -\right) $ & $\left(
0.09\right) $ &  & $\left( 0.09\right) $ & $\left( -\right) $ & $\left(
0.10\right) $ &  & $\left( 0.09\right) $ & $\left( -\right) $ & $\left(
0.09\right) $ \\ 
\multicolumn{1}{l}{Lag (10) islands only} & $1.55$ & $-$ & $0.21$ &  & $1.51$
& $-$ & $0.00$ &  & $1.41$ & $-$ & $-0.02$ \\ 
\multicolumn{1}{l}{} & $\left( 0.12\right) $ & $\left( -\right) $ & $\left(
0.13\right) $ &  & $\left( 0.13\right) $ & $\left( -\right) $ & $\left(
0.14\right) $ &  & $\left( 0.13\right) $ & $\left( -\right) $ & $\left(
0.14\right) $ \\ 
\multicolumn{1}{l}{Lag (10) log patents} & $-$ & $0.56$ & $0.52$ &  & $-$ & $%
0.58$ & $0.58$ &  & $-$ & $0.56$ & $0.55$ \\ 
\multicolumn{1}{l}{} & $\left( -\right) $ & $\left( 0.02\right) $ & $\left(
0.03\right) $ &  & $\left( -\right) $ & $\left( 0.03\right) $ & $\left(
0.03\right) $ &  & $\left( -\right) $ & $\left( 0.02\right) $ & $\left(
0.03\right) $ \\ 
\multicolumn{1}{l}{Constant} & $6.97$ & $6.34$ & $6.32$ &  & $4.72$ & $4.00$
& $3.99$ &  & $7.03$ & $6.32$ & $6.32$ \\ 
\multicolumn{1}{l}{} & $\left( 0.07\right) $ & $\left( 0.07\right) $ & $%
\left( 0.08\right) $ &  & $\left( 0.08\right) $ & $\left( 0.08\right) $ & $%
\left( 0.08\right) $ &  & $\left( 0.08\right) $ & $\left( 0.08\right) $ & $%
\left( 0.08\right) $ \\ 
\multicolumn{1}{l}{Estimation subsample} & \multicolumn{3}{c}{1980--1995} & 
& \multicolumn{3}{c}{1980--1995} &  & \multicolumn{3}{c}{1980--1995} \\ 
\multicolumn{1}{l}{In-sample $R^{2}$} & $0.228$ & $0.422$ & $0.426$ &  & $%
0.184$ & $0.396$ & $0.396$ &  & $0.179$ & $0.392$ & $0.393$ \\ 
\multicolumn{1}{l}{In-sample adjusted $R^{2}$} & $0.226$ & $0.422$ & $0.424$
&  & $0.182$ & $0.395$ & $0.394$ &  & $0.177$ & $0.391$ & $0.391$ \\ 
\multicolumn{1}{l}{In-sample observations} & $806$ & $806$ & $806$ &  & $791$
& $791$ & $791$ &  & $793$ & $793$ & $793$ \\ 
\multicolumn{1}{l}{Prediction subsample} & \multicolumn{3}{c}{1996--2005} & 
& \multicolumn{3}{c}{1996--2005} &  & \multicolumn{3}{c}{1996--2005} \\ 
\multicolumn{1}{l}{Out-of-sample correlation} & $0.526$ & $0.611$ & $0.618$
&  & $0.459$ & $0.569$ & \thinspace $0.572$ &  & $0.429$ & $0.559$ & $0.562$
\\ 
\multicolumn{1}{l}{Out-of-sample MSE} & $2.268$ & $1.919$ & $1.893$ &  & $%
2.725$ & $2.302$ & $2.286$ &  & $3.417$ & $2.854$ & $2.821$ \\ 
\multicolumn{1}{l}{Out-of-sample observations} & $1,845$ & $1,845$ & $1,845$
&  & $1,685$ & $1,685$ & $1,685$ &  & $1,740$ & $1,740$ & $1,740$ \\ 
\hline \hline
\end{tabular}
\begin{minipage}{475pt}
{\fontsize{9pt}{9pt}\selectfont \smallskip  \textit{Note}: Standard errors are in parentheses. Out-of-sample MSE is the mean of squared prediction errors.}
\end{minipage}
\end{center}
\label{Table - Forecasting}
\end{table}

\clearpage

\section*{Appendix G \ Comparison with Other Measures}

This section reports full results for section 5.5, which compares flare length with other measures: five network-centrality measures and the Jaffe measure of technological distance.

\subsection*{G.1 \ Network Centrality Measures}

Flare length is not the only way to measure firms’ innovation histories on a graph. Various measures of network centrality offer more conventional alternatives. We consider five of them: degree, closeness, harmonic, betweenness, and eigenvector centralities.

\begin{table}[tbh]
\caption{Correlations between Firms' Performances, Flare Length, and
Centrality Measures}
\begin{center}
\fontsize{9pt}{11pt}\selectfont%
\begin{tabular}{cccccccc}
\hline \hline
Pair-wise correlation & Flare & Number of & Degree & Closeness & Harmonic & 
Betweenness & Eigenvector \\ 
coefficients & length & nodes & centrality & centrality & centrality & 
centrality & centrality \\ \hline
\multicolumn{1}{l}{Log(Revenue)} & $0.506$ & $0.549$ & $-0.432$ & $-0.461$ & 
$-0.462$ & $-0.371$ & $-0.218$ \\ 
\multicolumn{1}{l}{Log(EBIT)} & $0.499$ & $0.519$ & $-0.421$ & $-0.403$ & $%
-0.403$ & $-0.404$ & $-0.183$ \\ 
\multicolumn{1}{l}{Log(Market value)} & $0.476$ & $0.516$ & $-0.417$ & $%
-0.423$ & $-0.421$ & $-0.427$ & $-0.206$ \\ 
\multicolumn{1}{l}{Flare length} & $1.000$ & $0.844$ & $-0.687$ & $-0.696$ & 
$-0.719$ & $-0.657$ & $-0.601$ \\ 
\multicolumn{1}{l}{Number of nodes} & $0.844$ & $1.000$ & $-0.707$ & $-0.844$
& $-0.841$ & $-0.649$ & $-0.577$ \\ \hline \hline
\end{tabular}
\begin{minipage}{475pt}
{\fontsize{9pt}{9pt}\selectfont \smallskip  \textit{Note}: For each centrality measure, we take the minimum of all nodes that contain a given firm. Other summary statistics, such as mean and maximum, correlate less strongly with the performance metrics.}
\end{minipage}
\end{center}
\label{Table - centrality correlations}
\end{table}

\begin{table}[tbh!!!!]
\caption{Firm Count by Number of Nodes}
\begin{center}
\fontsize{9pt}{11pt}\selectfont%
\begin{tabular}{lccccccccccc}
\hline \hline
Number of nodes & 2 & 3 & 4 & 5 & 6 & 7--8 & 9--10 & 11--15 & 16--20 & 21--30 & 31+ \\ \hline
Frequency & $89$ & $26$ & $31$ & $15$ & $35$ & $29$ & $19$ & $24$ & $23$ & $28$ & $11$ \\ 
Percentage & $26.97$ & $7.88$ & $9.39$ & $4.55$ & $10.61$ & $8.79$ & $5.76$
& $7.27$ & $6.97$ & $8.48$ & $3.33$ \\ 
Cumulative \% & $26.97$ & $34.85$ & $44.24$ & $48.79$ & $59.39$ & $68.18$ & $73.94$ & $81.21$ & $88.18$ & $96.67$ & $100.00$ \\ \hline \hline
\end{tabular}
\begin{minipage}{475pt}
{\fontsize{9pt}{9pt}\selectfont \smallskip  \textit{Note}: Some firms appear in more than 30 nodes because each firm-year observation may appear in 1--4 nodes due to the overlap between cover elements.}
\end{minipage}
\end{center}
\label{Table - Number of nodes}
\end{table}

Table \ref{Table - centrality correlations} reports the correlation coefficients between the three performance metrics, flare length, and the five centrality measures. We also include the number of nodes in which each firm appears, which is among the simplest statistics to characterize a firm’s history (Table \ref{Table - Number of nodes} reports its frequency distribution). Three findings emerge. First, flare length is strongly correlated with financial performances (0.506, 0.499, and 0.476). Note the number of nodes shows even stronger correlations, even though it underperforms flare length in regressions (see below). Second, the firms’ centralities are negatively correlated with their financial results, because centrality is an inverse measure of uniqueness in this context. Hence, high centrality means low differentiation and low profitability. Third, their correlations are not as strong as that of flare length (i.e., smaller in absolute value).

Tables \ref{Table - centrality reg (revenue)}--\ref{Table - centrality reg (mcap)} compare flare length and the centrality measures in regressions that control for the firm’s total patent count. Column 1 of each table is our baseline result with flare length and the islands-only dummy variable. Its adjusted $R^2$ is equal to or higher than all other columns’, which suggests flares and islands capture more relevant, original information than the simple count of nodes or centralities. Thus, flare length outperforms these conventional network centralities.\footnote{One can improve the fit of the latter by calculating their demeaned ranges (i.e., maximum minus minimum minus mean). This statistic captures some notion of the length of unique trajectories: the range reflects the length of the firm’s outward/inward move, and the mean reflects its overall position relative to others. However, we find their adjusted $R^2$s are still lower than our baseline results based on flares and islands.}

\begin{table}[tbh]
\caption{Revenue Regression with Centrality Measures}
\begin{center}
\fontsize{9pt}{11pt}\selectfont%
\begin{tabular}{cccccccc}
\hline \hline
LHS variable: & \multicolumn{7}{c}{Log(Revenue)} \\ 
\cline{2-4}\cline{3-4}\cline{5-7}\cline{8-8}
& (1) & (2) & (3) & (4) & (5) & (6) & (7) \\ 
& Baseline &  &  &  &  &  &  \\ \hline
\multicolumn{1}{l}{Flare length} & $0.34$ & $-$ & $-$ & $-$ & $-$ & $-$ & $-$
\\ 
\multicolumn{1}{l}{} & $\left( 0.08\right) $ & $\left( -\right) $ & $\left(
-\right) $ & $\left( -\right) $ & $\left( -\right) $ & $\left( -\right) $ & $%
\left( -\right) $ \\ 
\multicolumn{1}{l}{Islands only} & $0.96$ & $-$ & $-$ & $-$ & $-$ & $-$ & $-$
\\ 
\multicolumn{1}{l}{} & $\left( 0.84\right) $ & $\left( -\right) $ & $\left(
-\right) $ & $\left( -\right) $ & $\left( -\right) $ & $\left( -\right) $ & $%
\left( -\right) $ \\ 
\multicolumn{1}{l}{Number of nodes} & $-$ & $0.06$ & $-$ & $-$ & $-$ & $-$ & 
$-$ \\ 
\multicolumn{1}{l}{} & $\left( -\right) $ & $\left( 0.01\right) $ & $\left(
-\right) $ & $\left( -\right) $ & $\left( -\right) $ & $\left( -\right) $ & $%
\left( -\right) $ \\ 
\multicolumn{1}{l}{Degree centrality} & $-$ & $-$ & $-359.4$ & $-$ & $-$ & $-
$ & $-$ \\ 
\multicolumn{1}{l}{} & $\left( -\right) $ & $\left( -\right) $ & $\left(
96.2\right) $ & $\left( -\right) $ & $\left( -\right) $ & $\left( -\right) $
& $\left( -\right) $ \\ 
\multicolumn{1}{l}{Closeness centrality} & $-$ & $-$ & $-$ & $-15.71$ & $-$
& $-$ & $-$ \\ 
\multicolumn{1}{l}{} & $\left( -\right) $ & $\left( -\right) $ & $\left(
-\right) $ & $\left( 4.18\right) $ & $\left( -\right) $ & $\left( -\right) $
& $\left( -\right) $ \\ 
\multicolumn{1}{l}{Harmonic centrality} & $-$ & $-$ & $-$ & $-$ & $-0.010$ & 
$-$ & $-$ \\ 
\multicolumn{1}{l}{} & $\left( -\right) $ & $\left( -\right) $ & $\left(
-\right) $ & $\left( -\right) $ & $\left( 0.003\right) $ & $\left( -\right) $
& $\left( -\right) $ \\ 
\multicolumn{1}{l}{Betweenness centrality} & $-$ & $-$ & $-$ & $-$ & $-$ & $%
-104.2$ & $-$ \\ 
\multicolumn{1}{l}{} & $\left( -\right) $ & $\left( -\right) $ & $\left(
-\right) $ & $\left( -\right) $ & $\left( -\right) $ & $\left( 55.5\right) $ & $\left( -\right) $
\\ 
\multicolumn{1}{l}{Eigenvector centrality} & $-$ & $-$ & $-$ & $-$ & $-$ & $-
$ & $-25.31$ \\ 
\multicolumn{1}{l}{} & $\left( -\right) $ & $\left( -\right) $ & $\left(
-\right) $ & $\left( -\right) $ & $\left( -\right) $ & $\left( -\right) $ & $\left( 16.59\right) $
\\ 
\multicolumn{1}{l}{Log(Patents)} & $0.28$ & $0.24$ & $0.33$ & $0.31$ & $0.31$
& $0.36$ & $0.39$ \\ 
\multicolumn{1}{l}{} & $\left( 0.04\right) $ & $\left( 0.05\right) $ & $%
\left( 0.04\right) $ & $\left( 0.04\right) $ & $\left( 0.04\right) $ & $\left( 0.04\right) $ & $\left( 0.03\right) $
\\ 
\multicolumn{1}{l}{Constant} & $6.08$ & $6.07$ & $7.35$ & $7.46$ & $7.49$ & $%
6.08$ & $5.95$ \\ 
\multicolumn{1}{l}{} & $\left( 0.22\right) $ & $\left( 0.22\right) $ & $%
\left( 0.50\right) $ & $\left( 0.52\right) $ & $\left( 0.51\right) $ & $\left( 0.30\right) $ & $\left( 0.28\right) $
\\ 
\multicolumn{1}{l}{$R^{2}$} & $0.346$ & $0.338$ & $0.334$ & $0.334$ & $0.337$
& $0.312$ & $0.310$ \\ 
\multicolumn{1}{l}{Adjusted $R^{2}$} & $0.340$ & $0.334$ & $0.330$ & $0.330$
& $0.332$ & $0.308$ & $0.306$ \\ 
\multicolumn{1}{l}{Number of observations} & $328$ & $328$ & $328$ & $328$ & 
$328$ & $328$ & $328$ \\ \hline \hline
\end{tabular}
\begin{minipage}{400pt}
{\fontsize{9pt}{9pt}\selectfont \smallskip  \textit{Note}: For each centrality measure, we take the minimum of all nodes that contain a given firm. Other summary statistics, such as mean and maximum, correlate less strongly with the performance metrics. Standard errors are in parentheses.}
\end{minipage}
\end{center}
\label{Table - centrality reg (revenue)}
\end{table}

\begin{table}[tbh]
\caption{EBIT Regression with Centrality Measures}
\begin{center}
\fontsize{9pt}{11pt}\selectfont%
\begin{tabular}{cccccccc}
\hline \hline
LHS variable: & \multicolumn{7}{c}{Log(EBIT)} \\ 
\cline{2-4}\cline{3-4}\cline{5-7}\cline{8-8}
& (1) & (2) & (3) & (4) & (5) & (6) & (7) \\ 
& Baseline &  &  &  &  &  &  \\ \hline
\multicolumn{1}{l}{Flare length} & $0.33$ & $-$ & $-$ & $-$ & $-$ & $-$ & $-$
\\ 
\multicolumn{1}{l}{} & $\left( 0.08\right) $ & $\left( -\right) $ & $\left(
-\right) $ & $\left( -\right) $ & $\left( -\right) $ & $\left( -\right) $ & $%
\left( -\right) $ \\ 
\multicolumn{1}{l}{Islands only} & $0.94$ & $-$ & $-$ & $-$ & $-$ & $-$ & $-$
\\ 
\multicolumn{1}{l}{} & $\left( 0.88\right) $ & $\left( -\right) $ & $\left(
-\right) $ & $\left( -\right) $ & $\left( -\right) $ & $\left( -\right) $ & $%
\left( -\right) $ \\ 
\multicolumn{1}{l}{Number of nodes} & $-$ & $0.04$ & $-$ & $-$ & $-$ & $-$ & 
$-$ \\ 
\multicolumn{1}{l}{} & $\left( -\right) $ & $\left( 0.02\right) $ & $\left(
-\right) $ & $\left( -\right) $ & $\left( -\right) $ & $\left( -\right) $ & $%
\left( -\right) $ \\ 
\multicolumn{1}{l}{Degree centrality} & $-$ & $-$ & $-349.7$ & $-$ & $-$ & $-
$ & $-$ \\ 
\multicolumn{1}{l}{} & $\left( -\right) $ & $\left( -\right) $ & $\left(
102.7\right) $ & $\left( -\right) $ & $\left( -\right) $ & $\left( -\right) $
& $\left( -\right) $ \\ 
\multicolumn{1}{l}{Closeness centrality} & $-$ & $-$ & $-$ & $-9.17$ & $-$ & 
$-$ & $-$ \\ 
\multicolumn{1}{l}{} & $\left( -\right) $ & $\left( -\right) $ & $\left(
-\right) $ & $\left( 4.54\right) $ & $\left( -\right) $ & $\left( -\right) $
& $\left( -\right) $ \\ 
\multicolumn{1}{l}{Harmonic centrality} & $-$ & $-$ & $-$ & $-$ & $-0.006$ & 
$-$ & $-$ \\ 
\multicolumn{1}{l}{} & $\left( -\right) $ & $\left( -\right) $ & $\left(
-\right) $ & $\left( -\right) $ & $\left( 0.003\right) $ & $\left( -\right) $
& $\left( -\right) $ \\ 
\multicolumn{1}{l}{Betweenness centrality} & $-$ & $-$ & $-$ & $-$ & $-$ & $%
-158.9$ & $-$ \\ 
\multicolumn{1}{l}{} & $\left( -\right) $ & $\left( -\right) $ & $\left(
-\right) $ & $\left( -\right) $ & $\left( -\right) $ & $\left( 59.1\right) $ & $\left( -\right) $
\\ 
\multicolumn{1}{l}{Eigenvector centrality} & $-$ & $-$ & $-$ & $-$ & $-$ & $-
$ & $-14.32$ \\ 
\multicolumn{1}{l}{} & $\left( -\right) $ & $\left( -\right) $ & $\left(
-\right) $ & $\left( -\right) $ & $\left( -\right) $ & $\left( -\right) $ & $\left( 17.53\right) $
\\ 
\multicolumn{1}{l}{Log(Patents)} & $0.29$ & $0.30$ & $0.34$ & $0.36$ & $0.36$
& $0.35$ & $0.41$ \\ 
\multicolumn{1}{l}{} & $\left( 0.05\right) $ & $\left( 0.06\right) $ & $%
\left( 0.04\right) $ & $\left( 0.04\right) $ & $\left( 0.04\right) $ & $\left( 0.04\right) $ & $\left( 0.04\right) $
\\ 
\multicolumn{1}{l}{Constant} & $3.97$ & $3.85$ & $5.20$ & $4.61$ & $4.65$ & $%
4.21$ & $3.73$ \\ 
\multicolumn{1}{l}{} & $\left( 0.24\right) $ & $\left( 0.24\right) $ & $%
\left( 0.53\right) $ & $\left( 0.57\right) $ & $\left( 0.55\right) $ & $\left( 0.32\right) $ & $\left( 0.30\right) $
\\ 
\multicolumn{1}{l}{$R^{2}$} & $0.345$ & $0.325$ & $0.334$ & $0.317$ & $0.319$
& $0.324$ & $0.310$ \\ 
\multicolumn{1}{l}{Adjusted $R^{2}$} & $0.338$ & $0.320$ & $0.329$ & $0.313$
& $0.314$ & $0.320$ & $0.305$ \\ 
\multicolumn{1}{l}{Number of observations} & $301$ & $301$ & $301$ & $301$ & 
$301$ & $301$ & $301$ \\ \hline \hline
\end{tabular}
\begin{minipage}{400pt}
{\fontsize{9pt}{9pt}\selectfont \smallskip  \textit{Note}: For each centrality measure, we take the minimum of all nodes that contain a given firm. Other summary statistics, such as mean and maximum, correlate less strongly with the performance metrics. Standard errors are in parentheses.}
\end{minipage}
\end{center}
\label{Table - centrality reg (ebit)}
\end{table}

\begin{table}[tbh]
\caption{Market-Value Regression with Centrality Measures}
\begin{center}
\fontsize{9pt}{11pt}\selectfont%
\begin{tabular}{cccccccc}
\hline \hline
LHS variable: & \multicolumn{7}{c}{Log(Market value)} \\ 
\cline{2-4}\cline{3-4}\cline{5-7}\cline{8-8}
& (1) & (2) & (3) & (4) & (5) & (6) & (7) \\ 
& Baseline &  &  &  &  &  &  \\ \hline
\multicolumn{1}{l}{Flare length} & $0.28$ & $-$ & $-$ & $-$ & $-$ & $-$ & $-$
\\ 
\multicolumn{1}{l}{} & $\left( 0.08\right) $ & $\left( -\right) $ & $\left(
-\right) $ & $\left( -\right) $ & $\left( -\right) $ & $\left( -\right) $ & $%
\left( -\right) $ \\ 
\multicolumn{1}{l}{Islands only} & $0.70$ & $-$ & $-$ & $-$ & $-$ & $-$ & $-$
\\ 
\multicolumn{1}{l}{} & $\left( 0.90\right) $ & $\left( -\right) $ & $\left(
-\right) $ & $\left( -\right) $ & $\left( -\right) $ & $\left( -\right) $ & $%
\left( -\right) $ \\ 
\multicolumn{1}{l}{Number of nodes} & $-$ & $0.04$ & $-$ & $-$ & $-$ & $-$ & 
$-$ \\ 
\multicolumn{1}{l}{} & $\left( -\right) $ & $\left( 0.02\right) $ & $\left(
-\right) $ & $\left( -\right) $ & $\left( -\right) $ & $\left( -\right) $ & $%
\left( -\right) $ \\ 
\multicolumn{1}{l}{Degree centrality} & $-$ & $-$ & $-326.6$ & $-$ & $-$ & $-
$ & $-$ \\ 
\multicolumn{1}{l}{} & $\left( -\right) $ & $\left( -\right) $ & $\left(
103.0\right) $ & $\left( -\right) $ & $\left( -\right) $ & $\left( -\right) $
& $\left( -\right) $ \\ 
\multicolumn{1}{l}{Closeness centrality} & $-$ & $-$ & $-$ & $-11.25$ & $-$
& $-$ & $-$ \\ 
\multicolumn{1}{l}{} & $\left( -\right) $ & $\left( -\right) $ & $\left(
-\right) $ & $\left( 4.50\right) $ & $\left( -\right) $ & $\left( -\right) $
& $\left( -\right) $ \\ 
\multicolumn{1}{l}{Harmonic centrality} & $-$ & $-$ & $-$ & $-$ & $-0.007$ & 
$-$ & $-$ \\ 
\multicolumn{1}{l}{} & $\left( -\right) $ & $\left( -\right) $ & $\left(
-\right) $ & $\left( -\right) $ & $\left( 0.003\right) $ & $\left( -\right) $
& $\left( -\right) $ \\ 
\multicolumn{1}{l}{Betweenness centrality} & $-$ & $-$ & $-$ & $-$ & $-$ & $-185.9$ & $-$ \\ 
\multicolumn{1}{l}{} & $\left( -\right) $ & $\left( -\right) $ & $\left(
-\right) $ & $\left( -\right) $ & $\left( -\right) $ & $\left( 58.7\right) $ & $\left( -\right) $
\\ 
\multicolumn{1}{l}{Eigenvector centrality} & $-$ & $-$ & $-$ & $-$ & $-$ & $-
$ & $-21.16$ \\ 
\multicolumn{1}{l}{} & $\left( -\right) $ & $\left( -\right) $ & $\left(
-\right) $ & $\left( -\right) $ & $\left( -\right) $ & $\left( -\right) $ & $\left( 17.67\right) $ 
\\ 
\multicolumn{1}{l}{Log(Patents)} & $0.34$ & $0.33$ & $0.37$ & $0.38$ & $0.38$
& $0.37$ & $0.43$ \\ 
\multicolumn{1}{l}{} & $\left( 0.05\right) $ & $\left( 0.06\right) $ & $%
\left( 0.04\right) $ & $\left( 0.04\right) $ & $\left( 0.04\right) $ & $\left( 0.04\right) $ & $\left( 0.04\right) $
\\ 
\multicolumn{1}{l}{Constant} & $6.20$ & $6.12$ & $7.40$ & $7.15$ & $7.18$ & $%
6.62$ & $6.11$ \\ 
\multicolumn{1}{l}{} & $\left( 0.24\right) $ & $\left( 0.24\right) $ & $%
\left( 0.53\right) $ & $\left( 0.56\right) $ & $\left( 0.55\right) $ & $\left( 0.32\right) $ & $\left( 0.30\right) $
\\ 
\multicolumn{1}{l}{$R^{2}$} & $0.342$ & $0.331$ & $0.340$ & $0.332$ & $0.334$
& $0.340$ & $0.323$ \\ 
\multicolumn{1}{l}{Adjusted $R^{2}$} & $0.336$ & $0.327$ & $0.336$ & $0.328$
& $0.330$ & $0.336$ & $0.318$ \\ 
\multicolumn{1}{l}{Number of observations} & $325$ & $325$ & $325$ & $325$ & 
$325$ & $325$ & $325$ \\ \hline \hline
\end{tabular}
\begin{minipage}{400pt}
{\fontsize{9pt}{9pt}\selectfont \smallskip  \textit{Note}: For each centrality measure, we take the minimum of all nodes that contain a given firm. Other summary statistics, such as mean and maximum, correlate less strongly with the performance metrics. Standard errors are in parentheses.}
\end{minipage}
\end{center}
\label{Table - centrality reg (mcap)}
\end{table}

\clearpage
\subsection*{G.2 \ Jaffe's Technological Distance}

 For any pair of firm-year observations $(i,t)$ and $(i',t')$, Jaffe's distance metric is defined as the cosine dissimilarity between the vectors $l_{i,t}$ and $l_{i',t'}$, each element of which is $l_{i,t,c}=\frac{\tilde{p}_{i,t,c}}{\sum_{c} \tilde{p}_{i,t,c}}$ (see Appendix D).
 
 We propose four ways to construct an ``aggregate'' Jaffe measure for each firm throughout the sample period, so that it can be used in regressions just like our flare measures.  The first two versions are based exclusively on 2005, the final year of our sample period, because our baseline regressions use the firms' performances in 2005 as the LHS variable. One is the Jaffe distance between a focal firm and its nearest neighbor (NN); the other is between a focal firm and the population mean (PM) of all firms. The former would reflect local competition, and the latter global competition. For the last two versions, we broaden our scope to all firm-years. We calculate each firm's mean location in 1976--2005, and then define the third and the fourth Jaffe distances based on the NN and the PM reference locations, respectively.

\begin{table}[tbh]
\caption{Regressions with Jaffe's Distance Metrics}
\begin{center}
\fontsize{9pt}{11pt}\selectfont%
\begin{tabular}{cccccccccccc}
\hline \hline
LHS variable: & \multicolumn{3}{c}{Log(Revenue)} &  & \multicolumn{3}{c}{
Log(EBIT)} &  & \multicolumn{3}{c}{Log(Market value)} \\ 
\cline{2-4}\cline{3-4}\cline{6-8}\cline{10-12}
& (1) & (2) & (3) &  & (4) & (5) & (6) &  & (7) & (8) & (9) \\ \hline
\multicolumn{1}{l}{Jaffe distance I (2005 NN)} & $-0.00$ & $0.68$ & $0.72$ & 
& $-1.03$ & $-0.37$ & $-0.33$ &  & $-1.90$ & $-1.23$ & $-1.20$ \\ 
\multicolumn{1}{l}{} & $\left( 0.51\right) $ & $\left( 0.41\right) $ & $%
\left( 0.40\right) $ &  & $\left( 0.55\right) $ & $\left( 0.46\right) $ & $%
\left( 0.45\right) $ &  & $\left( 0.52\right) $ & $\left( 0.43\right) $ & $%
\left( 0.43\right) $ \\ 
\multicolumn{1}{l}{Flare length} & $-$ & $-$ & $0.28$ &  & $-$ & $-$ & $0.27$
&  & $-$ & $-$ & $0.22$ \\ 
\multicolumn{1}{l}{} & $\left( -\right) $ & $\left( -\right) $ & $\left(
0.08\right) $ &  & $\left( -\right) $ & $\left( -\right) $ & $\left(
0.09\right) $ &  & $\left( -\right) $ & $\left( -\right) $ & $\left(
0.08\right) $ \\ 
\multicolumn{1}{l}{Islands only} & $-$ & $-$ & $0.65$ &  & $-$ & $-$ & $0.71$
&  & $-$ & $-$ & $0.53$ \\ 
\multicolumn{1}{l}{} & $\left( -\right) $ & $\left( -\right) $ & $\left(
0.83\right) $ &  & $\left( -\right) $ & $\left( -\right) $ & $\left(
0.88\right) $ &  & $\left( -\right) $ & $\left( -\right) $ & $\left(
0.88\right) $ \\ 
\multicolumn{1}{l}{Log patents} & $-$ & $0.48$ & $0.36$ &  & $-$ & $0.47$ & $%
0.35$ &  & $-$ & $0.47$ & $0.37$ \\ 
\multicolumn{1}{l}{} & $\left( -\right) $ & $\left( 0.28\right) $ & $\left(
0.05\right) $ &  & $\left( -\right) $ & $\left( 0.04\right) $ & $\left(
0.06\right) $ &  & $\left( -\right) $ & $\left( 0.04\right) $ & $\left(
0.06\right) $ \\ 
\multicolumn{1}{l}{Adjusted $R^{2}$} & $-0.003$ & $0.348$ & $0.371$ &  & $%
0.009$ & $0.327$ & $0.347$ &  & $0.041$ & $0.343$ & $0.353$ \\ 
\multicolumn{1}{l}{Number of observations} & $293$ & $293$ & $293$ &  & $268$
& $268$ & $268$ &  & $290$ & $290$ & $290$ \\ \hline \hline
LHS variable: & \multicolumn{3}{c}{Log(Revenue)} &  & \multicolumn{3}{c}{
Log(EBIT)} &  & \multicolumn{3}{c}{Log(Market value)} \\ 
\cline{2-4}\cline{3-4}\cline{6-8}\cline{10-12}
& (1) & (2) & (3) &  & (4) & (5) & (6) &  & (7) & (8) & (9) \\ \hline
\multicolumn{1}{l}{Jaffe distance II (2005 PM)} & $-4.16$ & $-0.56$ & $-0.30$
&  & $-5.19$ & $-1.87$ & $-1.62$ &  & $-6.16$ & $-3.04$ & $-2.86$ \\ 
\multicolumn{1}{l}{} & $\left( 0.73\right) $ & $\left( 0.71\right) $ & $%
\left( 0.70\right) $ &  & $\left( 0.76\right) $ & $\left( 0.77\right) $ & $%
\left( 0.76\right) $ &  & $\left( 0.72\right) $ & $\left( 0.74\right) $ & $%
\left( 0.74\right) $ \\ 
\multicolumn{1}{l}{Flare length} & $-$ & $-$ & $0.27$ &  & $-$ & $-$ & $0.25$
&  & $-$ & $-$ & $0.18$ \\ 
\multicolumn{1}{l}{} & $\left( -\right) $ & $\left( -\right) $ & $\left(
0.08\right) $ &  & $\left( -\right) $ & $\left( -\right) $ & $\left(
0.09\right) $ &  & $\left( -\right) $ & $\left( -\right) $ & $\left(
0.08\right) $ \\ 
\multicolumn{1}{l}{Islands only} & $-$ & $-$ & $0.64$ &  & $-$ & $-$ & $0.52$
&  & $-$ & $-$ & $0.17$ \\ 
\multicolumn{1}{l}{} & $\left( -\right) $ & $\left( -\right) $ & $\left(
0.83\right) $ &  & $\left( -\right) $ & $\left( -\right) $ & $\left(
0.88\right) $ &  & $\left( -\right) $ & $\left( -\right) $ & $\left(
0.87\right) $ \\ 
\multicolumn{1}{l}{Log patents} & $-$ & $0.46$ & $0.35$ &  & $-$ & $0.42$ & $%
0.32$ &  & $-$ & $0.39$ & $0.32$ \\ 
\multicolumn{1}{l}{} & $\left( -\right) $ & $\left( 0.04\right) $ & $\left(
0.05\right) $ &  & $\left( -\right) $ & $\left( 0.05\right) $ & $\left(
0.06\right) $ &  & $\left( -\right) $ & $\left( 0.05\right) $ & $\left(
0.06\right) $ \\ 
\multicolumn{1}{l}{Adjusted $R^{2}$} & $0.098$ & $0.343$ & $0.365$ &  & $%
0.145$ & $0.340$ & $0.357$ &  & $0.198$ & $0.362$ & $0.368$ \\ 
\multicolumn{1}{l}{Number of observations} & $293$ & $293$ & $293$ &  & $268$
& $268$ & $268$ &  & $290$ & $290$ & $290$ \\ \hline \hline
LHS variable: & \multicolumn{3}{c}{Log(Revenue)} &  & \multicolumn{3}{c}{
Log(EBIT)} &  & \multicolumn{3}{c}{Log(Market value)} \\ 
\cline{2-4}\cline{3-4}\cline{6-8}\cline{10-12}
& (1) & (2) & (3) &  & (4) & (5) & (6) &  & (7) & (8) & (9) \\ \hline
\multicolumn{1}{l}{Jaffe distance III (mean NN)} & $-0.60$ & $0.68$ & $0.68$
&  & $-1.51$ & $-0.20$ & $-0.19$ &  & $-2.03$ & $-0.72$ & $-0.72$ \\ 
\multicolumn{1}{l}{} & $\left( 0.49\right) $ & $\left( 0.42\right) $ & $%
\left( 0.04\right) $ &  & $\left( 0.53\right) $ & $\left( 0.46\right) $ & $%
\left( 0.45\right) $ &  & $\left( 0.52\right) $ & $\left( 0.45\right) $ & $%
\left( 0.45\right) $ \\ 
\multicolumn{1}{l}{Flare length} & $-$ & $-$ & $0.34$ &  & $-$ & $-$ & $0.33$
&  & $-$ & $-$ & $0.28$ \\ 
\multicolumn{1}{l}{} & $\left( -\right) $ & $\left( -\right) $ & $\left(
0.08\right) $ &  & $\left( -\right) $ & $\left( -\right) $ & $\left(
0.08\right) $ &  & $\left( -\right) $ & $\left( -\right) $ & $\left(
0.08\right) $ \\ 
\multicolumn{1}{l}{Islands only} & $-$ & $-$ & $0.99$ &  & $-$ & $-$ & $0.93$
&  & $-$ & $-$ & $0.66$ \\ 
\multicolumn{1}{l}{} & $\left( -\right) $ & $\left( -\right) $ & $\left(
0.84\right) $ &  & $\left( -\right) $ & $\left( -\right) $ & $\left(
0.88\right) $ &  & $\left( -\right) $ & $\left( -\right) $ & $\left(
0.90\right) $ \\ 
\multicolumn{1}{l}{Log patents} & $-$ & $0.42$ & $0.29$ &  & $-$ & $0.41$ & $%
0.29$ &  & $-$ & $0.43$ & $0.33$ \\ 
\multicolumn{1}{l}{} & $\left( -\right) $ & $\left( 0.03\right) $ & $\left(
0.04\right) $ &  & $\left( -\right) $ & $\left( 0.04\right) $ & $\left(
0.05\right) $ &  & $\left( -\right) $ & $\left( 0.04\right) $ & $\left(
0.05\right) $ \\ 
\multicolumn{1}{l}{Adjusted $R^{2}$} & $0.001$ & $0.306$ & $0.343$ &  & $%
0.024$ & $0.304$ & $0.336$ &  & $0.042$ & $0.321$ & $0.339$ \\ 
\multicolumn{1}{l}{Number of observations} & $328$ & $328$ & $328$ &  & $301$
& $301$ & $301$ &  & $325$ & $325$ & $325$ \\ \hline \hline
LHS variable: & \multicolumn{3}{c}{Log(Revenue)} &  & \multicolumn{3}{c}{
Log(EBIT)} &  & \multicolumn{3}{c}{Log(Market value)} \\ 
\cline{2-4}\cline{3-4}\cline{6-8}\cline{10-12}
& (1) & (2) & (3) &  & (4) & (5) & (6) &  & (7) & (8) & (9) \\ \hline
\multicolumn{1}{l}{Jaffe distance IV (mean PM)} & $-5.49$ & $-0.71$ & $-0.26$
&  & $-6.10$ & $-1.28$ & $-0.88$ &  & $-6.95$ & $-2.22$ & $-1.89$ \\ 
\multicolumn{1}{l}{} & $\left( 0.78\right) $ & $\left( 0.87\right) $ & $%
\left( 0.86\right) $ &  & $\left( 0.82\right) $ & $\left( 0.95\right) $ & $%
\left( 0.93\right) $ &  & $\left( 0.81\right) $ & $\left( 0.93\right) $ & $%
\left( 0.93\right) $ \\ 
\multicolumn{1}{l}{Flare length} & $-$ & $-$ & $0.34$ &  & $-$ & $-$ & $0.32$
&  & $-$ & $-$ & $0.26$ \\ 
\multicolumn{1}{l}{} & $\left( -\right) $ & $\left( -\right) $ & $\left(
0.08\right) $ &  & $\left( -\right) $ & $\left( -\right) $ & $\left(
0.08\right) $ &  & $\left( -\right) $ & $\left( -\right) $ & $\left(
0.08\right) $ \\ 
\multicolumn{1}{l}{Islands only} & $-$ & $-$ & $0.93$ &  & $-$ & $-$ & $0.85$
&  & $-$ & $-$ & $0.49$ \\ 
\multicolumn{1}{l}{} & $\left( -\right) $ & $\left( -\right) $ & $\left(
0.84\right) $ &  & $\left( -\right) $ & $\left( -\right) $ & $\left(
0.89\right) $ &  & $\left( -\right) $ & $\left( -\right) $ & $\left(
0.91\right) $ \\ 
\multicolumn{1}{l}{Log patents} & $-$ & $0.38$ & $0.27$ &  & $-$ & $0.38$ & $%
0.27$ &  & $-$ & $0.38$ & $0.29$ \\ 
\multicolumn{1}{l}{} & $\left( -\right) $ & $\left( 0.04\right) $ & $\left(
0.05\right) $ &  & $\left( -\right) $ & $\left( 0.05\right) $ & $\left(
0.05\right) $ &  & $\left( -\right) $ & $\left( 0.04\right) $ & $\left(
0.05\right) $ \\ 
\multicolumn{1}{l}{Adjusted $R^{2}$} & $0.130$ & $0.302$ & $0.338$ &  & $%
0.153$ & $0.308$ & $0.338$ &  & $0.182$ & $0.327$ & $0.343$ \\ 
\multicolumn{1}{l}{Number of observations} & $328$ & $328$ & $328$ &  & $301$
& $301$ & $301$ &  & $325$ & $325$ & $325$ \\ \hline \hline
\end{tabular}
\begin{minipage}{475pt}
{\fontsize{9pt}{9pt}\selectfont \smallskip  \textit{Note}: Each of the four panels from top to bottom uses a different version of the Jaffe measure (see text for their definitions). Standard errors are in parentheses. The constant term is included in all regressions but suppressed in the table to save space.}
\end{minipage}
\end{center}
\label{Table - Jaffe reg}
\end{table}

 Table \ref{Table - Jaffe reg} reports results based on these four Jaffe distances. The top panel uses the first version (2005 NN) as the main regressor, the second panel uses the second version (2005 PM), and so on. The fit is low (i.e., the adjusted $R^2$ is frequently close to zero and never above 0.2) when Jaffe's measure is the only regressor (columns 1, 4, and 7). Recall the analogous regressions in Table \ref{Table - Regressions} in section 5.2, in which flares and islands achieve a much better fit (i.e., the adjusted $R^2$ is always above 0.2).
 
 The Jaffe measure seems to contribute more to the fit when we control for patent count (columns 2, 5, and 8) and flares/islands (3, 6, and 9), but its coefficient estimate is mostly statistically insignificant. Its sign is usually negative but sometimes positive (columns 2 and 3 in the first and the third panels). Thus, the relationships between the firms' performances and their Jaffe distances lack cohesion and are difficult to interpret.

\clearpage

\section*{Appendix H \ Additional Exhibits}

This section reports additional pictures: (i) three-dimensional PCA, (ii) the Mapper graph in a PCA-based layout, (iii) coloring of the Mapper graph by year, patent count, and sector, and (iv) a time series of year-by-year PCA plots and Mapper graphs.

\paragraph{Three-dimensional PCA.}

Whereas Figure \ref{Figure - mapper(n20_m0_cos)} (a) is a two-dimensional PCA plot, Figure \ref{Figure - mapper(cos_log_pca3d)} is a three-dimensional PCA plot. Their comparison suggests the reduction of even one dimension could entail some important information loss.

\begin{figure}[htb!!!!]\centering%
\caption{Three-Dimensional PCA}%
\includegraphics[width=0.60\textwidth]{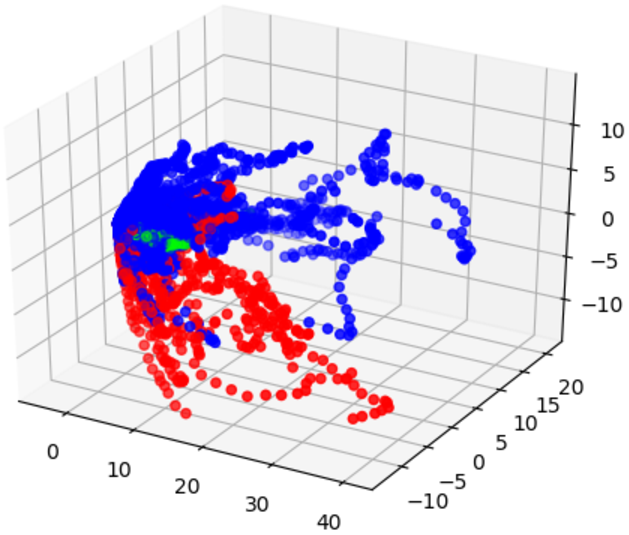}
\caption*{\footnotesize {%
\textit{Note}: Red markers are IT firms, green markers are drug makers, and blue markers are all others.}}%
\label{Figure - mapper(cos_log_pca3d)}
\end{figure}%

\clearpage

\paragraph{Mapper Graph with Nodes Fixed at PCA Locations.}

Figure \ref{Figure - mapper at PCA} shows another version of the Mapper graph in Figure \ref{Figure - mapper(n20_m0_cos)} (b) in which the positions of the nodes are fixed at those in Figure \ref{Figure - mapper(n20_m0_cos)} (a).\footnote{More precisely, the position of each node is fixed at \textit{the average of the PCA positions of the firm-year observations} that are contained in that node.} The advantage of this layout is that the correspondence between (a) and (b) becomes clearer, as it overlays the Mapper graph's edges on the PCA plot to show where continuity is being detected. Its disadvantage is that the central part of the graph is too crowded for further investigation.

\begin{figure}[htb!!!!]\centering%
\caption{Mapper Graph in Two-Dimensional PCA Layout}%
\includegraphics[width=0.9\textwidth]{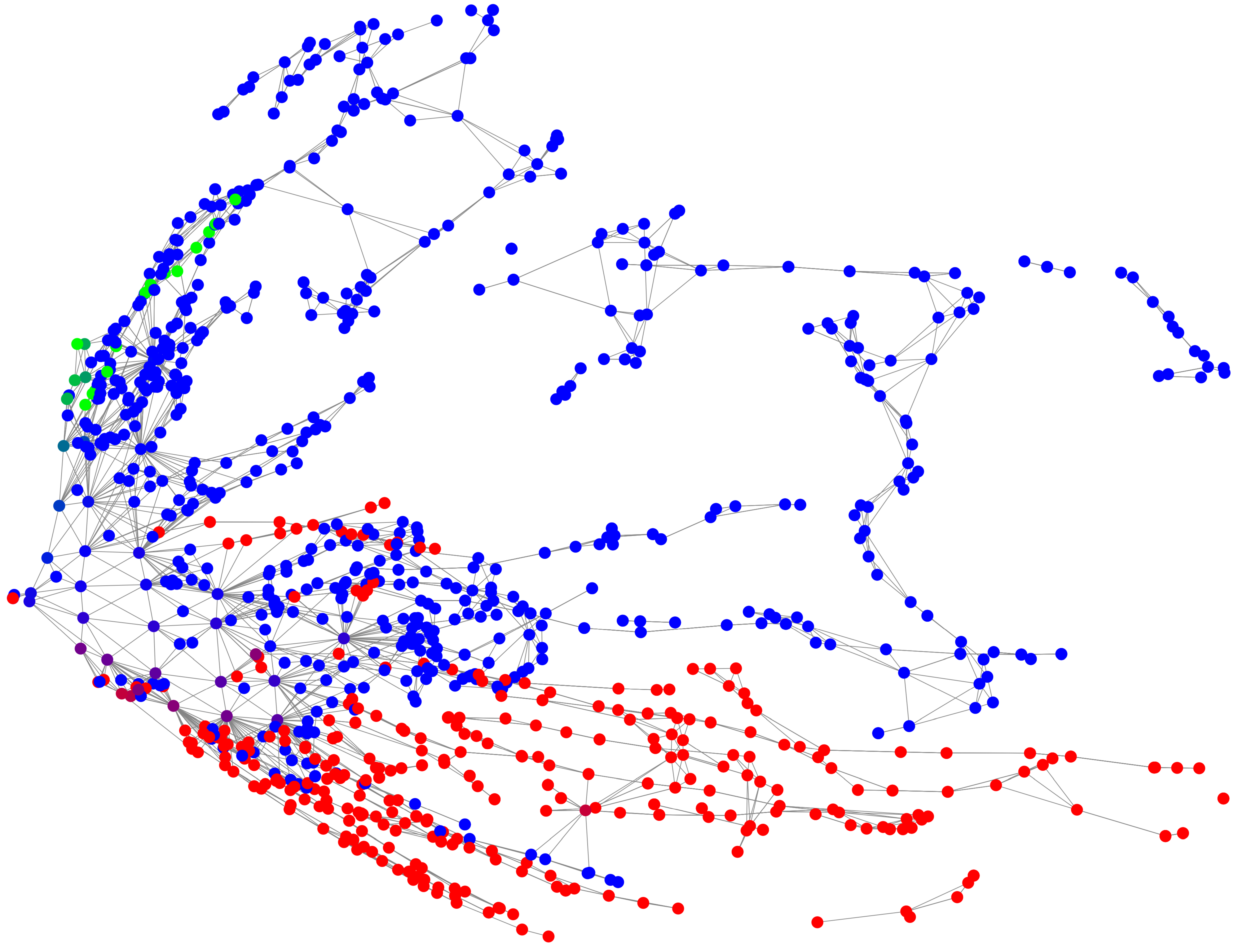}
\caption*{\footnotesize {%
\textit{Note}: Red markers are IT firms, green markers are drug makers, and blue markers are all others.}}%
\label{Figure - mapper at PCA}
\end{figure}%

\clearpage

\paragraph{Coloring by Year, Patent Count, and Sector.}

Figure \ref{Figure - color maps} shows six versions of the main Mapper graph under alternative color schemes to highlight different aspects of the data. The underlying data and graph are the same as in Figures \ref{Figure - mapper(n20_m0_cos)} (b), \ref{Figure - mapper(cos_log_details_1)}, and \ref{Figure - mapper(cos_log_details_2)}.

Panel (a) colors nodes according to the average calendar years of their component observations (firm-years), with blue and violet colors representing earlier years and red colors representing later years. Whereas our case studies in section 4.3 use arrows to represent firms' moves over time, this picture conveys similar information with a spectrum of colors. Flares with red nodes at the ends suggest centrifugal moves; those with blue ends suggest centripetal moves. Many central nodes display green and yellow colors because they contain observations in all years, the average of which lies between the two extreme colors.

Panel (b) shows nodes with many patents in red colors. The scale of inventive activities is positively correlated with their uniqueness. Hence, nodes at the end of the flares and the islands of large conglomerates tend to appear in red.

Panels (c)--(f) paint nodes with high concentration of firms in each of the four selected S\&P sectors (basic materials, capital goods, health care, and technology, respectively) in dark red, followed by bright red, orange, yellow, and so on. By contrast, dark violet/blue nodes do not contain firms in the focal sector. Their concentration patterns broadly agree with the industry annotations in Figure \ref{Figure - mapper(n20_m0_cos)} (b) and the firm-level case studies in section 4.3.

\begin{figure}[htb!!!!]
\caption{Mapper Graphs under Alternative Color Schemes}%

\begin{subfigure}{0.45\textwidth}
\centering
\includegraphics[width=0.85\linewidth]{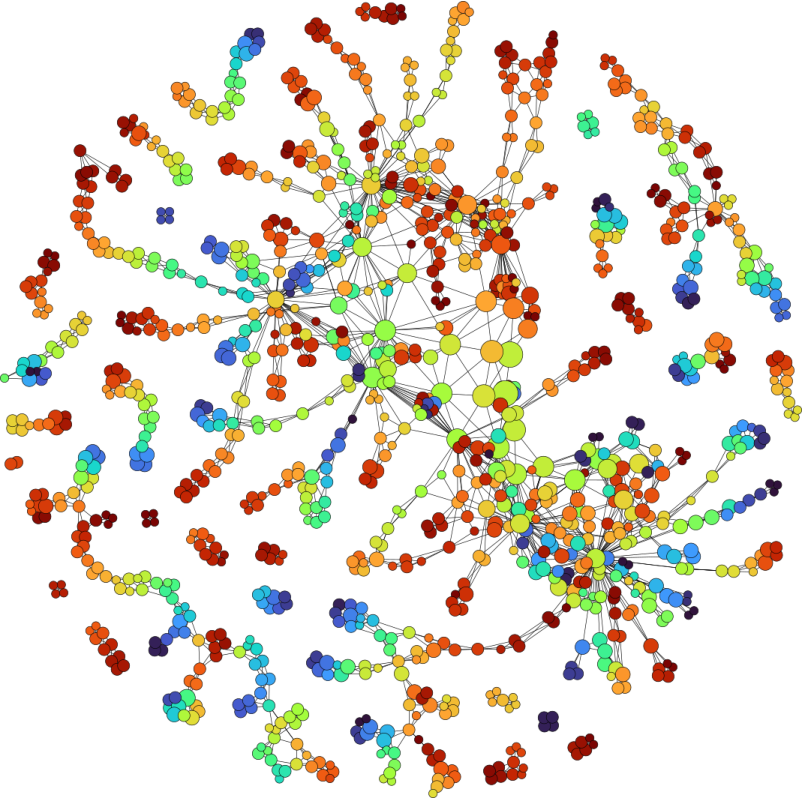}
\caption{By Year (red = later years)}%
\end{subfigure}
\begin{subfigure}{0.45\textwidth}
\centering
\includegraphics[width=0.85\linewidth]{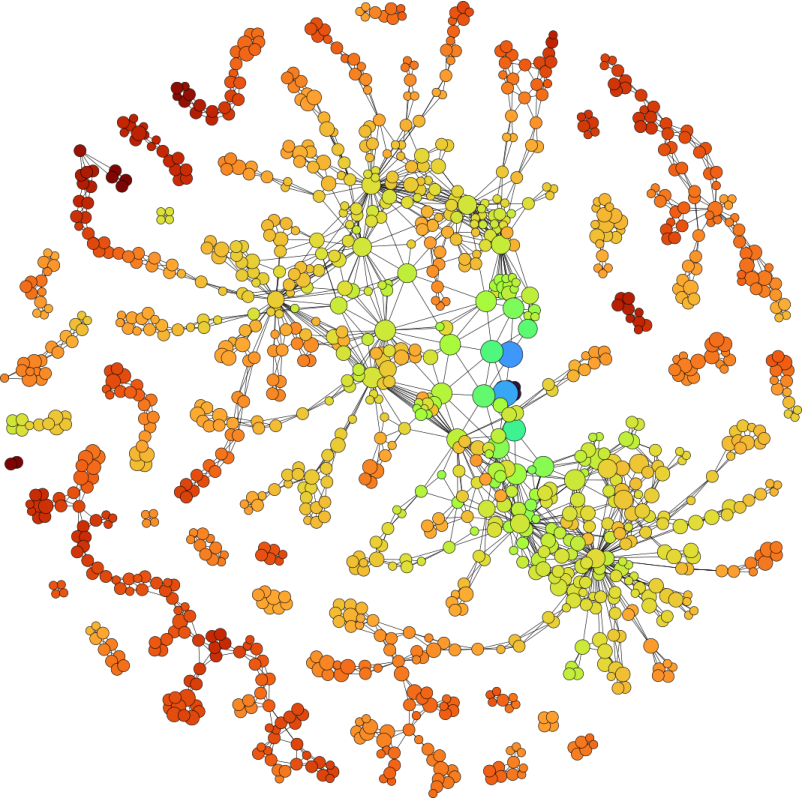}
\caption{By Patent Count (red = more patents)}%
\end{subfigure}

\begin{subfigure}{0.45\textwidth}
\centering
\includegraphics[width=0.85\linewidth]{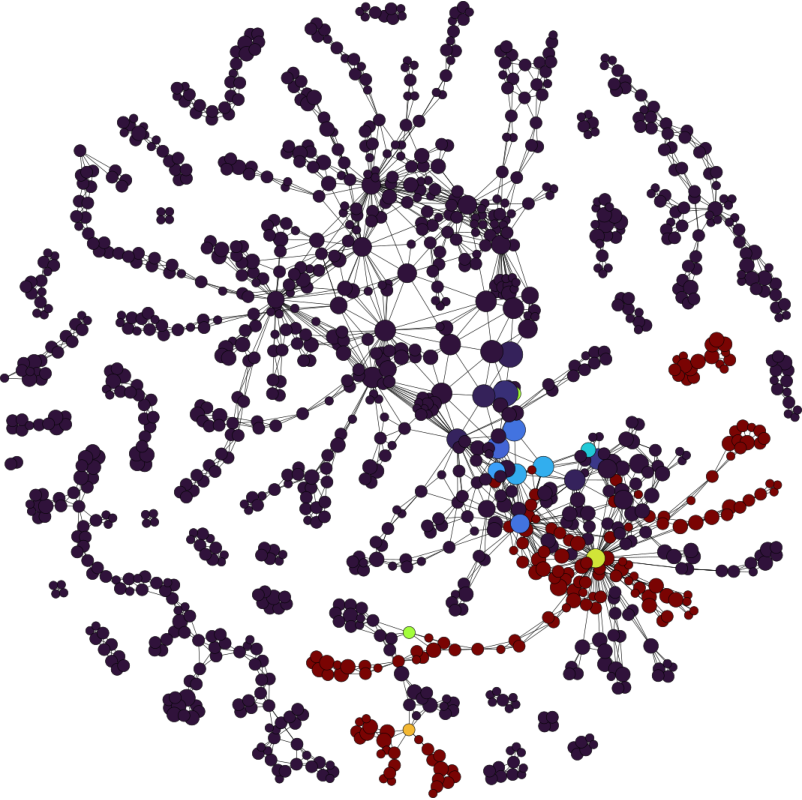}
\caption{Sector = Basic Materials}%
\end{subfigure}
\begin{subfigure}{0.45\textwidth}
\centering
\includegraphics[width=0.85\linewidth]{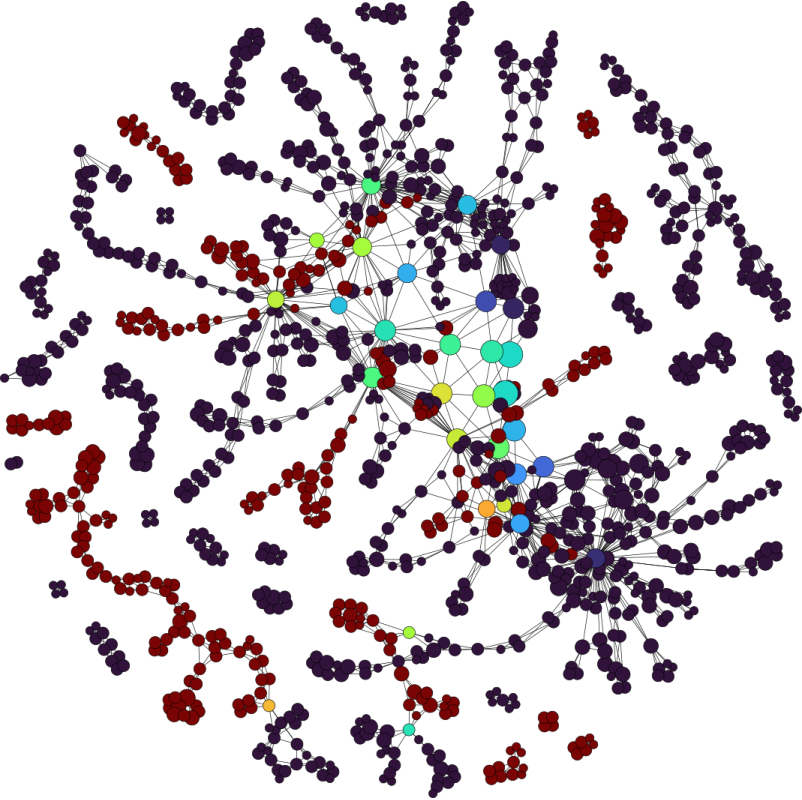}
\caption{Sector = Capital Goods}%
\end{subfigure}

\begin{subfigure}{0.45\textwidth}
\centering
\includegraphics[width=0.85\linewidth]{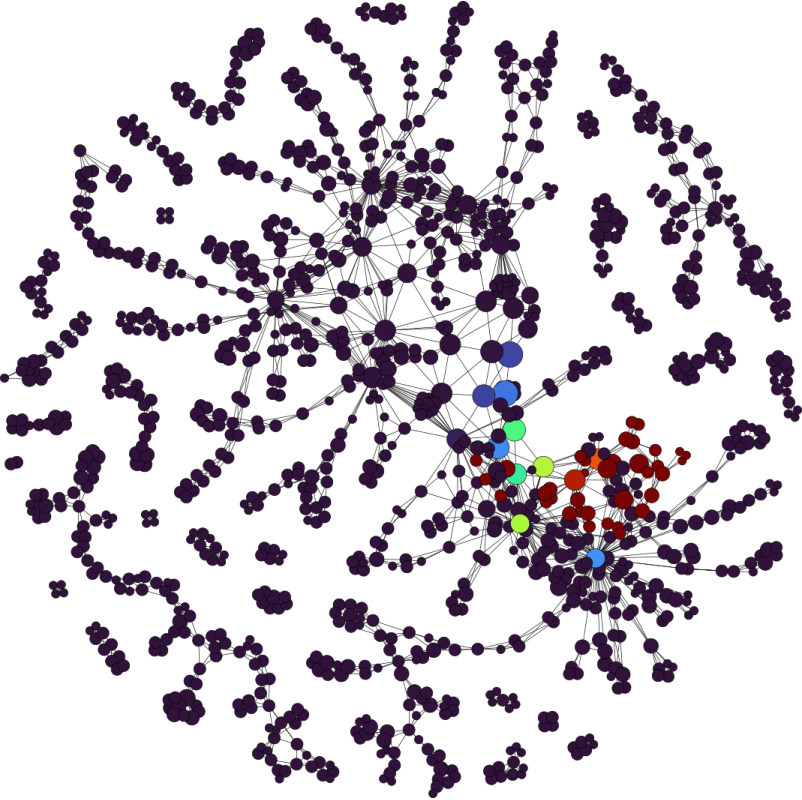}
\caption{Sector = Health Care}%
\end{subfigure}
\begin{subfigure}{0.45\textwidth}
\centering
\includegraphics[width=0.85\linewidth]{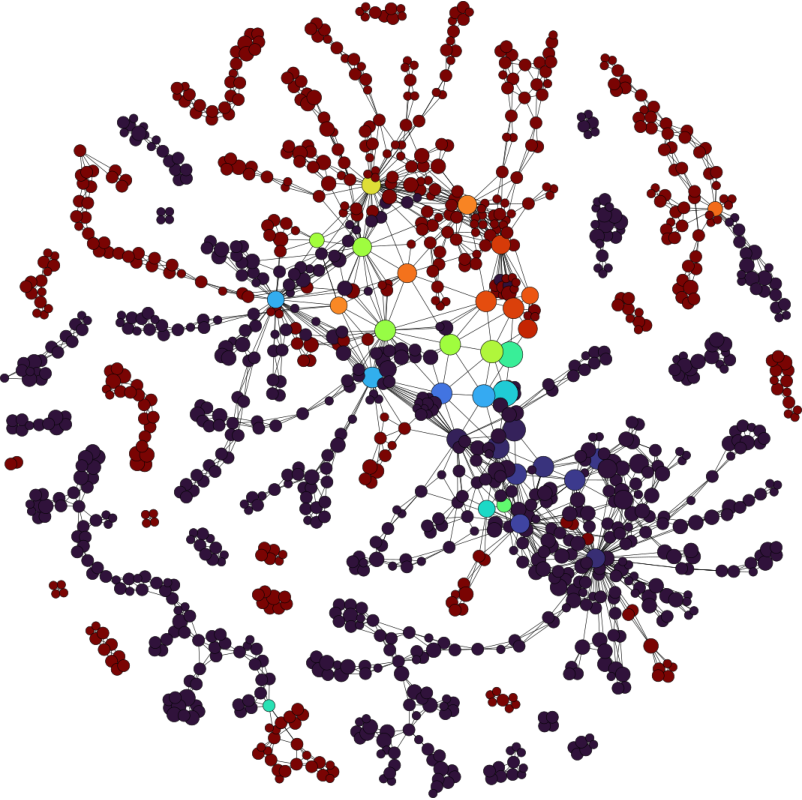}
\caption{Sector = Technology}%
\end{subfigure}
\caption*{\footnotesize {%
\textit{Note}: Each plot shows the main graph in Figure \ref{Figure - mapper(n20_m0_cos)} (b) under a different color scheme.}}%
\label{Figure - color maps}
\end{figure}%

\clearpage

\paragraph{Year-by-Year Mapper Graphs.}

Our baseline analysis pools the entire panel data and studies them in a single Mapper graph; one might wonder how the results will look year by year. Figures \ref{Figure - year-by-year 1} and \ref{Figure - year-by-year 2} show PCA plots and Mapper graphs for selected years (1980, 1985, ..., 2005). The PCA plots of different years look broadly similar, with the majority of firms near $(0,0)$ in the south-western part. Other firms sparsely populate the rest of the PCA spaces. Note any comparison across years cannot be precise because the coordinates (i.e., the first two principal components) are independently defined for each cross section.

The Mapper graphs for early years are fragmented, whereas those for later years feature most firms in a single, large connected component. This difference reflects the increasing number of firms over years, which fill the gaps between initially distant data points. The relative positions of industries resemble those in Figure \ref{Figure - mapper(n20_m0_cos)} (b): the IT sector and the pharmaceutical industry appear in the opposite ends of the main continent, which also contains other industries including aerospace, engineering, medical devices, and materials.

Despite these similarities, the exact locations of firms and industries cannot be compared across different graphs. Each graph represents the distribution of data points in a single cross section, based on the corresponding PCA plot, its cover, and so on. Just like we cannot precisely compare the PCA plots of different subsamples, we cannot compare the locations of the same firm across multiple Mapper graphs. Thus, even though analyzing data year by year is interesting in its own right, splitting the panel data into repeated cross sections entails a significant information loss: we can no longer track the firms' moves over time or characterize their long-run trajectories.

Even if one is willing to forgo the study of dynamics, using all data at once is still preferable because Mapper's usefulness is predicated on the notion of continuity in data. A Mapper graph of sparsely distributed data points, such as Figure \ref{Figure - year-by-year 1} (b), does not reveal much structure besides fragmentation. Pooling all years reduces such gaps and maximizes Mapper's capability to contextualize each observation within global, continuous patterns.

\begin{figure}[htb!!!!]
\caption{Year-by-Year Plots by PCA and Mapper (1 of 2)}%

\begin{subfigure}{0.45\textwidth}
\centering
\includegraphics[width=0.85\linewidth]{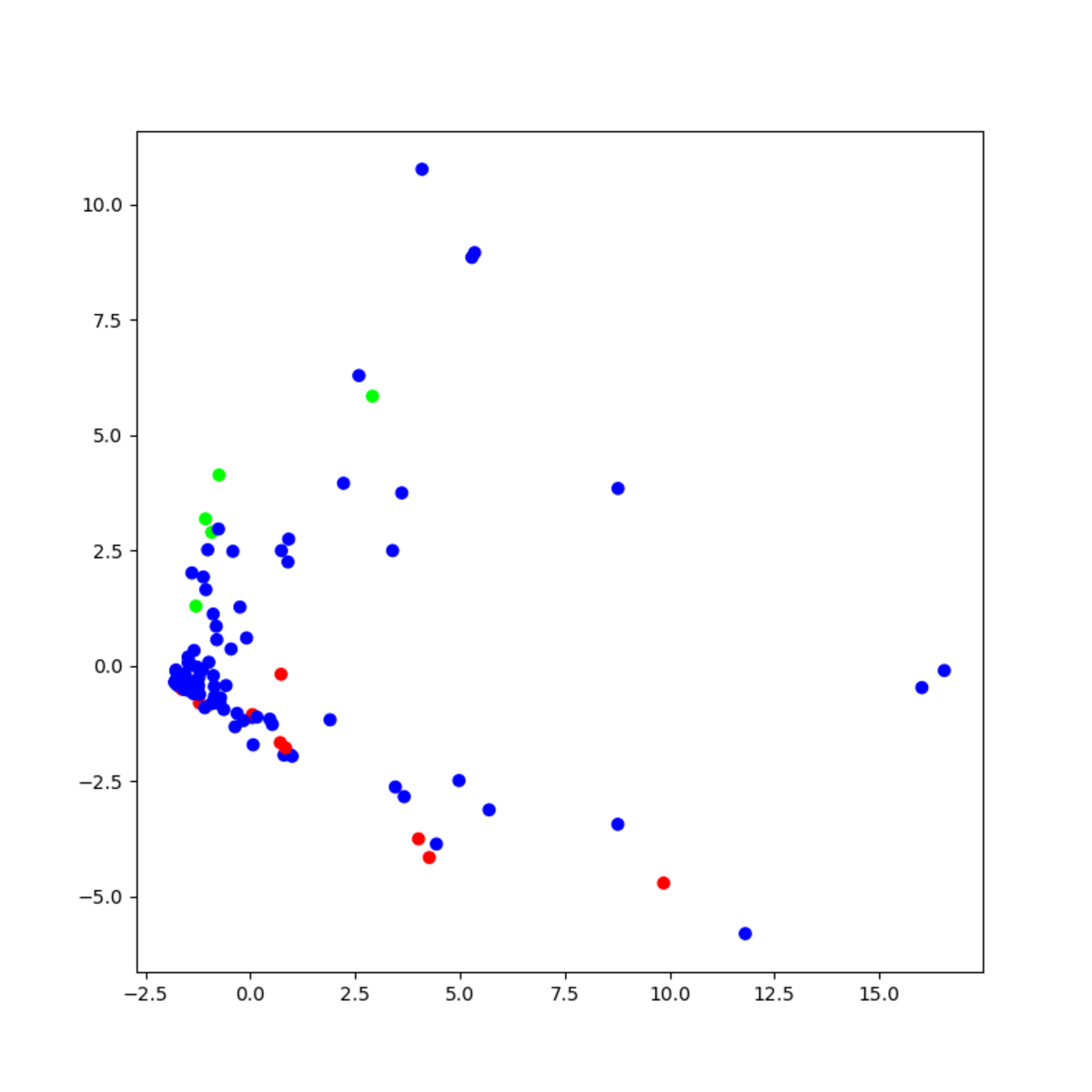}
\caption{1980 by PCA}%
\end{subfigure}
\begin{subfigure}{0.45\textwidth}
\centering
\includegraphics[width=0.85\linewidth]{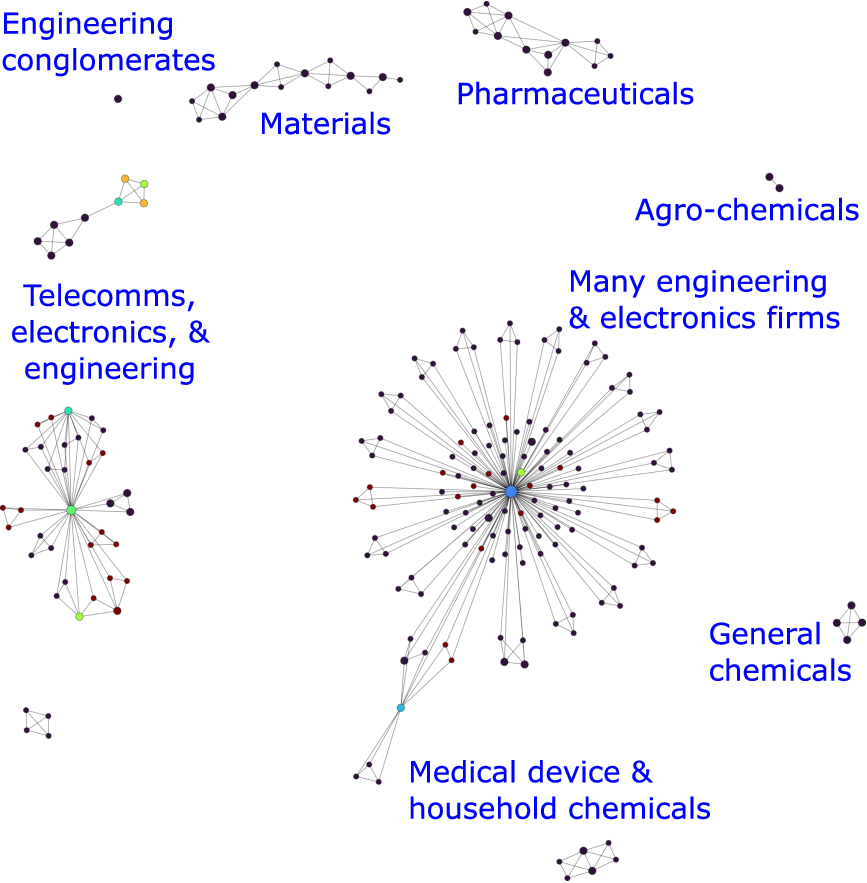}
\caption{1980 by Mapper}%
\end{subfigure}

\begin{subfigure}{0.45\textwidth}
\centering
\includegraphics[width=0.85\linewidth]{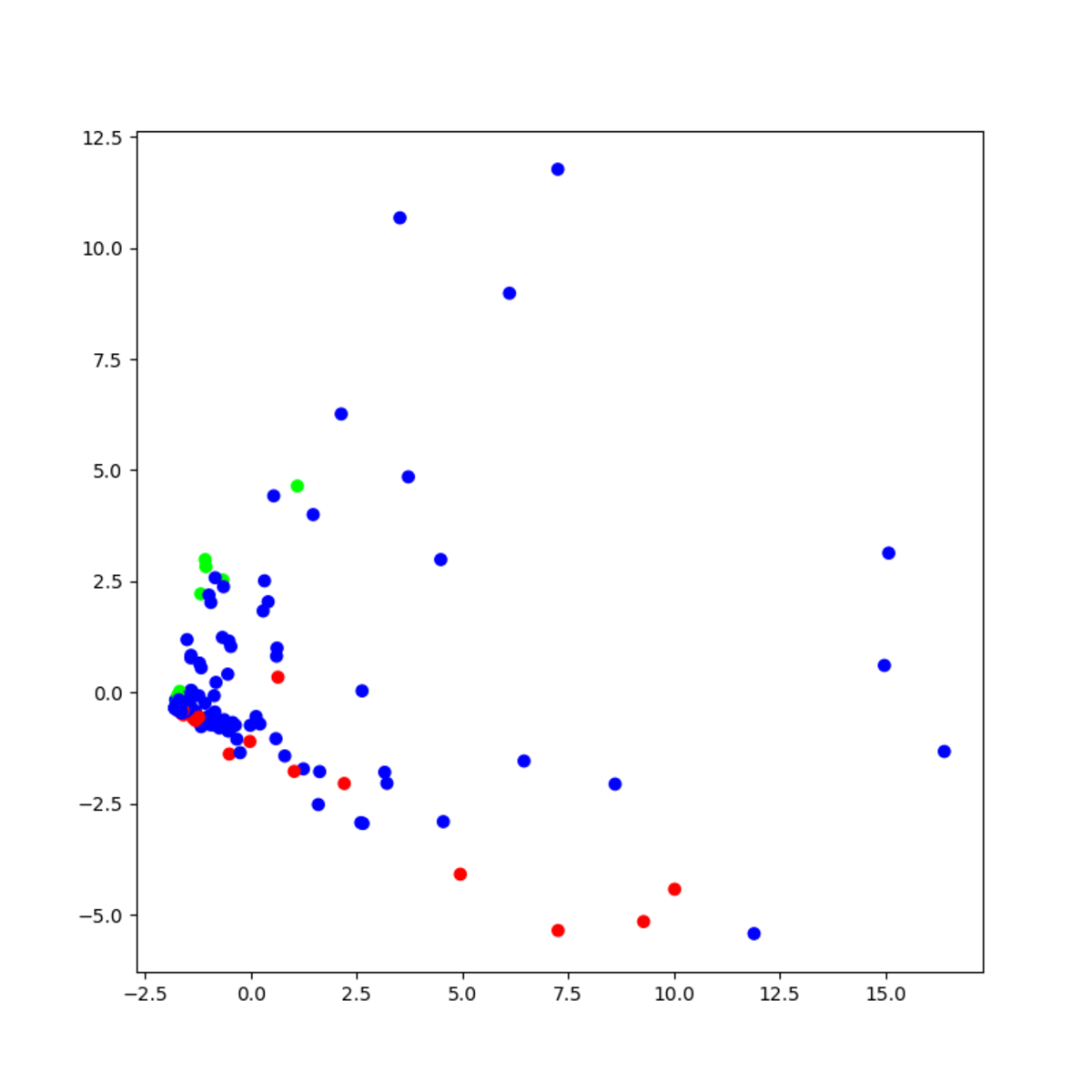}
\caption{1985 by PCA}%
\end{subfigure}
\begin{subfigure}{0.45\textwidth}
\centering
\includegraphics[width=0.90\linewidth]{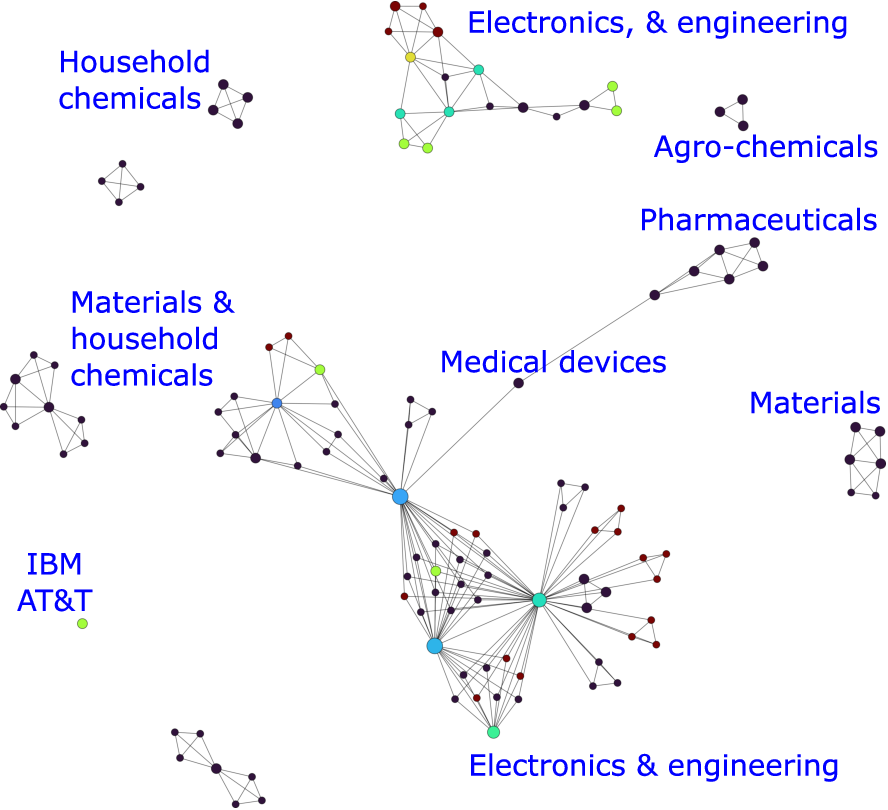}
\caption{1985 by Mapper}%
\end{subfigure}

\begin{subfigure}{0.45\textwidth}
\centering
\includegraphics[width=0.85\linewidth]{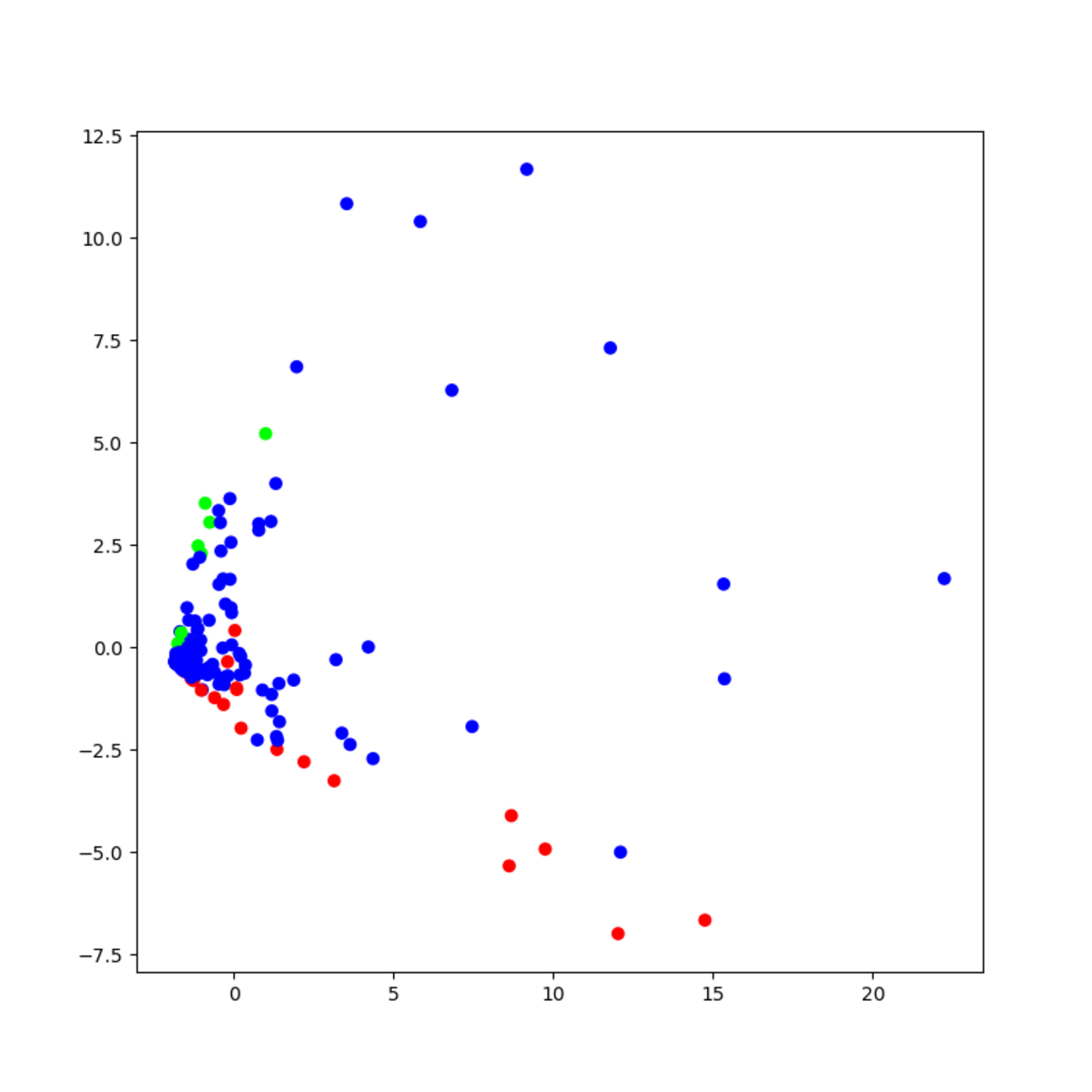}
\caption{1990 by PCA}%
\end{subfigure}
\begin{subfigure}{0.45\textwidth}
\centering
\includegraphics[width=0.90\linewidth]{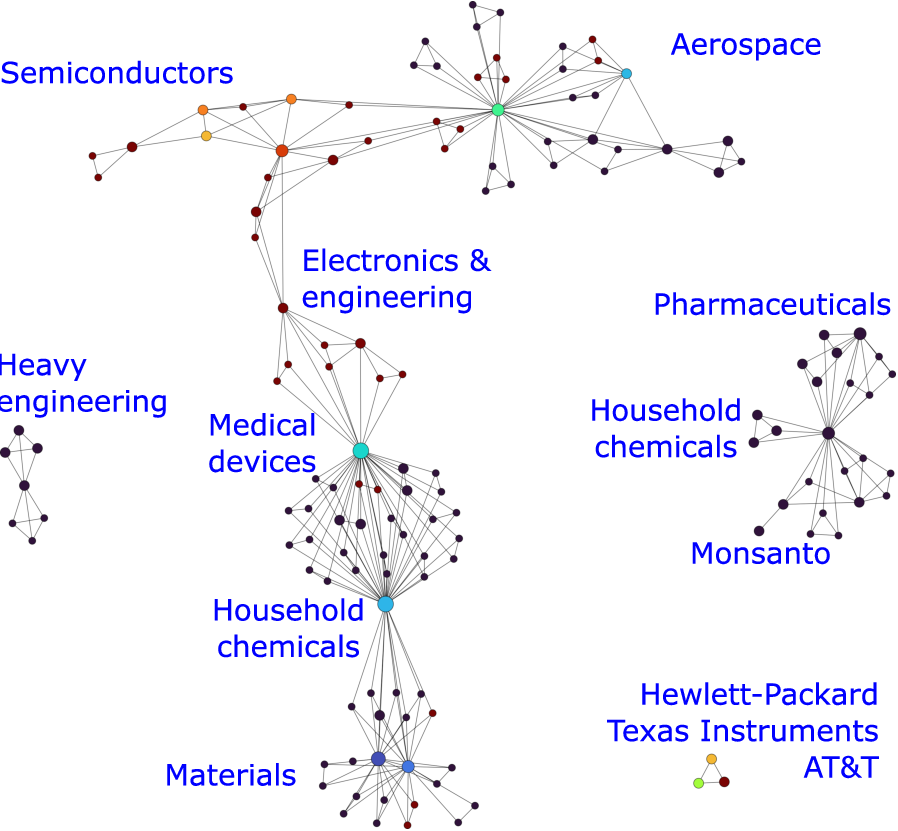}
\caption{1990 by Mapper}%
\end{subfigure}
\caption*{\footnotesize {%
\textit{Note}: Each plot shows a cross section of firms in a specific year by two-dimensional PCA or Mapper. The color scheme of the Mapper graphs highlights the S\&P "Technology" sector as in Figure \ref{Figure - color maps} (f).}}%
\label{Figure - year-by-year 1}
\end{figure}%

\begin{figure}[htb!!!!]
\caption{Year-by-Year Plots by PCA and Mapper (2 of 2)}%

\begin{subfigure}{0.45\textwidth}
\centering
\includegraphics[width=0.85\linewidth]{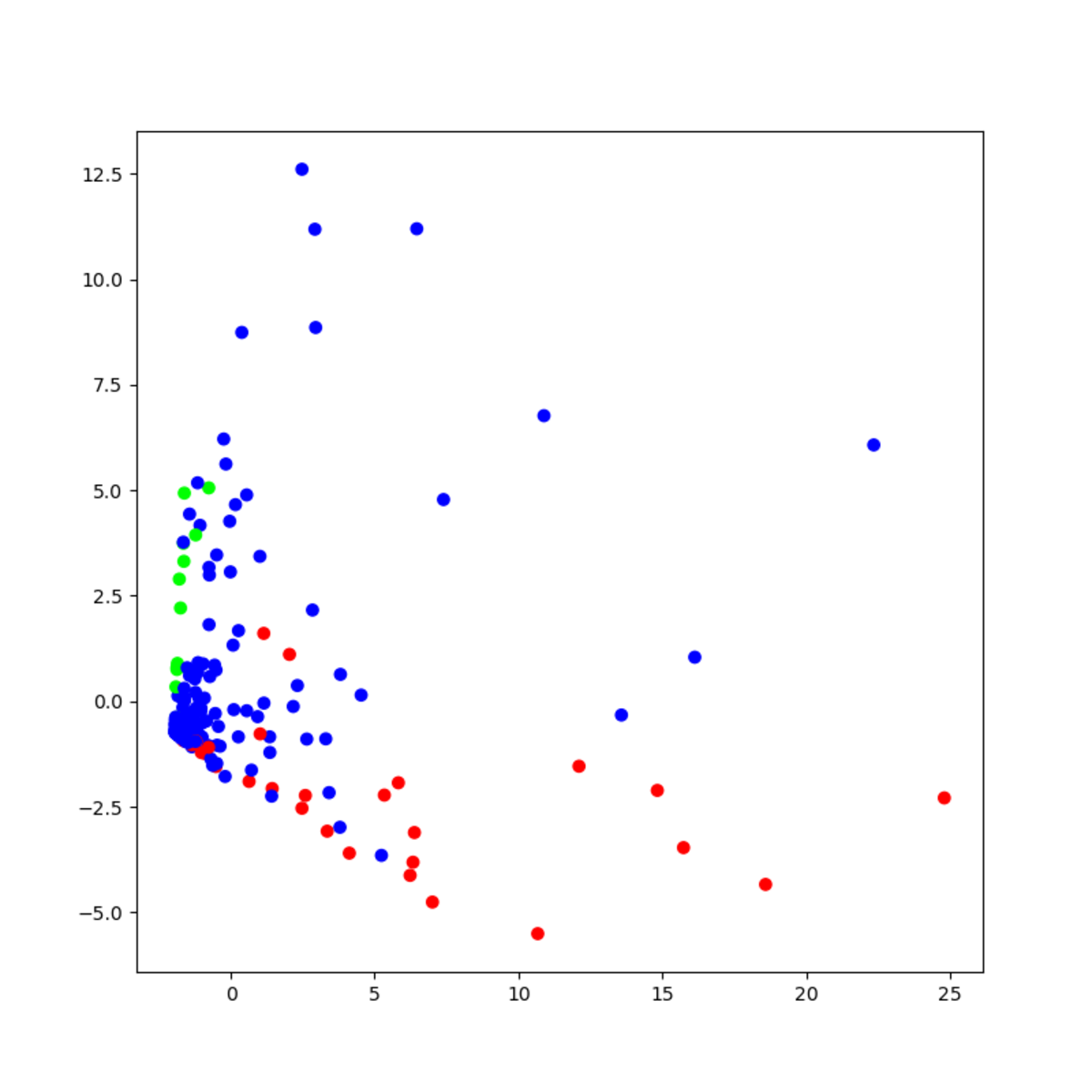}
\caption{1995 by PCA}%
\end{subfigure}
\begin{subfigure}{0.45\textwidth}
\centering
\includegraphics[width=0.80\linewidth]{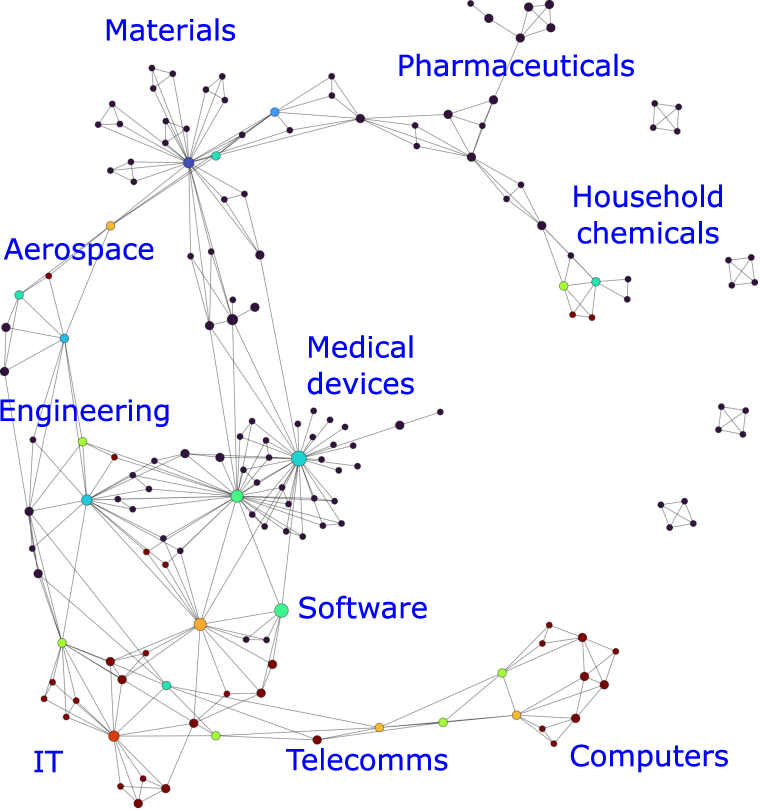}
\caption{1995 by Mapper}%
\end{subfigure}

\begin{subfigure}{0.45\textwidth}
\centering
\includegraphics[width=0.85\linewidth]{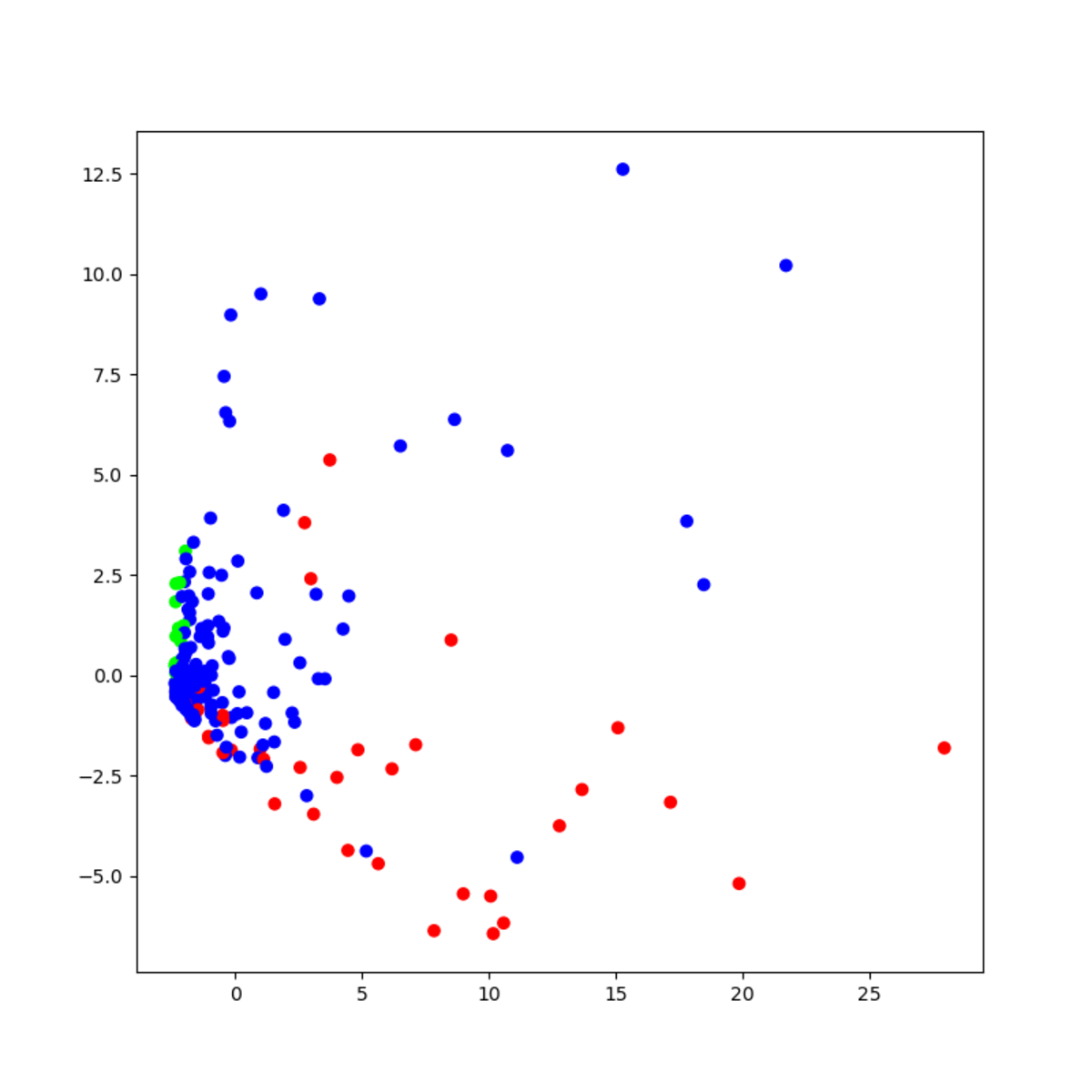}
\caption{2000 by PCA}%
\end{subfigure}
\begin{subfigure}{0.45\textwidth}
\centering
\includegraphics[width=0.75\linewidth]{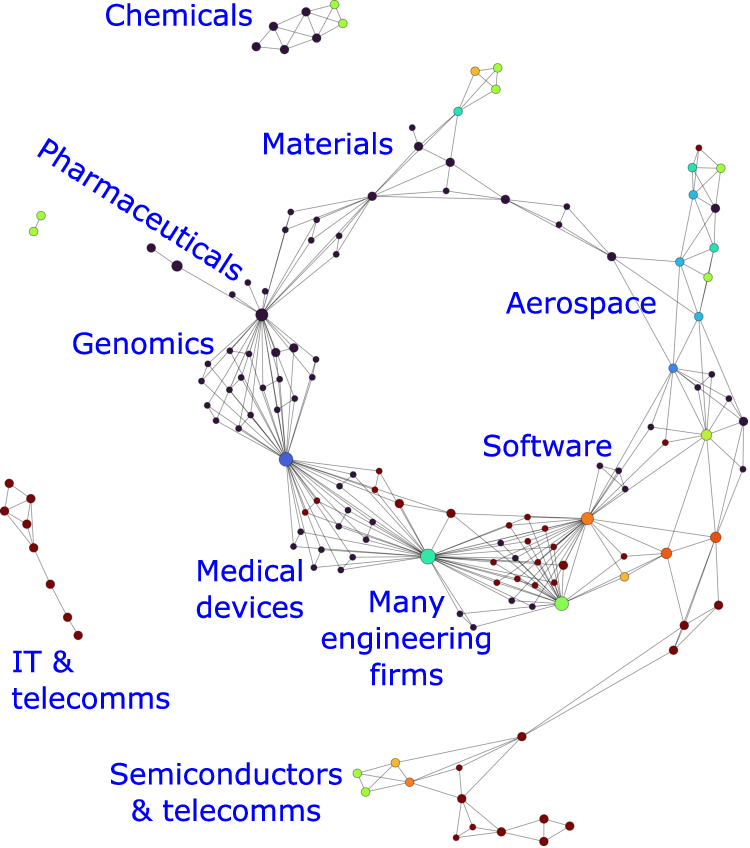}
\caption{2000 by Mapper}%
\end{subfigure}

\begin{subfigure}{0.45\textwidth}
\centering
\includegraphics[width=0.85\linewidth]{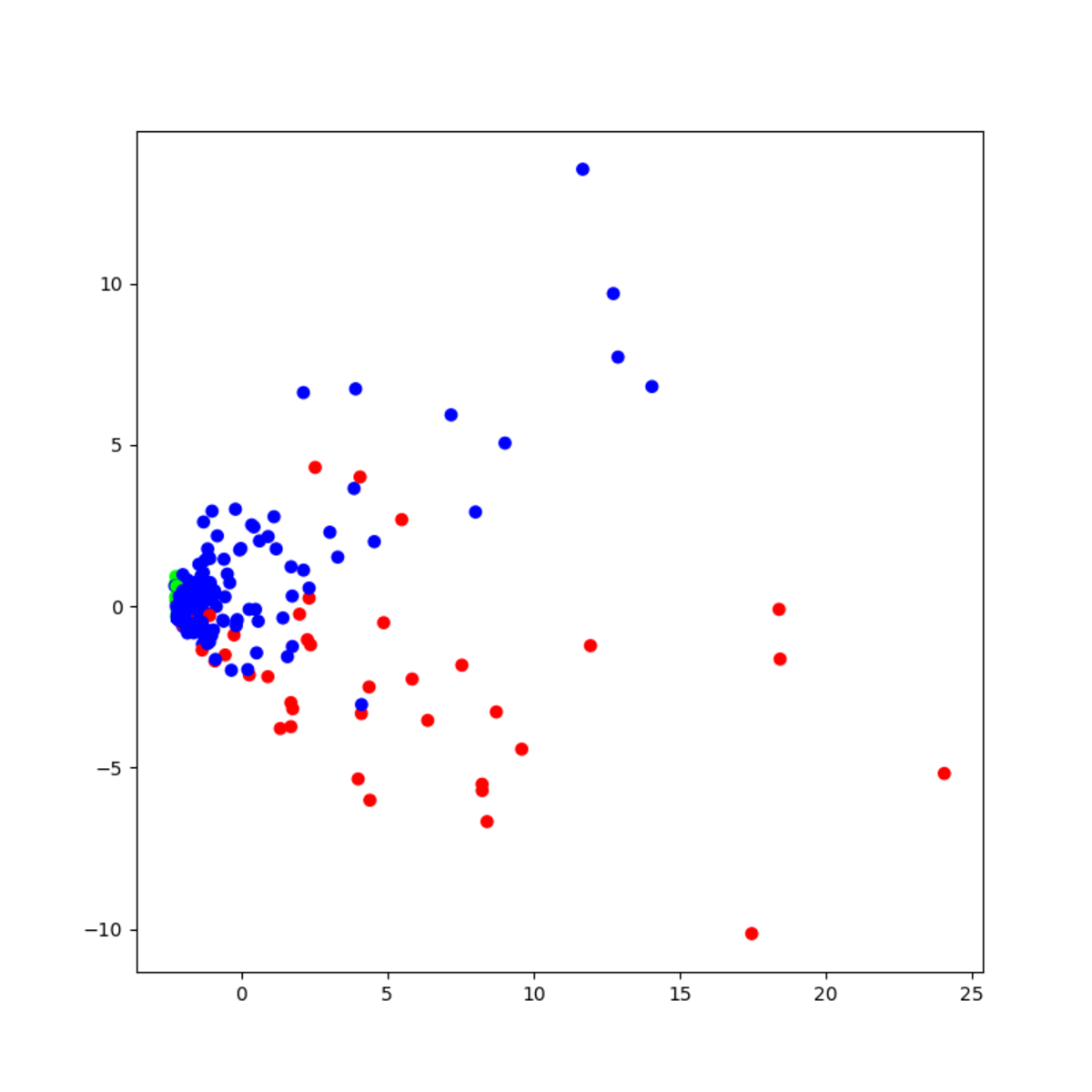}
\caption{2005 by PCA}%
\end{subfigure}
\begin{subfigure}{0.45\textwidth}
\centering
\includegraphics[width=0.75\linewidth]{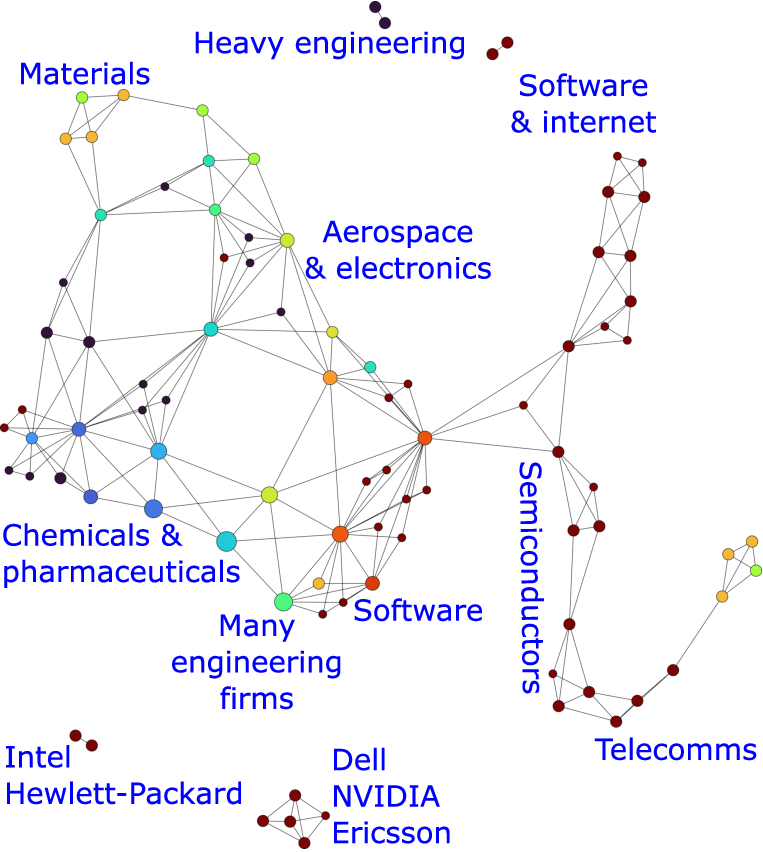}
\caption{2005 by Mapper}%
\end{subfigure}
\caption*{\footnotesize {%
\textit{Note}: Each plot shows a cross section of firms in a specific year by two-dimensional PCA or Mapper. The color scheme of the Mapper graphs highlights the S\&P "Technology" sector as in Figure \ref{Figure - color maps} (f).}}%
\label{Figure - year-by-year 2}
\end{figure}%

%\bibliographystyle{plain}
%\bibliography{refs}

\clearpage
\printbibliography

@article{lum2013extracting,
  title={Extracting insights from the shape of complex data using topology},
  author={Lum, Pek Y and Singh, Gurjeet and Lehman, Alan and Ishkanov, Tigran and Vejdemo-Johansson, Mikael and Alagappan, Muthu and Carlsson, John and Carlsson, Gunnar},
  journal={Scientific reports},
  volume={3},
  pages={1236},
  year={2013},
  publisher={Nature Publishing Group}
}

@inproceedings{singh2007topological,
  title={Topological methods for the analysis of high dimensional data sets and 3d object recognition.},
  author={Singh, Gurjeet and M{\'e}moli, Facundo and Carlsson, Gunnar},
  booktitle={SPBG},
  pages={91--100},
  year={2007}
}

@article{rizvi2017single,
  title={Single-cell topological RNA-seq analysis reveals insights into cellular differentiation and development},
  author={Rizvi, Abbas H and Camara, Pablo G and Kandror, Elena K and Roberts, Thomas J and Schieren, Ira and Maniatis, Tom and Rabadan, Raul},
  journal={Nature biotechnology},
  volume={35},
  number={6},
  pages={551},
  year={2017},
  publisher={Nature Publishing Group}
}

@article{nicolau2011topology,
  title={Topology based data analysis identifies a subgroup of breast cancers with a unique mutational profile and excellent survival},
  author={Nicolau, Monica and Levine, Arnold J and Carlsson, Gunnar},
  journal={Proceedings of the National Academy of Sciences},
  volume={108},
  number={17},
  pages={7265--7270},
  year={2011},
  publisher={National Acad Sciences}
}

@article{yao2009topological,
  title={Topological methods for exploring low-density states in biomolecular folding pathways},
  author={Yao, Yuan and Sun, Jian and Huang, Xuhui and Bowman, Gregory R and Singh, Gurjeet and Lesnick, Michael and Guibas, Leonidas J and Pande, Vijay S and Carlsson, Gunnar},
  journal={The Journal of chemical physics},
  volume={130},
  number={14},
  pages={04B614},
  year={2009},
  publisher={AIP}
}

@inproceedings{edelsbrunner2000topological,
  title={Topological persistence and simplification},
  author={Edelsbrunner, Herbert and Letscher, David and Zomorodian, Afra},
  booktitle={Proceedings 41st Annual Symposium on Foundations of Computer Science},
  pages={454--463},
  year={2000},
  organization={IEEE}
}

@MISC{KeplerMapper2019,
    author       = "Van Veen, Hendrik Jacob and Saul, Nathaniel",
    title        = "KeplerMapper",
    url  = "http://doi.org/10.5281/zenodo.1054444",
    year         = "2019"
}

@unpublished{Ozcan2015,
  title={Innovation and Acquisition: Two-Sided Matching in M\&A Markets},
  author={Yasin Ozcan},
  year={2015},
}

@article{IgamiSubrahmanyam2019,
author = "Igami, Mitsuru and Subrahmanyam, Jai",
title = "Patent Statistics as an Innovation Indicator? Evidence from the Hard Disk Drive Industry",
year = "2019",
journal = "Japanese Economic Review",
volume = "70",
number = "3",
pages = "308 - 330",
}

@article{JAFFE198987,
title = "Characterizing the ``technological position'' of firms, with application to quantifying technological opportunity and research spillovers",
journal = "Research Policy",
volume = "18",
number = "2",
pages = "87 - 97",
year = "1989",
author = "Jaffe, Adam",
}

@article{Jaffe1986,
 author = {Jaffe, Adam},
 journal = {American Economic Review},
 number = {5},
 pages = {984--1001},
  title = {Technological Opportunity and Spillovers of R\&D: Evidence from Firms' Patents, Profits, and Market Value},
 volume = {76},
 year = {1986}
}

@incollection{PakesGriliches1984,
  author = {Pakes, Ariel and Griliches, Zvi},
  title        = {Patents and R\&D at the Firm Level: A First Look},
  booktitle    = {R\&D, Patents and Productivity},
  publisher    = {University of Chicago Press},
  year         = 1984,
  editor       = {Zvi Griliches},
  address      = {Chicago, Illinois},
}

@incollection{Cohen2010, 
title = "Fifty Years of Empirical Studies of Innovative Activity and Performance",
author = "Wesley M. Cohen",
booktitle = "Handbook of the Economics of Innovation",
volume = "1",
year = "2010",
pages = "129-213",
editor = "Bronwyn H. Hall, and Nathan Rosenberg",
publisher = "Elsevier",
}

@article{BennerWaldfogel2008,
title = "Close to you? Bias and precision in patent-based measures of technological proximity",
year = "2008",
author = "Mary Benner and Joel Waldfogel",
journal = "Research Policy",
volume = "37", 
issue = "9", 
pages = "1556-1567",
}

@article{BarLeiponen2012,
title = "A measure of technological distance",
journal = "Economics Letters",
volume = "116",
number = "3",
pages = "457 - 459",
year = "2012",
author = "Talia Bar and Aija Leiponen",
}

@article{BloomVanReenenSchankerman2013,
author = {Bloom, Nicholas and Schankerman, Mark and Van Reenen, John},
title = {Identifying Technology Spillovers and Product Market Rivalry},
journal = {Econometrica},
volume = {81},
number = {4},
pages = {1347-1393},
year = {2013}
}

@article{JaffeTrajtenbergHenderson1993,
title = "Geographic Localization of Knowledge Spillovers as Evidenced by Patent Citations",
author = {Jaffe, Adam and Trajtenberg, Manuel and Henderson, Rebecca},
    title = {Geographic Localization of Knowledge Spillovers as Evidenced by Patent Citations},
    journal = {Quarterly Journal of Economics},
    volume = {108},
    number = {3},
    pages = {577-598},
    year = {1993},
}

@report{EC2017,
 author = "{European Commission}",
   title = {CASE M.7932 - Dow/DuPont},
 year = {2017}
}

@incollection{Nelson1962,
  author = {Nelson, Richard},
  title        = {Introduction},
  booktitle    = {The Rate and Direction of Inventive Activity: Economic and Social Factors},
    publisher = "Princeton University Press",
    address   = "Princeton, NJ",
  year         = "1962",
  editor       = {Universities-National Bureau Committee for Economic Research, Committee on Economic Growth of the Social Science Research Council},
}

@incollection{LernerStern2012,
  author = {Lerner, Josh and Stern, Scott},
  title        = {Introduction},
  booktitle    = {The Rate and Direction of Inventive Activity Revisited},
    publisher = "University of Chicago Press",
    address   = "Chicago, IL",
  year         = 2012,
  editor       = {Lerner, Josh and Stern, Scott},
}

@article{AzoulayEtAl2019,
Author = {Azoulay, Pierre and Fons-Rosen, Christian and Graff Zivin, Joshua S.},
Title = {Does Science Advance One Funeral at a Time?},
Journal = {American Economic Review},
Volume = {109},
Number = {8},
Year = {2019},
Month = {8},
Pages = {2889-2920},
}

@article{Myers2020,
Author = {Myers, Kyle},
Title = {The Elasticity of Science},
Journal = {American Economic Journal: Applied Economics},
Volume = {},
Number = {},
Year = {2020},
Month = {},
Pages = {},
}

@article{EpsteinCarlssonEdelsbrunner2011,
	year = 2011,
	publisher = {{IOP} Publishing},
	volume = {27},
	number = {12},
	pages = {120201},
	title = {Topological data analysis},
	journal = {Inverse Problems},
	author = {Charles Epstein and Gunnar Carlsson and Herbert Edelsbrunner}
}

@article{SizemoreEtAl2018,
author = {Sizemore, Ann E. and Phillips-Cremins, Jennifer E. and Ghrist, Robert and Bassett, Danielle S.},
title = {The importance of the whole: Topological data analysis for the network neuroscientist},
journal = {Network Neuroscience},
volume = {3},
number = {3},
pages = {656-673},
year = {2019},
}

@article{zomorodian2005computing,
  title={Computing persistent homology},
  author={Zomorodian, Afra and Carlsson, Gunnar},
  journal={Discrete \& Computational Geometry},
  volume={33},
  number={2},
  pages={249--274},
  year={2005},
  publisher={Springer}
}

@article{carlsson2009topology,
  title={Topology and data},
  author={Carlsson, Gunnar},
  journal={Bulletin of the American Mathematical Society},
  volume={46},
  number={2},
  pages={255--308},
  year={2009}
}

@book{edelsbrunner2010computational,
  title={Computational topology: an introduction},
  author={Edelsbrunner, Herbert and Harer, John},
  year={2010},
  publisher={American Mathematical Society}
}

@article{chazal2017introduction,
  title={An introduction to Topological Data Analysis: fundamental and practical aspects for data scientists},
  author={Chazal, Fr{\'e}d{\'e}ric and Michel, Bertrand},
  journal={arXiv preprint arXiv:1710.04019},
  year={2017}
}

@article{saggar2018towards,
  title={Towards a new approach to reveal dynamical organization of the brain using topological data analysis},
  author={Saggar, Manish and Sporns, Olaf and Gonzalez-Castillo, Javier and Bandettini, Peter A and Carlsson, Gunnar and Glover, Gary and Reiss, Allan L},
  journal={Nature communications},
  volume={9},
  number={1},
  pages={1--14},
  year={2018},
  publisher={Nature Publishing Group}
}

@article{hall2005market,
  title={Market value and patent citations},
  author={Hall, Bronwyn and Jaffe, Adam and Trajtenberg, Manuel},
  journal={RAND Journal of economics},
  pages={16--38},
  year={2005}
}

@article{hiraoka2016hierarchical,
  title={Hierarchical structures of amorphous solids characterized by persistent homology},
  author={Hiraoka, Yasuaki and Nakamura, Takenobu and Hirata, Akihiko and Escolar, Emerson G and Matsue, Kaname and Nishiura, Yasumasa},
  journal={Proceedings of the National Academy of Sciences of the United States of America},
  volume={113},
  number={26},
  pages={7035--7040},
  year={2016}
}

@article{BLP1995,
  title={Automobile prices in market equilibrium},
  author={Berry, Steven and Levinsohn, James and Pakes, Ariel},
  journal={Econometrica},
  pages={841--890},
  year={1995}
}

@article{EP1995,
  title={Markov-perfect industry dynamics: A framework for empirical work},
  author={Ericson, Richard and Pakes, Ariel},
  journal={Review of Economic Studies},
  volume={62},
  number={1},
  pages={53--82},
  year={1995}
}

@book{grove1996,
  title={Only the paranoid survive: How to exploit the crisis points that challenge every company and career},
  author={Grove, Andrew S},
  year={1996},
  publisher={Currency}
}

@book{gerstner2002,
  title={Who Says Elephants Can't Dance},
  author={Gerstner, Louis V.},
  year={2002},
  publisher={Harper Business}
}

% \begin{thebibliography}{9}
% \bibitem{} xxxx. 20xx. \textquotedblleft xxxx.\textquotedblright \  \textit{%
% Xxxx}, x (x): xxxx--xxxx.
% \end{thebibliography}

\end{document}